\documentclass[12pt]{article}
\usepackage{enumerate}
\usepackage{url} 

\newcommand{\blind}{0}

\usepackage[font={small}]{caption}
\DeclareMathSizes{10}{9}{6}{5}
\usepackage{subcaption}

\RequirePackage{amsthm,amsmath,amsfonts,amssymb}
\usepackage[authoryear, round, sort&compress]{natbib}
\usepackage{xcolor}
\usepackage{xr-hyper}

\usepackage[pdftex,bookmarks,colorlinks,breaklinks]{hyperref}
\definecolor{dullmagenta}{rgb}{0.4,0,0.4} 
\definecolor{darkblue}{rgb}{0,0,0.4}
\definecolor{coquelicot}{rgb}{0.20, 0.12, 0.72}
\definecolor{navyblue}{rgb}{0,0,0.5}
\hypersetup{linkcolor=blue,citecolor=blue,filecolor=blue,urlcolor=blue} 
\usepackage{tabularx}

\def\independenT#1#2{\mathrel{\rlap{$#1#2$}\mkern2mu{#1#2}}}

\usepackage{epstopdf}

\usepackage{multirow}
\usepackage{mathtools}
\usepackage{bbm}
\usepackage{graphicx} 
\usepackage{flafter} 
\usepackage{bm} 

\usepackage{colortbl}
\usepackage{arydshln}
\usepackage{booktabs}

\usepackage{algorithm}
\usepackage{algpseudocode}

\usepackage{pifont}

\usepackage{tikz}
\usetikzlibrary{positioning}
\usetikzlibrary{arrows.meta,arrows}

\definecolor{darkblue}{rgb}{0,0,0.4}
\definecolor{coquelicot}{rgb}{0.90, 0.42, 0.72}
\definecolor{burntorange}{rgb}{0.8, 0.33, 0.0}

\def\bx{\mathbf{x}}
\def\bX{\mathbf{X}}

\def\bZ{\mathbf{Z}}

\def\bU{\mathbf{U}}

\def\bV{\mathbf{V}}

\def\bI{\mathbf{I}}
\def\bM{\mathbf{M}}
\def\bQ{\mathbf{Q}}

\def\balpha{\boldsymbol{\alpha}}

\def\bbeta{\boldsymbol{\beta}}

\def\bSigma{\boldsymbol{\Sigma}}

\def\fhat{\widehat{f}}

\newcommand\ind{\protect\mathpalette{\protect\independenT}{\perp}}
\def\independenT#1#2{\mathrel{\rlap{$#1#2$}\mkern4mu{#1#2}}}

\def\bS{\mathbf{S}}
\def\bs{\mathbf{s}}

\def\ghat{\widehat{g}}

\def\bDelta{\boldsymbol{\Delta}}

\def\bW{\mathbf{W}}

\def\bgamma{\boldsymbol{\gamma}}
\def\bdelta{\boldsymbol{\delta}}

\def\R{\mathbb{R}}

\def\bSigma{\boldsymbol{\Sigma}}

\def \hs2{\hspace{2mm}}

\numberwithin{table}{section}
\numberwithin{equation}{section}

\definecolor{jcolor}{RGB}{041,122,000}
\definecolor{darkred}{RGB}{100,000,000}
\definecolor{purple}{RGB}{200,000,200}

\def\boxit#1{\vbox{\hrule\hbox{\vrule\kern6pt  \vbox{\kern6pt#1\kern6pt}\kern6pt\vrule}\hrule}}

\def\bS{\mathbf{S}}

\theoremstyle{plain}
\newtheorem{theorem}{Theorem}[section]

\newtheorem{lemma}[theorem]{Lemma}

\theoremstyle{remark}
\newtheorem{assumption}{Assumption}

\renewcommand{\thetable}{\arabic{table}}

\begin{document}

	\def\spacingset#1{\renewcommand{\baselinestretch}%
		{#1}\small\normalsize} \spacingset{1}


	\date{}

	\if0\blind
	{
		\title{\bf Balancing utility and cost in dynamic treatment regimes}
		\author{Kai Chen \hspace{.2cm}\\
			Institute of Statistics and Big Data, Renmin University of China \\
			\href{mailto:kaichen@ruc.edu.cn}{kaichen@ruc.edu.cn}\\
			and \\
			Yuqian Zhang \hspace{.2cm}\\
			Institute of Statistics and Big Data, Renmin University of China\\
			\href{mailto:yuqianzhang@ruc.edu.cn}{yuqianzhang@ruc.edu.cn}
		}
		\maketitle
		\vspace{-2em}
	} \fi

	\if1\blind
	{
		\title{\bf Balancing utility and cost in dynamic treatment regimes}
		\maketitle
	} \fi

	\bigskip
	\begin{abstract}
	Dynamic treatment regimes (DTRs) are personalized, adaptive strategies designed to guide the sequential allocation of treatments based on individual characteristics over time. Before each treatment assignment, covariate information is collected to refine treatment decisions and enhance their effectiveness. The more information we gather, the more precise our decisions can be. However, this also leads to higher costs during the data collection phase. In this work, we propose a balanced Q-learning method that strikes a balance between the utility of the DTR and the costs associated with both treatment assignment and covariate assessment. The performance of the proposed method is demonstrated through extensive numerical studies, including simulations and a real-data application to the MIMIC-III database.
	\end{abstract}
	
	\noindent%
	\vfill
	
	\newpage
	\spacingset{1.9} 
	
	\section{Introduction}\label{sec:introduction}
	Dynamic treatment regimes (DTRs) represent a class of statistical methods designed to guide individualized treatment decisions across multiple stages. In contrast to static strategies that maintain the same treatment regardless of how circumstances evolve, DTRs allow treatment decisions to be updated over time based on the changing state of the individual. At each decision point, context-specific information, including prior treatment response, relevant measurements, and other observed characteristics, is used to construct a decision rule that prescribes the most appropriate treatment. For instance, in precision medicine, DTRs are applied to personalize the management of chronic diseases such as diabetes and cancer \citep{kosorok2019precision,komorowski2018artificial,yu2021reinforcement}. In these applications, treatment decisions are adjusted based on real-time biomarker information to modify medication dosages or treatment protocols according to individual patient needs. 
	
	The traditional objective of DTRs is to maximize a utility function through a sequence of adaptive decision rules defined by an individual's evolving history. Toward this end, a range of effective methods have been developed for estimating optimal DTRs. These include regression-based approaches such as Q-learning \citep{zhu2019proper, clifton2020q}, A-learning \citep{murphy2003optimal, shi2018high}, doubly robust methods \citep{ertefaie2021robust, zhang2012robust}, and classification-based approaches such as outcome-weighted learning \citep{zhao2015new} and augmented outcome-weighted learning \citep{liu2018augmented}.
	
	However, a major limitation of the aforementioned methods in practical applications is their exclusive focus on maximizing utility while overlooking the associated costs. A prominent category of cost stems from treatment assignments, which has received considerable attention in the existing literature. Notable examples include \cite{kitagawa2018should,athey2021policy,xu2024optimal}, which consider known non-random cost functions, and \cite{liu2024controlling,laber2014set,wang2018learning,lizotte2012linear,luckett2021estimation,huang2020estimating}, which model the cost or risk as a random variable generated by an unknown process. Nevertheless, these studies primarily focus on the costs of treatment assignments, while the costs associated with covariate assessment have been relatively underexplored.
	
	In the decision-making process, there is a continual demand to collect more and higher-quality covariates to enable more precise and personalized guidance. However, this goal is often constrained by the high costs associated with covariate collection. In precision medicine, gene expression profiling has been demonstrated to improve diagnosis and treatment outcomes for severe diseases. For instance, \cite{wright2023genomic} reports high diagnostic accuracy for pediatric rare diseases using genome-wide data, while \cite{jamal2017tracking} identifies chromosomal instability as a prognostic indicator in lung cancer. In addition, neuroimaging techniques, such as MRI and fMRI, play an important role in disease diagnosis and management \citep{bakker2019supplemental,gorgec2024mri}. Despite their clinical value, these tools are often prohibitively expensive and financially inaccessible in routine practice. Therefore, it is important to develop methods that offer guidance on when and which covariate information should be collected to achieve an effective balance between utility and cost.
	
	Within the field of optimal DTR learning, by enforcing sparsity to exclude redundant variables, regularizations in Q-learning and A-learning \citep{zhu2019proper,shi2018high} are helpful in reducing the costs associated with covariate assessment. However, these methods focus exclusively on the contribution of covariates to improving utility and do not incorporate cost information into the learning process. As a result, covariates that are clinically expensive to collect may still be retained, even if their contribution to utility improvement is marginal.
		
	Another  related line of research focuses on covariate monitoring \citep{neugebauer2017identification,caniglia2016monitor,kreif2021exploiting}. In longitudinal studies, the frequency of covariate monitoring plays an important role in shaping treatment decisions. Increasing the monitoring frequency generally allows for more precise treatment adjustments and potentially improved clinical outcomes, but also leads to higher associated costs. The aforementioned studies focus on evaluating the DTRs under pre-specified monitoring strategies with a given frequency. These approaches are suitable for settings with a small number of candidate strategies, where it is feasible to assess each option individually and choose the one that achieves the best utility-cost trade-off. However, when the set of candidate strategies is large or even infinite, evaluating each strategy individually becomes impractical or computationally intractable. In such situations, directly optimizing over the entire space of strategies offers a more practical solution. Methods such as Q-learning, which directly learn and optimize treatment strategies, are particularly well-suited to these scenarios. Furthermore, while existing work provides guidance on when to collect covariate information by treating all covariates as a single group, it does not address the question of which specific covariates should be collected. In this work, we aim to develop a framework that offers individualized guidance on which covariates should be assessed at different stages, based on the evolving status of each individual.

\paragraph*{Organization}

	The remainder of this paper is structured as follows. Section \ref{sec:formulation} introduces the problem formulation and notation. Section \ref{sec:BQL} describes the proposed Balanced Q-learning (BQL) methodology. Section \ref{sec:theory} establishes the theoretical guarantees for BQL. Section \ref{sec:numerical} evaluates its performance through numerical studies. Discussions are provided in Section \ref{sec:dis}. All technical proofs and supplementary results are provided in the supplementary material.

	\section{Problem Formulation}\label{sec:formulation}
	
	We introduce the necessary notation and the problem formulation used in this paper.

\paragraph*{Notation}

	For any integer $k>0$, define $[k]=\{1,2,\dots,k\}$. Given a vector $\bx\in\mathbb{R}^d$ and a subset $L\subseteq [d]$, let $\bx_L$ denote the subvector of $\bx$ consisting of the components indexed by $L$. For any event $A$, the indicator function $I(A)$ equals $1$ if $A$ occurs and $0$ otherwise. The cardinality of set $B$ is denoted as $|B|$. For a symmetric matrix $\bSigma\in\mathbb{R}^{d\times d}$, $\lambda_{\rm min}(\bSigma)$ and $\lambda_{\rm max}(\bSigma)$ represent its smallest and largest eigenvalues, respectively. We denote the $d\times d$ identity matrix by $\bI_d$, and use $\mathbf{1}_d$ and $\mathbf{0}_d$ to represent the d-dimensional vectors of all ones and all zeros, respectively. For a function $f(\cdot)$, let $\|f\|_{P,2}=\{\int f(x)^2 dP(x)\}^{1/2}$ to denote the $L_2(P)$ norm under the probability measure $P$.

\paragraph*{Observed data} For simplicity of presentation, we focus on two-stage trials with binary treatment assignments, although the proposed theoretical framework and methodology can be readily extended to multi-stage trials with multiple treatment options. The observed data consist of $n$ independent and identically distributed samples, $(\bW_i)_{i=1}^n=(\bS_{1i}, A_{1i}, \bS_{2i}, A_{2i}, Y_i)$. Here, $A_{1i}$ and $A_{2i}$ are binary treatment indicators for the first and second stages, respectively; $\bS_{1i} \in \mathbb{R}^{d_1}$ and $\bS_{2i} \in \mathbb{R}^{d_2}$ represent the pre-treatment covariates at the two stages, which may differ in dimension; and $Y_i \in \mathbb{R}$ denotes the outcome of interest, with higher values indicating more favorable outcomes. We adopt the potential outcome framework. For any $a_1, a_2 \in \{0,1\}$, let $\bS_{2i}(a_1)$ denote the potential value of of the second-stage covariate had treatment $a_1$ been assigned at the first stage, and let $Y_i(a_1, a_2)$ denote the potential outcome corresponding to the treatment sequence $(a_1, a_2)$. For each individual, only the potential values $\bS_{2i} = \bS_{2i}(A_{1i})$ and $Y_i = Y_i(A_{1i}, A_{2i})$ are observable.
	
\paragraph*{Test data} Based on the observed samples $(\bW_i)_{i=1}^n$, our goal is to assign an optimal treatment rule to new individuals from an independent test set, aiming to achieve a balance between utility and associated costs. Throughout the paper, we use the subscript $i$ to denote variables corresponding to observed samples, while variables without the subscript $i$ refer to a new individual from the test population. For each new individual, we posit the existence of hypothetical variables $\bW = (\bS_1, A_1, \bS_2, A_2, Y)$, although this full set of information may not be fully collected due to cost considerations. The available information is collected sequentially according to the following procedure.

	When a new subject arrives, we assume that a subset of baseline covariates, such as demographic information, denoted by $\bS_{l_1} \subseteq \bS_1$ with a known index set $l_1 \subseteq [d_1]$, can be collected at negligible cost. Based on the observed baseline covariates $\bS_{l_1}$, the first decision involves determining which additional covariates should be collected before assigning the first-stage treatment. For example, in precision medicine, this decision may involve assessing whether it is worthwhile to conduct additional laboratory tests, gene expression profiling, neuroimaging scans, or other diagnostic procedures to obtain more comprehensive information on the patient's health status. For a decision rule $\pi^{1c}(\cdot): \R^{|l_1|} \mapsto \mathcal{J}_1$, we denote $J_1 = \pi^{1c}(\bS_{l_1}) \in \mathcal{J}_1$ as the (random) index set of additional covariates to be collected, where $\mathcal{J}_1$ is a user-specified collection of candidate subsets of $[d_1] \setminus l_1$. After selecting the index set $J_1$, we collect the corresponding covariates $\bS_{J_1} = \bS_{\pi^{1c}(\bS_{l_1})}$ and define $\bS_{\bar{J}_1} := (\bS_{l_1}, \bS_{J_1})$ as the full set of covariates assessed prior to the first treatment decision. Consider a treatment rule $\pi^{1t}(\cdot,\cdot)$, where for any $j_1\in\mathcal J_1$, $\pi^{1t}(\cdot,j_1):\R^{|l_1|+|j_1|}\mapsto\{0,1\}$, the first-stage treatment is assigned as $A_1= \pi^{1t}(\bS_{\bar{J}_1},J_1) = \pi^{1t}(\bS_{l_1}, \bS_{\pi^{1c}(\bS_{l_1})},\pi^{1c}(\bS_{l_1}))\in\{0,1\}$.
	
	Similarly, prior to the second-stage treatment, a baseline subset of covariates $\bS_{l_2}$, corresponding to a known index set $l_2 \subseteq [d_2]$, is observed at no cost. Let $\bS_{\bar{L}_2} := (\bS_{\bar{J}_1}, \bS_{l_2})$ represent the history of all covariates collected up to the current stage and $\mathcal{J}_2$ denote the collection of candidate index sets, where each element is a subset of $[d_2] \backslash l_2$. For a decision rule $\pi^{2c}(\cdot, \cdot, \cdot)$, where $\pi^{2c}(\cdot, j_1, a_1):\R^{|l_1|+|j_1|+|l_2|}\mapsto\mathcal J_2$ for any $j_1\in\mathcal J_1$ and $a_1\in\{0,1\}$, we define $J_2=\pi^{2c}(\bS_{\bar{L}_2}, J_1, A_1)\in\mathcal{J}_2$ as the (random) index set of additional covariates to be collected prior to the second-stage treatment. The corresponding additional covariates are denoted by $\bS_{J_2}$. We define $\bS_{\bar{J}_2} := (\bS_{\bar{L}_2}, \bS_{J_2})$ as the full set of covariates available prior to the second-stage treatment. Lastly, consider a treatment rule $\pi^{2t}(\cdot, \cdot, \cdot, \cdot)$, where $\pi^{2t}(\cdot, j_1, a_1, j_2):\R^{|l_1|+|j_1|+|l_2|+|j_2|}\mapsto\{0,1\}$, the second-stage treatment is then assigned as $A_2= \pi^{2t}(\bS_{\bar{J}_2}, J_1, A_1, J_2) \in \{0, 1\}$.

\paragraph*{The utility and cost} 

For a given regime $(\pi^{1c}, \pi^{1t}, \pi^{2c}, \pi^{2t})$, we define its utility as the expected outcome obtained by following this regime, that is, $\text{Utility} = E\{Y(A_1, A_2)\}$. The utility depends on all components of the regime $(\pi^{1c}, \pi^{1t}, \pi^{2c}, \pi^{2t})$ because the assigned treatments, $A_1 = \pi^{1t}(\bS_{\bar{J}_1}, J_1)$ and $A_2 = \pi^{2t}(\bS_{\bar{J}_2}, J_1, A_1, J_2)$, are functions of both the treatment assignment rules $(\pi^{1t}, \pi^{2t})$ and the index sets $(J_1, J_2)$. Meanwhile, these index sets, $J_1 = \pi^{1c}(\bS_{l_1})$ and $J_2 = \pi^{2c}(\bS_{\bar{L}_2}, J_1, A_1)$, are determined by the covariate assessment rules $(\pi^{1c}, \pi^{2c})$. Therefore, to be precise, we should write $J_1 = J_1(\pi^{1c})$, $A_1 = A_1(\pi^{1c}, \pi^{1t})$, $J_2 = J_2(\pi^{1c}, \pi^{1t}, \pi^{2c})$, and $A_2 = A_2(\pi^{1c}, \pi^{1t}, \pi^{2c}, \pi^{2t})$. For simplicity of notation, we omit the explicit dependence of $(J_1, A_1, J_2, A_2)$ on the regime.

In this formulation, we implicitly assume that the covariate assessment decisions $(J_1, J_2)$ may influence the final outcome $Y=Y(A_1, A_2)$ only indirectly, through their impact on the treatment assignments $(A_1, A_2)$, and do not have a direct effect on the outcome. In practice, this implies that although performing additional diagnostic procedures, such as laboratory tests, may contribute to more precise treatment assignments $(A_1, A_2)$ and thereby improve the health outcome $Y(A_1, A_2)$, we assume that the act of conducting these additional diagnostic procedures does not directly influence the health status itself.

	Moreover, we consider user-specified cost functions $C^{1c}(\cdot)$, $C^{1t}(\cdot)$, $C^{2c}(\cdot)$, and $C^{2t}(\cdot)$. Specifically, for any $a_1,a_2\in\{0,1\}$, $C^{1t}(a_1)$ and $C^{2t}(a_2)$ represent the costs associated with allocating $a_1$ and $a_2$ in the first and second stages, respectively. Similarly, $C^{1c}(j_1)$ and $C^{2c}(j_2)$ denote the costs of acquiring covariates corresponding to index sets $j_1 \in \mathcal{J}_1$ and $j_2 \in \mathcal{J}_2$, respectively. The profit is then defined as the utility minus all the expected cost:
	\begin{align}\label{def:profit}
	{\rm Profit}(\pi^{1c},\pi^{1t},\pi^{2c},\pi^{2t}):=E\{Y(A_1,A_2)-C^{1c}(J_1)-C^{1t}(A_1)-C^{2c}(J_2)-C^{2t}(A_2)\}.
	\end{align}
	
	In this paper, our objective is to learn the optimal regime that maximizes the overall profit. Handling the costs associated with treatment assignments, $C^{1t}(A_1)$ and $C^{2t}(A_2)$, is relatively straightforward. This can be achieved by incorporating constant thresholds $C^{1t}(1) - C^{1t}(0)$ and $C^{2t}(1) - C^{2t}(0)$ into the decision boundaries for the first- and second-stage treatment decisions, respectively \citep{athey2021policy}. In contrast, incorporating the costs arising from covariate assessments, $C^{1c}(J_1)$ and $C^{2c}(J_2)$, presents a more challenging problem. Due to these cost considerations, it is often infeasible to collect complete covariate information, rendering standard approaches for optimal DTR estimation under fully observed data inapplicable. 
	
	The central challenge is to determine the optimal size and composition of the index sets $(J_1, J_2)$ used for covariate assessment. Larger index sets allow for the collection of more information, potentially resulting in more accurate treatment decisions and improved outcomes. However, collecting additional covariates generally leads to increased costs. Balancing this trade-off is a key focus of our approach. Furthermore, it is important to identify and retain ``important and inexpensive'' covariates in the assessment process, while excluding ``non-informative and expensive'' covariates to enhance both cost efficiency and decision quality.
	
	To identify the optimal DTR, we adopt the following standard assumptions on the \emph{observed data} \citep{murphy2003optimal,schulte2015q}.
	
	\begin{assumption}[Identification conditions]\label{ass:identification}
		For any $a=(a_1,a_2)\in\{0,1\}^2$, let the following conditions hold: (a) (Consistency) $Y_i=Y_i(A_{1i},A_{2i})$ and $\bS_{2i}=\bS_{2i}(A_{1i})$; (b) (Sequential ignorability) $A_{1i}\ind \{\bS_{2i}(a_1),Y_i(a)\}\mid \bS_{1i}$ and $A_{2i}\ind Y_i(a)\mid (\bS_{1i},A_{1i}=a_1,\bS_{2i})$; (c) (Positivity) $P(A_{1i}=a_1\mid\bS_{1i})>0$ and $P(A_{2i}=a_2\mid\bS_{1i},A_{1i}=a_1,\bS_{2i})>0$ almost surely.
	\end{assumption}
	
	To enable the estimation of the optimal DTR for the test population using the observed data, we assume that the following conditional distributions remain invariant between the observed and test data: \(P_{\bS_{1}}\), \(P_{\bS_{2} \mid (\bS_{1}, A_{1})}\), and \(P_{Y \mid (\bS_{1}, A_{1}, \bS_{2}, A_{2})}\). In particular, we do not assume that individuals in the observed data followed the optimal regime we seek to estimate for the test population. Additionally, we assume full access to all covariates in the observed data and concentrate exclusively on the cost considerations arising in the test population.

	To simplify the presentation, we assume $j_1^{\mathrm{f}} := [d_1] \backslash l_1 \in \mathcal{J}_1$ and $j_2^{\mathrm{f}} := [d_2] \backslash l_2 \in \mathcal{J}_2$, indicating that collecting all covariates is included as one possible option in the covariate assessment strategies. This assumption is not essential and can be omitted. Under this construction, we have $\bS_1 = (\bS_{l_1}, \bS_{j_1^{\mathrm{f}}})$ and $\bS_2 = (\bS_{l_2}, \bS_{j_2^{\mathrm{f}}})$. Additionally, we define $\bS_{\bar{l}_2^{\,\mathrm{f}}} := (\bS_1, \bS_{l_2})$ and $\bar{\bS}_2 := (\bS_1, \bS_2)$.

	\section{Balanced Q-learning}\label{sec:BQL}

	In the following, we present the identification results for the optimal DTR that maximizes the overall profit at the population level in Section \ref{sec:optimal}. In Section \ref{sec:algorithm}, we introduce our Balanced Q-learning (BQL) method for estimating the optimal DTR.

	\subsection{The Optimal Regime}\label{sec:optimal}
	
	We first define the Q-functions corresponding to the four key decisions in the considered setup: the first- and second-stage covariate assessments and the first- and second-stage treatment assignments. This construction extends the classical framework \citep{clifton2020q} to accommodate scenarios where covariate assessments are integrated into the decision-making process. The Q-functions are formulated in a backward fashion, beginning from the second-stage treatment assignment and sequentially moving backward through the decision stages.
	
	In standard settings where covariate assessment costs are not considered, full access to covariate information is available prior to treatment assignments. Under such scenarios, the second-stage Q-function is typically defined as follows: for any $\bar{\bs}_2 \in \mathbb{R}^{d_1 + d_2}$ and $a_1, a_2 \in \{0,1\}$, 
$$
\bar{Q}^{2t}(\bar{\bs}_2, a_1, a_2) = E\{ Y - C^{2t}(a_2) \mid \bar{\bS}_2 = \bar{\bs}_2, A_1 = a_1, A_2 = a_2 \},
$$
which quantifies the expected profit of assigning treatment $a_2$ given the current state.

Notably, although the construction of $\bar{Q}^{2t}$ relies on the full covariate vector $\bar{\bS}_2$, it remains estimable using the \emph{observed data}, where complete covariate information is available. However, for new individuals in the \emph{test population}, treatment decisions cannot be made by directly applying the standard approach of maximizing $\bar{Q}^{2t}(\bar{\bs}_2, a_1, a_2)$ over $a_2 \in \{0,1\}$, because this procedure requires access to the entire covariate vector $\bar{\bS}_2$, which is often unavailable due to cost considerations. When covariate assessment decisions are incorporated into the decision-making process, the historical information accessible at each stage is constrained by the preceding decisions on which covariates to collect. As a result, treatment assignments must be determined based solely on the subset of covariates that have been collected up to the current decision point, rather than the full covariate history.

For any given index sets $j_1 \in \mathcal{J}_1$ and $j_2 \in \mathcal{J}_2$, we denote $\bS_{\bar{j}_2} := (\bS_{l_1}, \bS_{j_1}, \bS_{l_2}, \bS_{j_2})$ as the set of historical covariates available prior to the second-stage treatment decision. We then define the second-stage Q-function under restricted covariate access as: for any $\bs_{\bar{j}_2}\in\R^{|l_1|+|j_1|+|l_2|+|j_2|}$, $a_1, a_2 \in \{0,1\}$, $j_1\in\mathcal J_1$, and $j_2\in\mathcal J_2$,
\begin{align*}
Q^{2t}(\bs_{\bar{j}_2}, j_1, a_1, j_2, a_2) = E\left\{\bar{Q}^{2t}(\bar{\bS}_2, a_1, a_2) \mid \bS_{\bar{j}_2} = \bs_{\bar{j}_2} \right\},
\end{align*}
which represents the expected profit associated with assigning treatment $a_2$ at the second stage, given access to the historical covariates $\bS_{\bar{j}_2}$ and the first-stage treatment $A_1=a_1$. Accordingly, we define the rule that selects the treatment maximizing the Q-function as
$\check{\pi}^{2t}(\bs_{\bar{j}_2}, j_1, a_1, j_2) = \mathop{\arg\max}_{a_2 \in \{0,1\}} Q^{2t}(\bs_{\bar{j}_2}, j_1, a_1, j_2, a_2).$

	For the second-stage covariate assessment, we first define the following unrestricted Q-function based on the rule $\check{\pi}^{2t}$: for any $\bs_{\bar{l}_2^{\,\mathrm{f}}}\in\R^{d_1+|l_2|}$, $a_1 \in \{0,1\}$, $j_1\in\mathcal J_1$, and $j_2\in\mathcal J_2$,
\begin{align}\label{def:Qbar_2c}
\bar{Q}^{2c}(\bs_{\bar{l}_2^{\,\mathrm{f}}}, j_1, a_1, j_2) = E\left\{\bar{Q}^{2t}\left(\bar{\bS}_2, a_1, \check{\pi}^{2t}(\bS_{\bar{j}_2}, j_1, a_1, j_2)\right) - C^{2c}(j_2) \mid \bS_{\bar{l}_2^{\,\mathrm{f}}} = \bs_{\bar{l}_2^{\,\mathrm{f}}}, A_1 = a_1 \right\}.
\end{align}
	This Q-function represents the expected profit under the following hypothetical scenario: (a) the second-stage treatment will be assigned according to the optimal rule $\check{\pi}^{2t}$, which maximizes the profit based on the subset of covariates $\bS_{\bar{j}_2}$; and (b) the current expectation is evaluated assuming full access to the historical covariate $\bS_{\bar{l}_2^{\,\mathrm{f}}} = (\bS_1, \bS_{l_2}) = (\bS_{l_1}, \bS_{j_1^\mathrm{f}}, \bS_{l_2})$ up to the point of the second covariate assessment. However, if $j_1 \in \mathcal{J}_1$ denotes the index set of covariates assessed before the first treatment, the available information prior to the second-stage covariate assessment is restricted to the subset $\bS_{\bar{l}_2} := (\bS_{l_1}, \bS_{j_1}, \bS_{l_2})$. Therefore, we define the restricted Q-function as: for any $\bs_{\bar{l}_2}\in\R^{|l_1|+|j_1|+|l_2|}$, $a_1 \in \{0,1\}$, $j_1\in\mathcal J_1$, and $j_2\in\mathcal J_2$,
\begin{align*}
Q^{2c}(\bs_{\bar{l}_2}, j_1, a_1, j_2) = E\left\{\bar{Q}^{2c}(\bS_{\bar{l}_2^{\,\mathrm{f}}}, j_1, a_1, j_2) \mid \bS_{\bar{l}_2} = \bs_{\bar{l}_2}\right\},
\end{align*}
and specify the corresponding decision rule by
$\check{\pi}^{2c}(\bs_{\bar{l}_2}, j_1, a_1) = \mathop{\arg\max}_{j_2 \in \mathcal{J}_2} Q^{2c}(\bs_{\bar{l}_2}, j_1, a_1, j_2).$

	Similarly, when assigning the first-stage treatment, we define the following Q-functions based on the rule $\check{\pi}^{2c}$: for any $\bs_1\in\R^{d_1}$, $\bs_{\bar{j}_1}\in\R^{|l_1|+|j_1|}$, $a_1 \in \{0,1\}$, and $j_1\in\mathcal J_1$,
\begin{align}
\bar{Q}^{1t}(\bs_{1}, j_1, a_1) &= E\left\{\bar{Q}^{2c}\left(\bS_{\bar{l}_2^{\,\mathrm{f}}}, j_1, a_1, \check{\pi}^{2c}(\bS_{\bar{l}_2}, j_1, a_1)\right) - C^{1t}(a_1) \mid \bS_{1} = \bs_{1}, A_1 = a_1\right\},\label{def:Q1t}\\
Q^{1t}(\bs_{\bar{j}_1}, j_1, a_1) &= E\left\{\bar{Q}^{1t}(\bS_{1}, j_1, a_1) \mid \bS_{\bar{j}_1} = \bs_{\bar{j}_1}\right\}.\nonumber
\end{align}
The first Q-function, $\bar{Q}^{1t}$, represents the expected profit when subsequent decisions are made using the (restricted) optimal rules $(\check{\pi}^{2c},\check{\pi}^{2t})$, with the expectation evaluated assuming full access to the historical covariate information $\bS_1 = (\bS_{l_1}, \bS_{j_1^{\mathrm{f}}})$ up to the point of the first treatment assignment. The second Q-function, $Q^{1t}$, denotes the expected profit using only the collected covariates $\bS_{\bar{j}_1} := (\bS_{l_1}, \bS_{j_1})$. The optimal first-stage treatment decision rule is then defined by
$\check{\pi}^{1t}(\bs_{\bar{j}_1},j_1) = \mathop{\arg\max}_{a_1 \in \{0,1\}} Q^{1t}(\bs_{\bar{j}_1}, j_1, a_1).$
	
	Lastly, for the first-stage covariate assessment, we define the Q-function as: for any $\bs_{l_1}\in\R^{|l_1|}$ and $j_1\in\mathcal J_1$,
\begin{align}\label{def:Q1c}
Q^{1c}(\bs_{l_1}, j_1) = E\left\{ \bar{Q}^{1t}\left(\bS_1, j_1, \check{\pi}^{1t}(\bS_{\bar{j}_1},j_1)\right) - C^{1c}(j_1) \mid \bS_{l_1} = \bs_{l_1} \right\},
\end{align}
which represents the expected profit obtained by assessing the covariate subset indexed by $j_1$, followed by applying the treatment rules $(\check{\pi}^{1t},\check{\pi}^{2c},\check{\pi}^{2t})$. The corresponding optimal covariate assessment decision rule is given by
$\check{\pi}^{1c}(\bs_{l_1}) = \mathop{\arg\max}_{j_1 \in \mathcal{J}_1} Q^{1c}(\bs_{l_1}, j_1),$
which selects the covariate subset that maximizes the expected profit based on the baseline covariates $\bs_{l_1}$.
	
	The following theorem establishes the optimality of the constructed regime $(\check{\pi}^{1c}, \check{\pi}^{1t}, \check{\pi}^{2c}, \check{\pi}^{2t})$.

	\begin{theorem}\label{thm:optimal_DTR}
		Let Assumption \ref{ass:identification} hold. Consider any dynamic treatment regime (DTR) $(\pi^{1c},\pi^{1t},\pi^{2c},\pi^{2t})$, we have 
		${\rm Profit}(\check\pi^{1c},\check\pi^{1t},\check\pi^{2c},\check\pi^{2t})\ge {\rm Profit}(\pi^{1c},\pi^{1t},\pi^{2c},\pi^{2t}).$
	\end{theorem}

	\subsection{Algorithm}\label{sec:algorithm}
	
	After establishing the identification results for the population-level optimal DTR in Theorem \ref{thm:optimal_DTR}, we propose a Balanced Q-learning (BQL) algorithm to estimate the optimal DTR using the observed samples. The detailed construction is provided in Algorithm \ref{alg:BQL2}. In the following, we outline the key intuition underlying the construction of the algorithm.

	Theorem \ref{thm:optimal_DTR} demonstrates that the population-level optimal DTR can be identified by finding the maximizers of the previously defined Q-functions. To determine the maximizer of the Q-function at a given state, it is sufficient to select one action as the baseline and calculate the differences between the Q-functions corresponding to different actions and the chosen baseline. Specifically, for each $a_1,a_2\in\{0,1\}$, $j_1\in\mathcal J_1$, and $j_2\in\mathcal J_2$, we define the following contrast functions: 
	\begin{align*}
	\Delta^{2t}(\bs_{\bar{j}_2},j_1,a_1,j_2)&:=Q^{2t}(\bs_{\bar{j}_2},j_1,a_1,j_2,1)-Q^{2t}(\bs_{\bar{j}_2},j_1,a_1,j_2,0),\\
	\Delta^{2c}(\bs_{\bar{l}_2},j_1,a_1,j_2)&:=Q^{2c}(\bs_{\bar{l}_2},j_1,a_1,j_2)-Q^{2c}(\bs_{\bar{l}_2},j_1,a_1,j_2^{\rm f}),\\
	\Delta^{1t}(\bs_{\bar{j}_1},j_1)&:=Q^{1t}(\bs_{\bar{j}_1},j_1,1)-Q^{1t}(\bs_{\bar{j}_1},j_1,0),\\
	\Delta^{1c}(\bs_{l_1},j_1)&:=Q^{1c}(\bs_{l_1},j_1)-Q^{1c}(\bs_{l_1},j_{1}^{\rm f}).
	\end{align*}
	
	\begin{algorithm}[htbp] \caption{Balanced Q-learning (BQL)}\label{alg:BQL2}
		\linespread{1}\selectfont
		\begin{algorithmic}[1]
			\Require Observations $(\bW_i)_{i=1}^n=(\bS_{1i},A_{1i},\bS_{2i},A_{2i},Y_i)_{i=1}^n$, with user-specified cost functions.
			\State Split the data into $K$ folds ($K \geq 2$), indexed by $\mathcal{I}_k$. Define $\mathcal{D}_{-k} = (\bW_i)_{i \in [n] \setminus \mathcal{I}_k}$.
			\State Regress $Y_i \sim (\bar{\bS}_{2i}, A_{1i})$ and $A_{2i} \sim (\bar{\bS}_{2i}, A_{1i})$ using $\mathcal{D}_{-k}$ to obtain $\widehat{f}_2^{-k}(\cdot)$ and $\widehat{g}_2^{-k}(\cdot)$.
			\State Define $\hat r_{f_2,i}^{-k}:=Y_{i}-\fhat_{2}^{-k}(\bar{\bS}_{2i},A_{1i})$, $\hat r_{g_2,i}^{-k}:=A_{2i}-\ghat_{2}^{-k}(\bar{\bS}_{2i},A_{1i})$, and compute
			\begin{align*}
				\widetilde{\balpha}&=\mathop{\arg\min}_{\balpha}\sum_{k\in[K]}\sum_{i\in\mathcal{I}_k}\left[\hat r_{f_2,i}^{-k}-\hat r_{g_2,i}^{-k}\left\{(\bar{\bS}_{2i},A_{1i})^\top\balpha+ C^{2t}(1) - C^{2t}(0)\right\}\right]^2,\\
				\widehat{\balpha}_{j_1 a_1 j_2} &= \mathop{\arg\min}_{\balpha} \sum_{i \in [n]} \left\{ (\bar{\bS}_{2i}, a_1)^\top \widetilde{\balpha} - \bS_{\bar{j}_2 i}^\top \balpha \right\}^2\;\;\text{for each}\;\;j_1 \in \mathcal{J}_1,\;\;a_1 \in \{0,1\},\;\;j_2 \in \mathcal{J}_2.
			\end{align*}
			\State Define pseudo-outcomes
			\begin{align}\label{def:pseudo}
	\widehat{Y}_{j_1j_2i}^{2c}=(\bar{\bS}_{2i},A_{1i})^\top\widehat{\balpha}\left\{I\left(\bS_{\bar{j}_2i}^\top\widehat{\balpha}_{j_1A_{1i}j_2}> 0\right)-I\left(\bS_{\bar{j}_2^{\rm f}i}^\top\widehat{\balpha}_{j_1A_{1i}j_2^{\rm f}}> 0\right)\right\}-C^{2c}(j_2)+C^{2c}(j_2^\mathrm{f}).
	\end{align}
			 \State For each $a_1\in\{0,1\}$, $j_1\in\mathcal J_1$, and $j_2\in\mathcal J_2$, compute
			\begin{align*}
			\widetilde{\bbeta}_{j_1j_2}&=\mathop{\arg\min}_{\bbeta}\sum_{i\in[n]}\left\{\widehat{Y}_{j_1j_2i}^{2c}-(\bS_{\bar l_2^{\,\rm f}i},A_{1i})^\top\bbeta\right\}^2,\\
			\widehat{\bbeta}_{j_1a_1j_2}&=\mathop{\arg\min}_{\bbeta}\sum_{i\in[n]}\left\{(\bS_{\bar l_2^{\,\rm f}i},a_1)^\top\widetilde{\bbeta}_{j_1j_2}-\bS_{\bar l_2i}^\top\bbeta\right\}^2.
			\end{align*}
			\State Define pseudo-outcomes
			\begin{align}
				\widehat{Y}_{j_1i}^{1t}&=Y_i-C^{2t}(A_{2i})+(\bar{\bS}_{2i},A_{1i})^\top\widehat{\balpha}\left\{I\left(\bS_{\bar{j}_2^{\rm f}i}^\top\widehat{\balpha}_{j_1A_{1i}j_2^{\rm f}}> 0\right)-A_{2i}\right\}-C^{2c}(j_2^{\rm f})\nonumber\\
				&\quad+\sum\nolimits_{j_2\in\mathcal{J}_2}I\left(j_2={\arg\max}_{j_2^\prime\in\mathcal{J}_2}\bS_{\bar{l}_2i}^\top\widehat{\bbeta}_{j_1 A_{1i} j_2^\prime}\right)(\bS_{\bar{l}_2^{\, \rm f}i},A_{1i})^\top\widetilde{\bbeta}_{j_1j_2}-C^{1t}(A_{1i}).\label{def:Y_J1}
			\end{align}
			\State Repeat Steps 3-7 and construct a cross-fitted version of \eqref{def:Y_J1} using only samples in $\mathcal{D}_{-k}$ to obtain $\widehat{Y}_{j_1i}^{1t,-k}$. Regress $\widehat{Y}_{j_1i}^{1t,-k} \sim \bS_{1i}$ and $A_{1i} \sim \bS_{1i}$ using $\mathcal{D}_{-k}$ to obtain $\widehat{f}_{j_1}^{-k}(\cdot)$ and $\widehat{g}_1^{-k}(\cdot)$.
			\State Define $\hat r_{f_1,i}^{-k}:=\widehat{Y}_{j_1i}^{1t}-\fhat_{j_1}^{-k}(\bS_{1i})$, $\hat r_{g_1,i}^{-k}:=A_{1i}-\ghat_{1}^{-k}(\bS_{1i})$. For each $j_1\in\mathcal J_1$, compute
			\begin{align*}
				\widetilde{\bgamma}_{j_1}&=\mathop{\arg\min}_{\bgamma}\sum_{k\in[K]}\sum_{i\in\mathcal{I}_k}(\hat r_{f_1,i}^{-k}-\hat r_{g_1,i}^{-k}\bS_{1i}^\top\bgamma)^2,\;\; \widehat{\bgamma}_{j_1}=\mathop{\arg\min}_{\bgamma}\sum_{i\in[n]}\Big(\bS_{1i}^\top\widetilde{\bgamma}_{j_1}-\bS_{\bar{j}_1i}^\top\bgamma\Big)^2.
			\end{align*}
			\State Construct pseudo-outcomes $\widehat{Y}_{j_1i}^{1c}=\widehat{Y}_{j_1i}+\bS_{1i}^\top\widetilde{\bgamma}_{j_1}\{I(\bS_{\bar{j}_1}^\top\widehat{\bgamma}_{j_1}>0)-A_{1i}\}-C^{1c}(j_1)$ and compute
			\begin{align*}
				\widehat{\bdelta}_{j_1}&=\mathop{\arg\min}_{\bdelta}\sum_{i\in[n]}\left(\widehat{Y}_{j_1i}^{1c}-\widehat{Y}_{j_1^{\rm f}i}^{1c}-\bS_{l_1i}^\top\bdelta\right)^2\;\;\text{for each}\;\;j_1\in\mathcal J_1.
			\end{align*}
			\Return $(\widehat{\balpha}_{j_1a_1j_2})_{j_1\in\mathcal J_1,j_2\in\mathcal J_2,a_1\in\{0,1\}}$, $(\widehat{\bbeta}_{j_1a_1j_2})_{j_1\in\mathcal J_1,j_2\in\mathcal J_2,a_1\in\{0,1\}}$, $(\widehat{\bgamma}_{j_1})_{j_1\in\mathcal J_1}$, and $(\widehat{\bdelta}_{j_1})_{j_1\in\mathcal J_1}$.
		\end{algorithmic}
	\end{algorithm}

	These contrast functions quantify the relative benefits of different actions compared to the baseline. The optimal DTR can be equivalently defined as
	\begin{align*}
	\check{\pi}^{2t}(\bs_{\bar{j}_2}, j_1, a_1, j_2) &= I\{\Delta^{2t}(\bs_{\bar{j}_2},j_1,a_1,j_2)>0\},&\quad\check{\pi}^{2c}(\bs_{\bar{l}_2}, j_1, a_1) &= \mathop{\arg\max}_{j_2 \in \mathcal{J}_2} \Delta^{2c}(\bs_{\bar{l}_2}, j_1, a_1, j_2),\\
	\check{\pi}^{1t}(\bs_{\bar{j}_1},j_1) &= I\{\Delta^{1t}(\bs_{\bar{j}_1}, j_1, a_1)>0\},&\quad\check{\pi}^{1c}(\bs_{l_1}) &= \mathop{\arg\max}_{j_1 \in \mathcal{J}_1} \Delta^{1c}(\bs_{l_1}, j_1).
	\end{align*}
	To estimate the optimal regime, it suffices to estimate the above contrast functions. Modeling the contrast functions rather than the Q-functions directly offers the advantage of reducing the number of models that must be specified at each stage \citep{shi2018high,ertefaie2021robust} -- since the contrast function corresponding to the baseline action is always zero by definition, it does not require separate estimation.

	Moreover, since the restricted Q-functions $(Q^{2t}, Q^{2c}, Q^{1t})$ are defined as the conditional expectations of the unrestricted Q-functions $(\bar Q^{2t}, \bar Q^{2c}, \bar Q^{1t})$, it is natural to introduce the corresponding contrast functions for the unrestricted Q-functions. These contrast functions capture the differences in the unrestricted Q-functions relative to the chosen baseline actions and serve as intermediate quantities that facilitate the estimation of the restricted contrast functions. Specifically, define $\bar{\Delta}^{2t}(\bar{\bs}_2,a_1):=\bar{Q}^{2t}(\bar{\bs}_2,a_1,1)-\bar{Q}^{2t}(\bar{\bs}_2,a_1,0)$, $\bar{\Delta}^{2c}(\bs_{\bar{l}_2^{\,\rm f}},j_1,a_1,j_2):=\bar{Q}^{2c}(\bs_{\bar{l}_2^{\,\rm f}},j_1,a_1,j_2)-\bar{Q}^{2c}(\bs_{\bar{l}_2^{\,\rm f}},j_1,a_1,j_2^{\rm f})$, and $\bar{\Delta}^{1t}(\bs_{1},j_1):=\bar{Q}^{1t}(\bs_{1},j_1,1)-\bar{Q}^{1t}(\bs_{1},j_1,0)$. Then,
	\begin{align*}
	\Delta^{2t}(\bs_{\bar{j}_2},j_1,a_1,j_2)&=E\left\{\bar{\Delta}^{2t}(\bar{\bS}_2,a_1)\mid \bS_{\bar{j}_2} = \bs_{\bar{j}_2} \right\},\\
	\Delta^{2c}(\bs_{\bar{l}_2}, j_1, a_1, j_2)&=E\left\{\bar{\Delta}^{2c}(\bS_{\bar{l}_2^{\,\rm f}},j_1,a_1,j_2) \mid \bS_{\bar{l}_2} = \bs_{\bar{l}_2}\right\},\\
	\Delta^{1t}(\bs_{\bar{j}_1}, j_1, a_1)&=E\left\{\bar{\Delta}^{1t}(\bS_{1},j_1) \mid \bS_{\bar{j}_1} = \bs_{\bar{j}_1}\right\}.
	\end{align*}
	Therefore, once we have obtained estimates for the functions $(\bar{\Delta}^{2t}, \bar{\Delta}^{2c}, \bar{\Delta}^{1t})$, we can apply nested regression on the collected historical covariates to derive estimates for the restricted contrast functions $(\Delta^{2t}, \Delta^{2c}, \Delta^{1t})$. The contrast function corresponding to the Q-function at the first covariate assessment, $\Delta^{1c}$, does not require nested regression since it is defined based on the baseline variables $\bS_{l_1}$, which are always directly observable.

	In the following, we introduce the estimation procedures for the contrast functions $\bar{\Delta}^{2t}$, $\bar{\Delta}^{2c}$, $\bar{\Delta}^{1t}$, and $\Delta^{1c}$ in a sequential manner. We first focus on the second-stage treatment assignment. Define the nuisance functions $f_2(\bar{\bs}_{2}, a_1) = E(Y \mid \bar{\bS}_{2} = \bar{\bs}_{2}, A_1 = a_1)$ and $g_2(\bar{\bs}_{2}, a_1) = E(A_2 \mid \bar{\bS}_{2} = \bar{\bs}_{2}, A_1 = a_1)$. Using these definitions, we have the following representation:
\[
Y - f_2(\bar{\bS}_2, A_1) = \{A_2 - g_2(\bar{\bS}_2, A_1)\} \left\{\bar{\Delta}^{2t}(\bar{\bS}_2, A_1) + C^{2t}(1) - C^{2t}(0)\right\} +\varepsilon_{2t},
\]
with some $\varepsilon_{2t}\in\R$ satisfying $E(\varepsilon_{2t}\mid\bar{\bS}_2,A_1)=0$ and $E[\{A_2 - g_2(\bar{\bS}_2, A_1)\}\varepsilon_{2t}\mid\bar{\bS}_2,A_1]=0$. This relationship implies that, after obtaining estimates for the nuisance functions $(f_2, g_2)$, we can substitute these estimates and apply an R-learner approach \citep{robinson1988root,nie2021quasi,ertefaie2021robust} to estimate the contrast function $\bar{\Delta}^{2t}$. This procedure is detailed in Steps 2-3 of Algorithm \ref{alg:BQL2}, where a linear estimator $(\bar{\bs}_{2}, a_1)^\top \widetilde{\balpha}$ is used for the contrast function $\bar{\Delta}^{2t}(\bar{\bs}_2, a_1)$. Following this, we apply a nested regression to obtain a linear estimator $\bs_{\bar{j}_2}^\top \widehat{\balpha}_{j_1 a_1  j_2}$ for the restricted contrast function $\Delta^{2t}(\bs_{\bar{j}_2}, j_1, a_1, j_2)$, as outlined earlier.
	
	We now return to the second-stage covariate assessment and note the representation:
\begin{align*}
&\bar{Q}^{2t}\left(\bar{\bS}_2, A_1, \check{\pi}^{2t}(\bS_{\bar{j}_2}, j_1, A_1, j_2)\right) - C^{2c}(j_2) \\
&\qquad = \bar{Q}^{2t}(\bar{\bS}_2, A_1, 0) + \bar{\Delta}^{2t}(\bar{\bS}_2, A_1) I\left\{\Delta^{2t}(\bS_{\bar{j}_2}, j_1, A_1, j_2) > 0\right\} - C^{2c}(j_2),
\end{align*}
for any $j_1 \in \mathcal{J}_1$ and $j_2 \in \mathcal{J}_2$. Together with \eqref{def:Qbar_2c}, we obtain
\begin{align*}
&\bar{\Delta}^{2t}(\bar{\bS}_2, A_1) \left[ I\left\{\Delta^{2t}(\bS_{\bar{j}_2}, j_1, A_1, j_2) > 0 \right\} - I\left\{\Delta^{2t}(\bS_{\bar{j}_2^{\rm f}}, j_1, A_1, j_2^{\rm f}) > 0 \right\} \right] \\
&\qquad = \bar{\Delta}^{2c}(\bS_{\bar{l}_2^{\,\rm f}}, j_1, A_1, j_2) + C^{2c}(j_2) - C^{2c}(j_2^{\rm f}) + \varepsilon_{2c},
\end{align*}
where $\varepsilon_{2c} \in \mathbb{R}$ satisfies $E(\varepsilon_{2c} \mid \bS_{\bar{l}_2^{\,\mathrm{f}}}, A_1) = 0$. Based on this result, we define pseudo-outcomes as in \eqref{def:pseudo} by substituting the previously obtained estimates for the contrast functions $\bar{\Delta}^{2t}$ and $\Delta^{2t}$. We then conduct a nested regression to obtain a linear estimator $(\bs_{\bar{l}_2^{\,\rm f}}, a_1)^\top \widetilde{\bbeta}_{j_1 j_2}$ for $\bar{\Delta}^{2c}(\bs_{\bar{l}_2^{\,\rm f}}, j_1, a_1, j_2)$, followed by an additional nested regression to derive a linear estimator $\bs_{\bar{l}_2}^\top \widehat{\bbeta}_{j_1 a_1 j_2}$ for the restricted contrast function $\Delta^{2c}(\bs_{\bar{l}_2}, j_1, a_1, j_2)$; see Steps 4-5 of Algorithm \ref{alg:BQL2}.

	For the first-stage treatment assignment, we consider the oracle pseudo-outcome
	\begin{align*}
	\widetilde{Y}_{j_1}^{1t}&:=\bar{Q}^{2c}(\bS_{\bar{l}_2^{\,\mathrm{f}}}, j_1, A_1, \check{\pi}^{2c}(\bS_{\bar{l}_2}, j_1, A_1))-C^{1t}(A_1)\\
	&=E\left[\bar{Q}^{2t}\left(\bar{\bS}_2, A_1, 0\right)+\bar{\Delta}^{2t}(\bar{\bS}_2, A_1)I\left\{\Delta^{2t}(\bS_{\bar{J}_2^{\rm f}}, j_1, A_1, j_2^{\rm f})>0\right\} \mid \bS_{\bar{l}_2^{\,\mathrm{f}}}, A_1\right]- C^{2c}(j_2^{\rm f})\\
	&\qquad+\sum\nolimits_{j_2\in\mathcal{J}_2}I\left\{\check{\pi}^{2c}(\bS_{\bar{l}_2}, j_1, A_1)=j_2\right\}\bar{\Delta}^{2c}(\bS_{\bar{l}_2^{\,\rm f}},j_1,A_1,j_2)-C^{1t}(A_1).
	\end{align*}
By definition in \eqref{def:Q1t}, $\bar{Q}^{1t}(\bS_{1}, j_1, A_1) = E(\widetilde{Y}_{j_1}^{1t} \mid \bS_1, A_1)$. Moreover, we observe that $E\{Y-C^{2t}(A_2)-A_2\bar{\Delta}^{2t}(\bar{\bS}_2,A_1)-\bar{Q}^{2t}(\bar{\bS}_2,A_1,0)\mid \bar{\bS}_2,A_1\}=0$. This motivates the definition of
	\begin{align}
	Y_{j_1}^{1t}&:=Y-C^{2t}(A_2)+\bar{\Delta}^{2t}(\bar{\bS}_2, A_1)\{I(\Delta^{2t}(\bS_{\bar{j}_2^{\rm f}}, j_1, A_1, j_2^{\rm f})>0)-A_2\} - C^{2c}(j_2^{\rm f})\nonumber\\
	&\qquad+\sum\nolimits_{j_2\in\mathcal{J}_2}I\{\check{\pi}^{2c}(\bS_{\bar{l}_2}, j_1, A_1)=j_2\}\bar{\Delta}^{2c}(\bS_{\bar{l}_2^{\,\rm f}},j_1,A_1,j_2)-C^{1t}(A_1).\label{def:Y1t}
	\end{align}
It then follows that $\bar{Q}^{1t}(\bS_{1}, j_1, A_1) = E(Y_{j_1}^{1t} \mid \bS_1, A_1)$. Define the nuisance functions $f_{j_1}(\bs_1) = E( Y_{j_1}^{1t} \mid \bS_1 = \bs_1 )$ and $g_1(\bs_1) = E( A_1 \mid \bS_1 = \bs_1 )$. We have the following relationship:
\[
Y_{j_1}^{1t} - f_{j_1}(\bS_1) = \left\{ A_1 - g_1(\bS_1) \right\} \bar{\Delta}^{1t}(\bS_1,j_1) + \varepsilon_{1t},
\]
where $\varepsilon_{1t} \in \mathbb{R}$ satisfies $E\left( \varepsilon_{1t} \mid \bS_1 \right) = 0$ and $E[ \{ A_1 - g_1(\bS_1) \} \varepsilon_{1t} \mid \bS_1 ] = 0$. This formulation allows us to apply the R-learner framework again. Building on the constructed estimates of the contrast functions $\bar{\Delta}^{2t}$, $\Delta^{2t}$, $\bar{\Delta}^{2c}$, and $\Delta^{2c}$, we define the estimated pseudo-outcomes for $Y_{j_1}^{1t}$ as in \eqref{def:Y_J1}. After estimating the nuisance functions $(f_{j_1}, g_1)$, we apply the R-learner via a residual-on-residual regression to obtain a linear estimator $\bs_1^\top \widetilde{\bgamma}_{j_1}$ for $\bar{\Delta}^{1t}(\bs_1, j_1)$. The restricted contrast function $\Delta^{1t}(\bs_{\bar{j}_1}, j_1)$ is then estimated using a linear model $\bs_{\bar{j}_1}^\top \widehat{\bgamma}_{j_1}$ through a nested regression procedure. For detailed implementation steps, see Steps 6-8 of Algorithm \ref{alg:BQL2}.

	Finally, for the first-stage covariate assessment, we consider the oracle pseudo-outcome
	\begin{align*}
	\widetilde{Y}_{j_1}^{1c}&:=\bar{Q}^{1t}(\bS_1,j_1,\check{\pi}^{1t}(\bS_{\bar{j}_1},j_1))-C^{1c}(j_1)\\
	&=f_{j_1}(\bS_1)+\bar{\Delta}^{1t}\{\bS_1,j_1\}\left[I\left\{\Delta^{1t}(\bS_{\bar{j}_1},j_1)>0\right\}-g_{1}(\bS_1)\right]-C^{1c}(j_1).
	\end{align*}
According to \eqref{def:Q1c}, it holds that $Q^{1c}(\bS_{l_1}, j_1) = E( \widetilde{Y}_{j_1}^{1c} \mid \bS_{l_1} )$. We introduce 
$Y_{j_1}^{1c}:=Y^{1t}_{j_1}+\bar{\Delta}^{1t}(\bS_1,j_1)[I\{\Delta^{1t}(\bS_{\bar{j}_1},j_1)>0\}-A_1]-C^{1c}(j_1),$
which satisfies $E( \widetilde{Y}_{j_1}^{1c} - Y_{j_1}^{1c} \mid \bS_{l_1} ) = 0$, and consequently $Q^{1c}(\bS_{l_1}, j_1) = E( Y_{j_1}^{1c} \mid \bS_{l_1} )$. As a result, we have
$\Delta^{1c}(\bS_{l_1},j_1) = Y_{j_1}^{1c} - Y_{j_1^{\rm f}}^{1c} + \varepsilon_{1c},$
where $\varepsilon_{1c}\in\R$ satisfies $E\left( \varepsilon^{1c} \mid \bS_{l_1} \right) = 0$. This motivates the use of a nested regression to estimate $\Delta^{1c}(\bs_{l_1},j_1)$; see Step 9 of Algorithm \ref{alg:BQL2} where we introduce a linear estimator $\bs_{l_1}^\top\widehat{\bdelta}_{j_1}$.

	\begin{figure}[t]
		\vspace{1em}
		\centering
		\begin{tikzpicture}
			\path
			(0,0) node[rectangle,draw](A) {\scriptsize Assess covariates $\bS_{l_1}$}
			(0,-1) node[rectangle,draw](B) {\scriptsize Find the optimal assessment set $ J_1=\hat{\pi}^{1c}(\bS_{l_1}) := \arg\max_{j_1\in\mathcal J_1}\bS_{l_1}^\top\widehat{\bdelta}_{j_1}$}
			(0,-2) node[rectangle,draw](C) {\scriptsize Assess covariates $\bS_{J_1}$}
			(0,-3) node[rectangle,draw](D) {\scriptsize Assign $A_1=\hat{\pi}^{1t}(\bS_{\bar{J}_1},J_1) := I\left(\bS_{\bar J_1}^\top\widehat{\bgamma}_{J_1}>0\right)$}
			(0,-4) node[rectangle,draw](E) {\scriptsize Assess covariates $\bS_{l_2}$}
			(0,-5) node[rectangle,draw](F) {\scriptsize Find the optimal assessment set $J_2=\hat{\pi}^{2c}(\bS_{\bar{L}_2}, J_1, A_1) := \arg\max_{j_2\in\mathcal J_2}\bS_{\bar{L}_2}^\top\widehat{\bbeta}_{J_1A_1j_2}$}
			(0,-6) node[rectangle,draw](G) {\scriptsize Assess covariates $\bS_{J_2}$}
			(0,-7) node[rectangle,draw](H) {\scriptsize Assign $A_2=\hat{\pi}^{2t}(\bS_{\bar{J}_2}, J_1, A_1, J_2) := I\left(\bS_{\bar{J}_2}^\top\widehat{\balpha}_{J_1A_1J_2}>0\right)$}
			(0,-8) node[rectangle,draw](I) {\scriptsize Obtain the outcome $Y=Y(A_1,A_2)$}
			;
			\draw[-{Stealth[length=2mm, width=1mm]}] (A) -- (B);
			\draw[-{Stealth[length=2mm, width=1mm]}] (B) -- (C);
			\draw[-{Stealth[length=2mm, width=1mm]}] (C) -- (D);
			\draw[-{Stealth[length=2mm, width=1mm]}] (D) -- (E);
			\draw[-{Stealth[length=2mm, width=1mm]}] (E) -- (F);
			\draw[-{Stealth[length=2mm, width=1mm]}] (F) -- (G);
			\draw[-{Stealth[length=2mm, width=1mm]}] (G) -- (H);
			\draw[-{Stealth[length=2mm, width=1mm]}] (H) -- (I);
		\end{tikzpicture}
		\vspace{1em}
		\caption{Deploying the DTR learned from Algorithm \ref{alg:BQL2}}\label{fig:balance_GRF}
	\end{figure}
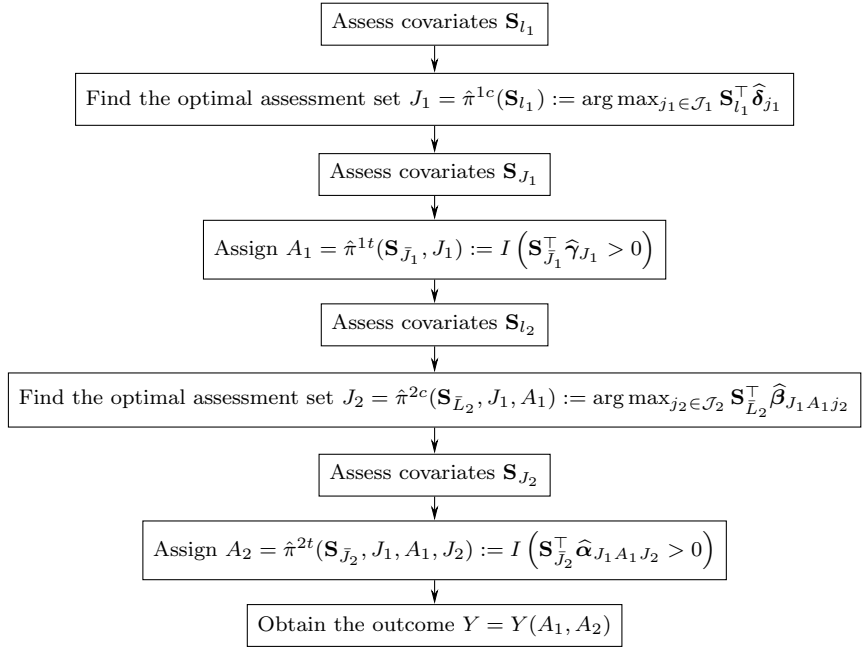

	In Algorithm \ref{alg:BQL2}, we allow flexible estimation for the nuisance functions $(f_2, g_2, f_{j_1}, g_1)$, permitting the use of machine learning methods, such as random forests and neural networks, to improve estimation accuracy. To mitigate the associated bias and reduce overfitting, we adopt a cross-fitting procedure \citep{chernozhukov2018double}. For the contrast functions, we adopt linear models to enhance interpretability and reduce modeling complexity, thereby facilitating the practical implementation of the estimated regime; see, for example, \cite{ertefaie2021robust} for similar approaches. Nonetheless, the proposed procedure can be readily extended to accommodate non-parametric estimators for the contrast functions when greater modeling flexibility is required. Moreover, while we employ R-learners to estimate the contrast functions for treatment assignments, the framework can be extended to incorporate alternative methods such as the S-, T-, X-, and DR-learners \citep{kunzel2019metalearners, foster2023orthogonal, kennedy2023towards}.

	After applying Algorithm \ref{alg:BQL2} and obtaining the estimates of the contrast functions, the learned DTR can be applied sequentially to new individuals from the test population. See an illustration of the whole deploying process in Figure \ref{fig:balance_GRF}.

	\section{Theoretical Results}\label{sec:theory}
	In this section, we establish the asymptotic properties of the proposed estimators and derive regret bounds for the estimated regime. 
	
	We first introduce the target linear parameters $(\balpha_{j_1a_1j_2}^*, \bbeta_{j_1a_1j_2}^*, \bgamma_{j_1}^*, \bdelta_{j_1}^*)$ associated with the restricted contrast functions. These parameters are defined as the population-level minimizers of the corresponding risk functions, serving as the population analogues of the empirical minimizers presented in Algorithm \ref{alg:BQL2}. For each $j_1 \in \mathcal{J}_1$ and $a_1 \in \{0,1\}$, we also define $\bar{Q}^{1t*}(\bs_{1}, j_1, a_1)$ as the working unrestricted Q-functions for the first-stage treatment assignment. They are constructed based on the linear working models for the contrast functions, representing the population-level target under the specified modeling framework. Moreover, we denote $f_{j_1}^*$ as the working model for $f_{j_1}(\bs_1) = E( Y_{j_1}^{1t} \mid \bS_1 = \bs_1 )$, where the pseudo-outcome $Y_{j_1}^{1t}$ in \eqref{def:Y1t} is replaced by the pseudo-outcome defined using linear regimes. The explicit forms of these parameters and working models are provided in \eqref{def:alpha-star}-\eqref{def:delta-star} in Section \ref{sec:sm_target} of the Supplementary Material.
	
	The following assumptions are introduced.
	
	\begin{assumption}[Boundedness]\label{ass:uni_bound_}
	There exists a constant $c_1>0$ such that $|\bar{\Delta}^{2t}(\bar{\bS}_{2i},A_{1i})|<c_1$, $|\bar{Q}^{1t*}(\bS_{1i},j_1,1)-\bar{Q}^{1t*}(\bS_{1i},j_1,0)|<c_1$, and $\|\bar{\bS}_{2i}\|_\infty<c_1$ for any $j_1\in\mathcal{J}_1$ almost surely.
	\end{assumption} 

	\begin{assumption}[Nuisance estimation errors]\label{ass:nuisance1_}
		For each $k\in[K]$ and $j_1\in\mathcal J_1$, 
		\begin{align*}
			\|\fhat_{2}^{-k}(\cdot)-f_{2}(\cdot)\|_{P,2}+\|\fhat_{j_1}^{-k}(\cdot)-f_{j_1}^*(\cdot)\|_{P,2}&=o_p(1),\\
			\|\ghat_{2}^{-k}(\cdot)-g_{2}(\cdot)\|_{P,2}+\|\ghat_{1}^{-k}(\cdot)-g_{1}(\cdot)\|_{P,2}&=o_p(n^{-1/4}),\\
			\|\ghat_{2}^{-k}(\cdot)-g_{2}(\cdot)\|_{P,2}\|\fhat_{2}^{-k}(\cdot)-f_{2}(\cdot)\|_{P,2}+\|\fhat_{j_1}^{-k}(\cdot)-f_{j_1}^*(\cdot)\|_{P,2}\|\ghat_{1}^{-k}(\cdot)-g_{1}(\cdot)\|_{P,2}&=o_p(n^{-1/2}).
		\end{align*}
	\end{assumption}
	
	\begin{assumption}[Lower bounds]\label{ass:PD_matrix_}
	There exists a constant $c_2>0$ such that $\lambda_{\min}[E\{{\rm var}(A_{2i}\mid \bar{\bS}_{2i},A_{1i})(\bar{\bS}_{2i},A_{1i})(\bar{\bS}_{2i},A_{1i})^\top\}]>c_2$ as well as $\lambda_{\min}[E\{{\rm var}(A_{1i}\mid \bS_{1i})\bS_{1i}\bS_{1i}^\top\}]>c_2$.
	\end{assumption}
	
	\begin{assumption}[Margin conditions]\label{ass:margin2}
		There exist constants $c_3>0$ and $r\ge 1$ such that for any $t>0$, $a_1\in\{0,1\}$, $j_1, j_1'\in\mathcal{J}_1$, $j_2, j_2'\in\mathcal{J}_2$ with $j_1\neq j_1'$ and $j_2\neq j_2'$:
		$P(|\bS_{\bar j_2i}^\top\balpha_{j_1a_1j_2}^*|\le t)\le c_3 t^r$, $P(|\bS_{\bar l_2i}^\top\bbeta^*_{j_1a_1j_2}-\bS_{\bar l_2i}^\top\bbeta^*_{j_1 a_1j_2^\prime}|\le t)\le c_3 t^r$, $P(|\bS_{\bar j_1i}^\top\bgamma_{j_1}^*|\le t)\le c_3 t^r$, $P(|\bS_{l_1i}^\top\bdelta^*_{j_1}-\bS_{l_1i}^\top\bdelta^*_{j_1^\prime}|\le t)\le c_3 t^r$.
	\end{assumption}
	\begin{assumption}[Model misspecification error]\label{ass:model_misspecification}
		For each $a_1\in\{0,1\}$, $j_1\in\mathcal{J}_1$, and $j_2\in\mathcal{J}_2$, the following holds with some $\epsilon_n\geq0$ almost surely: $|\Delta^{2t}(\bS_{\bar{j}_2i},j_1,a_1,j_2)-\bS_{\bar{j}_2i}^\top\balpha^*_{j_1a_1j_2}|+|\Delta^{2c}(\bS_{\bar{l}_2i},j_1,a_1,j_1)-\bS_{\bar{l}_2i}^\top\bbeta_{j_1a_1j_2}^*|+|\Delta^{1t}(\bS_{\bar{j}_1i},j_1)-\bS_{\bar{j}_1i}^\top\bgamma^*_{j_1}|+|\Delta^{1c}(\bS_{l_1i},j_1)-\bS_{l_1i}^\top\bdelta^*_{j_1}|\leq\epsilon_n$. 
	\end{assumption}
	
	Assumption \ref{ass:uni_bound_} imposes standard boundedness conditions on the covariates and contrast functions; see, for example, \cite{hu2020dtr, ertefaie2021robust}. Assumption \ref{ass:nuisance1_} specifies convergence rate requirements for the nuisance estimators, which are commonly required when applying the R-learner framework \citep{nie2021quasi}. In practice, various non-parametric methods can be used to achieve these convergence rates and ensure the validity of the estimation. Assumption \ref{ass:PD_matrix_} ensures that the ordinary least squares are well-defined by imposing positive definiteness conditions on the relevant design matrices. Assumption \ref{ass:margin2} quantifies the probability of individual covariates occurring near decision boundaries, thereby characterizing the complexity of making accurate decisions. Specifically, a smaller $r$ implies a larger proportion of individuals are located close to the decision boundary, which increases the difficulty of the decision-making task. The margin condition is widely used in classification \citep{mammen1999smooth} to describe classification complexity and generalization error, and it also plays a key role in regret bound analyses for decision-making problems \citep{hu2020dtr, bastani2020online}. 
	Lastly, while permitting potential misspecification in the linear models for contrast functions, Assumption \ref{ass:model_misspecification} explicitly quantifies the approximation errors between the linear working models and the true underlying contrast functions.

	The following theorem establishes the asymptotic normality of the proposed estimators. 
	
	\begin{theorem}\label{thm:asymptotic_normality}
		Let Assumptions \ref{ass:uni_bound_}-\ref{ass:margin2} hold with some $r>1$. Then for each $a_1\in\{0,1\}$, $j_1\in\mathcal J_1$, and $j_2\in\mathcal J_2$, as $n\to\infty$, 
		$\sqrt{n}(\widehat{\balpha}_{j_1a_1j_2}-\balpha_{j_1a_1j_2}^*)\xrightarrow{\rm d}\mathcal{N}(0,\bM_{j_1a_1j_2}^{\balpha})$, $\sqrt{n}(\widehat{\bbeta}_{j_1a_1j_2}-\bbeta_{j_1a_1j_2}^*)\xrightarrow{\rm d}\mathcal{N}(0,\bM_{j_1a_1j_2}^{\bbeta})$, $\sqrt{n}(\widehat{\bgamma}_{j_1}-\bgamma_{j_1}^*)\xrightarrow{\rm d}\mathcal{N}(0,\bM_{j_1}^{\bgamma})$, and $\sqrt{n}(\widehat{\bdelta}_{j_1}-\bdelta_{j_1}^*)\xrightarrow{\rm d}\mathcal{N}(0,\bM_{j_1}^{\bdelta})$ with some matrices $\bM_{j_1a_1j_2}^{\balpha}$, $\bM_{j_1a_1j_2}^{\bbeta}$, $\bM_{j_1}^{\bgamma}$, and $\bM_{j_1}^{\bdelta}$.	\end{theorem}

	A complete version of Theorem \ref{thm:asymptotic_normality}, which includes the explicit expressions of the covariance matrices, is presented in Theorem \ref{thm:asymptotic_normality2} of the Supplementary Material. As the covariance matrices can be consistently estimated using plug-in estimators, Theorem \ref{thm:asymptotic_normality} enables asymptotically valid statistical inference for the target linear parameters that characterize the contrast between the Q-functions at each decision stage.

	Additionally, we present a regret bound that characterizes the difference between the profits of the estimated regime and that of the population-level optimal regime.

	\begin{theorem}\label{thm:regret}
		Let Assumptions \ref{ass:uni_bound_}-\ref{ass:model_misspecification} hold for any $r\geq1$, then
		${\rm Profit}(\check{\pi}^{1c},\check{\pi}^{1t},\check{\pi}^{2c},\check{\pi}^{2t})-{\rm Profit}(\hat{\pi}^{1c},\hat{\pi}^{1t},\hat{\pi}^{2c},\hat{\pi}^{2t})=O_p\left\{(1/\sqrt{n})^{r+1}+\epsilon_n^{r+1}\right\}$ as $n\to\infty$.
	\end{theorem}
	
	Although Theorem \ref{thm:asymptotic_normality} assumes the margin condition to hold with $r>1$ to establish the asymptotic normality of the linear estimators, the regret bound in Theorem \ref{thm:regret} only requires the margin condition holds for any $r\geq1$, with a larger $r$ yielding a faster convergence rate. While our framework accommodates model misspecification, the estimated linear regime achieves profit convergence to the global optimal regime only when the misspecification error diminishes to zero, that is, $\epsilon_n=o(1)$. Moreover, in the absence of model misspecification, that is, $\epsilon_n=0$, the regret bound achieves the order $O_p(1/n)$ when $r=1$, which aligns with the results established in \cite{hu2025fast} for linear Fitted Q-Iteration.

	\section{Numerical experiments}\label{sec:numerical}
	
	We examine the finite-sample performance of the proposed Balanced Q-learning (BQL) method through simulation studies in Section \ref{sec:sim}, followed by a real data application in Section \ref{sec:real}. 
	
	For comparison, we also implement and report results based on the high-dimensional Q-learning (HDQ) algorithm of \cite{zhu2019proper} and the robust Q-learning (RQL) method of \cite{ertefaie2021robust}. To account for the costs $(C^{1t},C^{2t})$ associated with treatment assignments, we follow the approach of \cite{athey2021policy} and extend both the HDQ and RQL methods by incorporating constant thresholds $C^{1t}(1)-C^{1t}(0)$ and $C^{2t}(1)-C^{2t}(0)$ into the decision boundaries for the first- and second-stage treatment decisions, respectively. We compare our method with the adjusted versions of the HDQ and RQL methods, highlighting the effectiveness of the proposed BQL method in incorporating the costs associated with covariate assessments to achieve more cost-effective decision-making.

	The HDQ method includes regularization in the estimation process, resulting in sparse solutions for the estimated contrast functions. Consequently, the costs $(C^{1c},C^{2c})$ for covariate assessments are computed based on the selected covariates. In contrast, the RQL method produces dense solutions, which lead to the full covariate assessment costs $C^{1c}(j_1^{\rm f})$ and $C^{2c}(j_2^{\rm f})$. For both the BQL and RQL methods, the nuisance functions are estimated using the SuperLearner algorithm \citep{van2007super}, where we consider random forests and generalized linear models as the base learners.

	\subsection{Simulation results}\label{sec:sim}
	
	We illustrate the performance of the proposed BQL algorithm through simulations under various models. Additional results can be found in Section \ref{sec:sm_sim} of the Supplementary Material.
	
	Consider the following data-generating process for the observed samples throughout. For each $i\in[n]$, generate $\bS_{1i}\sim\mathcal{N}(0,\bI_p)$ and $A_{i1}\in\{0,1\}$ with $P(A_{1i}=1\mid\bS_1)=\phi(\bS_{1i}^\top\balpha_1)$, where the logistic function is defined as $\phi(t)={\rm exp}(t)/\{1+{\rm exp}(t)\}$ for $t\in\mathbb{R}$. For the second stage, we generate $\bS_{2i}=\bS_{1i}+A_{1i}\bS_{1i}+\mathcal{N}(0,\bI_p)$, and $A_{2i}\in\{0,1\}$ is drawn according to $P(A_{2i}=1\mid \bar{\bX}_{2i})=\phi(\bar{\bX}_{2i}^\top\balpha_2)$, where $\bar{\bX}_{2i}=(\bS_{1i}^\top,A_{1i},\bS_{2i}^\top)^\top$. Finally, the outcome variable is generated through $Y_i=\bar{\bX}_{2i}^\top(\bbeta_1+A_{1i}\bbeta_2+A_{2i}\bbeta_3)+\mathcal{N}(0,0.5^2)$. Our objective is to design a dynamic treatment regime (DTR) that optimizes the following adjusted profit:
	\begin{align}
	{\rm Profit}_{\lambda}(\pi^{1c},\pi^{1t},\pi^{2c},\pi^{2t}) = E\left[Y(A_1,A_2) - \lambda \left\{ C^{1c}(J_1) + C^{1t}(A_1) + C^{2c}(J_2) + C^{2t}(A_2) \right\} \right],\label{def:profit-lambda}
	\end{align}
where the user-specified parameter $\lambda \in [0, \infty)$ controls the trade-off between maximizing the outcome and penalizing the costs. Optimizing the adjusted ${\rm Profit}_{\lambda}(\pi^{1c},\pi^{1t},\pi^{2c},\pi^{2t})$ is equivalent to optimizing the original unadjusted ${\rm Profit}(\pi^{1c},\pi^{1t},\pi^{2c},\pi^{2t})$ as defined in \eqref{def:profit}, with the adjusted costs $(\lambda C^{1c}, \lambda C^{1t}, \lambda C^{2c}, \lambda C^{2t})$ replacing the original cost terms.
	
	In Models 1-3, we set $C^{1t}(a) = C^{2t}(a) = 0$ for $a \in \{0,1\}$ and focus exclusively on the costs associated with covariate assessments. Additional simulation results that incorporate treatment assignment costs are provided in Section \ref{sec:sm_sim} of the Supplementary Material.

	Model 1: Let $\balpha_{1} = (\mathbf{1}_2, \mathbf{0}_3)^\top$, $\balpha_{2} = (\mathbf{1}_2, \mathbf{0}_4, \mathbf{1}_2, \mathbf{0}_3)^\top$, $\bbeta_{1} = (\mathbf{1}_2, 0.5, \mathbf{0}_4, 1, 0.5, 0, 1)^\top$, $\bbeta_{2} = (\mathbf{1}_2, 0.5, \mathbf{0}_8)^\top$, $\bbeta_{3} = (\mathbf{0}_7, 1, 0.5, 0, 1)^\top$, with sample size $n = 500$ and covariate dimension $p = 5$. We set $l_1 = [p]$ and $\mathcal{J}_1 = \{\emptyset\}$, focusing on the costs associated with the second-stage covariate assessment. Specifically, we consider $\mathcal{J}_2 = \{j_{21}, j_{22}\}$ as the set of candidate assessment options and $l_2 = \{1\}$, where $j_{21} = \{2,3,4\}$ and $j_{22} = \{2,3,4,5\}$. The associated costs are $C^{2c}(j_{21}) = 0$ and $C^{2c}(j_{22}) = 0.1$, with varying values of $\lambda$ to control the utility-cost trade-off. In this model, the features $\{S_{22}, S_{23}, S_{25}\}$ are informative for the second-stage treatment assignment, as they contribute to the heterogeneity in potential outcomes between treatment groups.

	Model 2: Let $\balpha_1=\{1,\mathbf{0}_{4}\}^\top$, $\balpha_2=\{1,0,0.5\cdot \mathbf{1}_{5},0,0.5\cdot \mathbf{1}_2,0\}^\top$, $\bbeta_1=\{1,0,0.5,\mathbf{0}_3,1,0,\mathbf{1}_2,0\}^\top$, $\bbeta_2=\{1,0,0.5,\mathbf{0}_8\}^\top$, $\bbeta_3=(\mathbf{0}_{6},1,0,\mathbf{1}_2,0)^\top$, $n=500$, and $p=5$. We set $l_1=l_2=\{1\}$, $\mathcal{J}_1=\{j_{11},j_{12}\}$, and $\mathcal{J}_2=\{j_{21},j_{22},j_{23}\}$, where $j_{11}=\{3,4,5\}$, $j_{12}=\{2,3,4,5\}$, $j_{21}=\{2,3\}$, $j_{22}=\{2,3,4\}$, and $j_{23}=\{2,3,4,5\}$. The associated costs are $C^{1c}(j_{11})=C^{2c}(j_{21})=0$, $C^{2c}(j_{22})=0.1$, and $C^{1c}(j_{12})=C^{2c}(j_{23})=0.2$, with a varying $\lambda$. In this model, $\{S_{11},S_{13},S_{14}\}$ and $\{S_{21},S_{23},S_{24}\}$ contribute to the heterogeneity in potential outcomes.

\begin{figure}[t]
	\centering
	\begin{subfigure}[b]{0.4\textwidth}
		\centering
		\includegraphics[width=\textwidth]{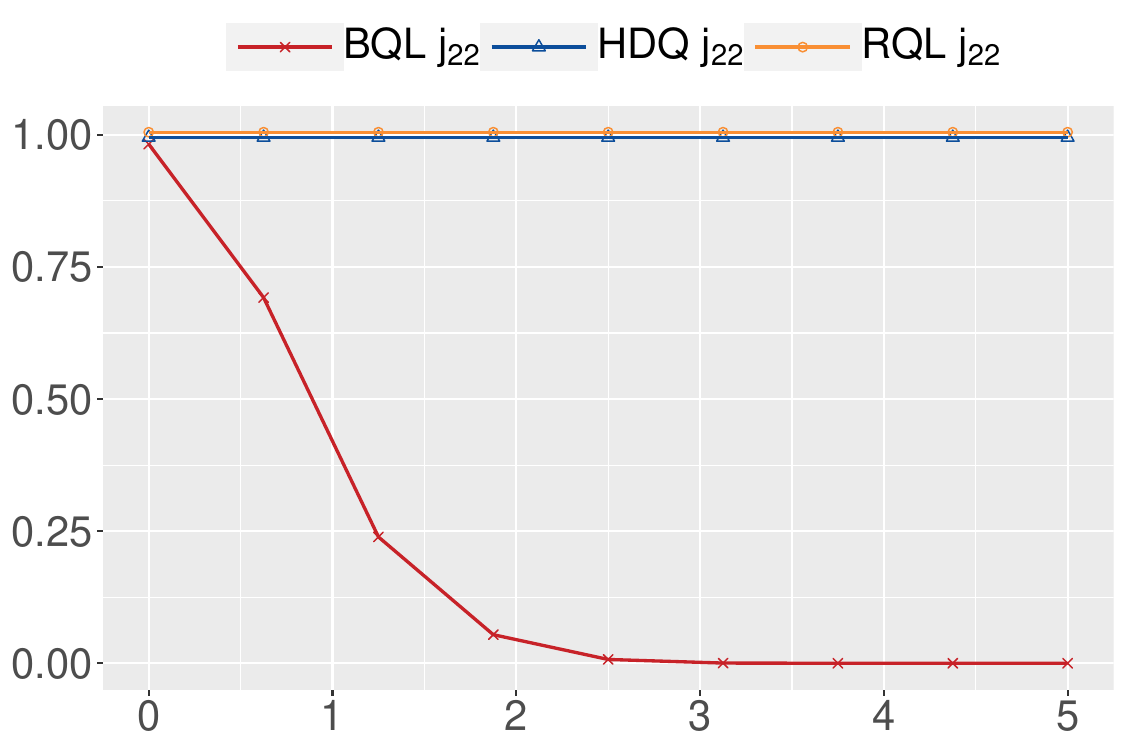}
		\captionsetup{skip=-5pt}
		\caption{Frequency as $\lambda$ varies}
		\label{fig:M1_freq}
	\end{subfigure}
	\hspace{0.03\textwidth}
	\begin{subfigure}[b]{0.4\textwidth}
		\centering
		\includegraphics[width=\textwidth]{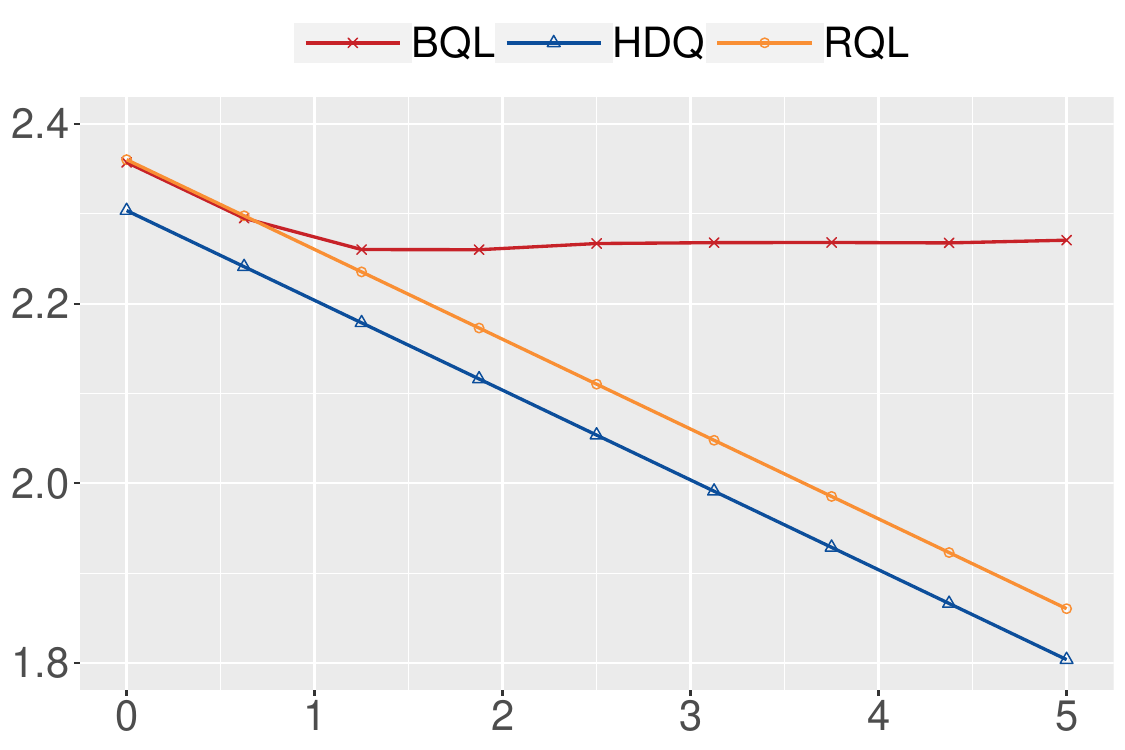}
		\captionsetup{skip=-5pt}
		\caption{Profit as $\lambda$ varies}
		\label{fig:M1_profit}
	\end{subfigure}
	\captionsetup{skip=-5pt}
	\caption{Simulation results under Model 1}
	\label{fig:Model 1}
	\end{figure}

	Model 3: Let $\balpha_{1}=(0.6,0.6^2,0.6^3)^\top$, $\balpha_{2}=(0.6,\cdots,0.6^4,0,0.6^6,0.6^7)^\top$, $\bbeta_{1}=(1.5,1,\mathbf{0}_3,1,2)^\top$, $\bbeta_{2}=\mathbf{0}_{11}^\top$, $\bbeta_{3}=(0.5,\mathbf{0}_4,1,2)^\top$, $n=500$, and $p=3$. We set $l_1=[p]$, $\mathcal{J}_1=\{\emptyset\}$, $l_2=\emptyset$, and $\mathcal{J}_2=\{j_{21},\cdots,j_{28}\}$, where $j_{21}=\emptyset$, $j_{22}=\{1\}$, $j_{23}=\{2\}$, $j_{24}=\{3\}$, $j_{25}=\{1,2\}$, $j_{26}=\{1,3\}$, $j_{27}=\{2,3\}$, $j_{28}=\{1,2,3\}$. Consider $C_2^c(j_{2i})=0.1|j_{2i}|$ for $i\in[8]$, with a varying $\lambda$. In this model, the features $\{S_{22},S_{23}\}$ contribute to the heterogeneity in potential outcomes.

	The simulations are repeated 200 times, and the results are summarized in Figures \ref{fig:Model 1}-\ref{fig:Model 3}. We report the selection frequencies of the candidate covariate sets by the DTRs learned from the BQL, HDQ, and RQL methods, along with their corresponding average profits evaluated using 5,000 independent test samples. Since neither the HDQ nor the RQL methods incorporate information on covariate assessment costs, their selection frequencies remain unchanged as $\lambda$ increases. Consequently, their overall profits exhibit a linear decline with increasing $\lambda$.

	Under Model 1, covariates in both candidate sets $j_{21}$ and $j_{22}$ are informative for the outcome, with $j_{22}$ including an additional feature that enhances predictive power but comes with a higher assessment cost. As shown in Figure \ref{fig:Model 1}, when $\lambda$ is close to zero, meaning the costs are given little weight, the BQL method tends to select the larger and more costly set $j_{22}$. Under such scenarios, BQL and RQL yield similar profit levels and both outperform HDQ. As $\lambda$ increases, the proposed BQL method progressively shifts towards selecting the smaller set $j_{21}$, which incurs lower assessment costs. This adaptive behavior leads to an improved trade-off between utility and cost, resulting in superior overall profit performance.

		\begin{figure}[ht]
		\centering
		\begin{subfigure}[b]{0.4\textwidth}
			\centering
			\includegraphics[width=\textwidth]{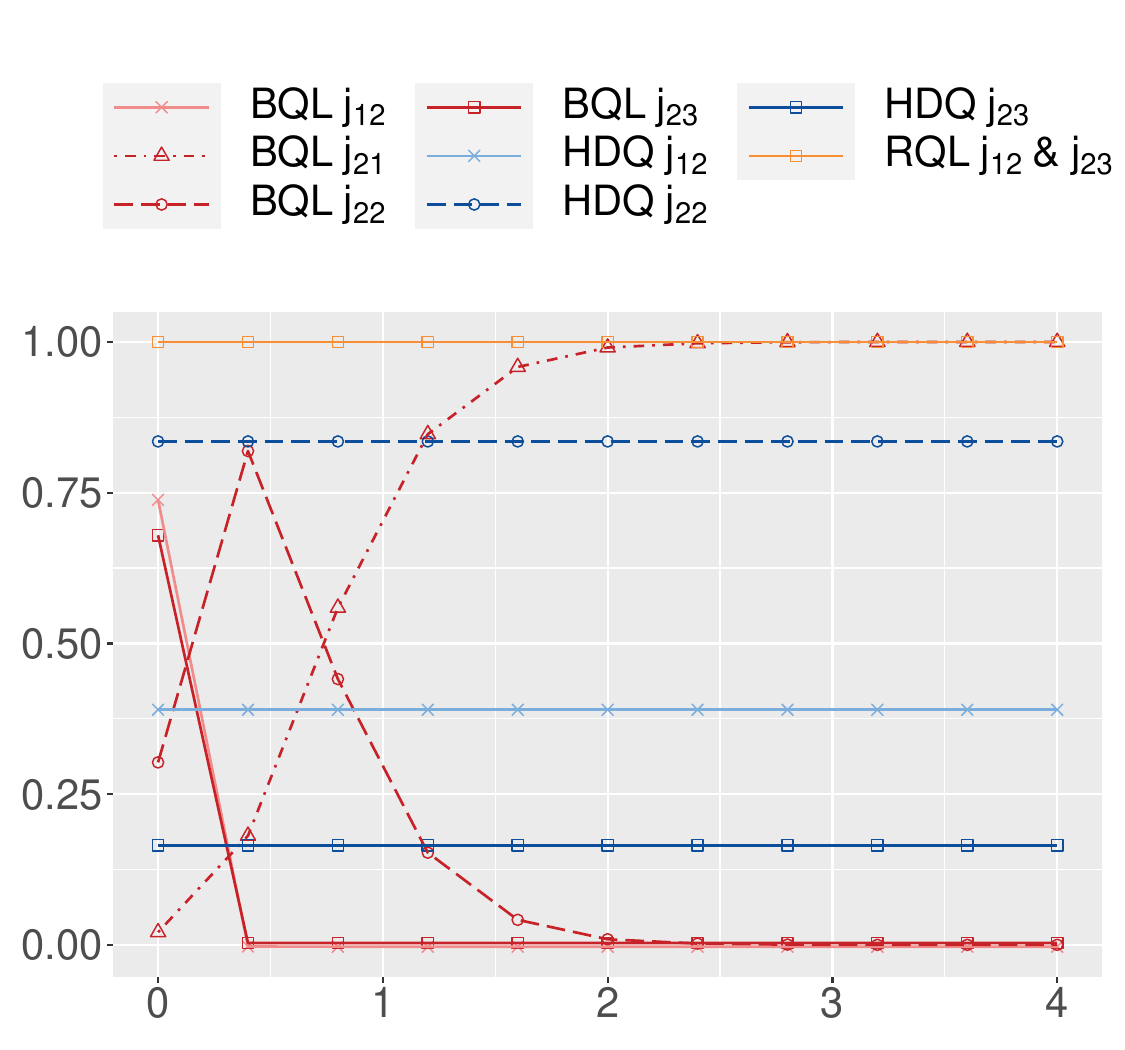}
			\captionsetup{skip=-5pt}
			\caption{Frequency as $\lambda$ varies}
			\label{fig:M3_freq}
		\end{subfigure}
		\hspace{0.03\textwidth}
		\begin{subfigure}[b]{0.4\textwidth}
			\centering
			\includegraphics[width=\textwidth]{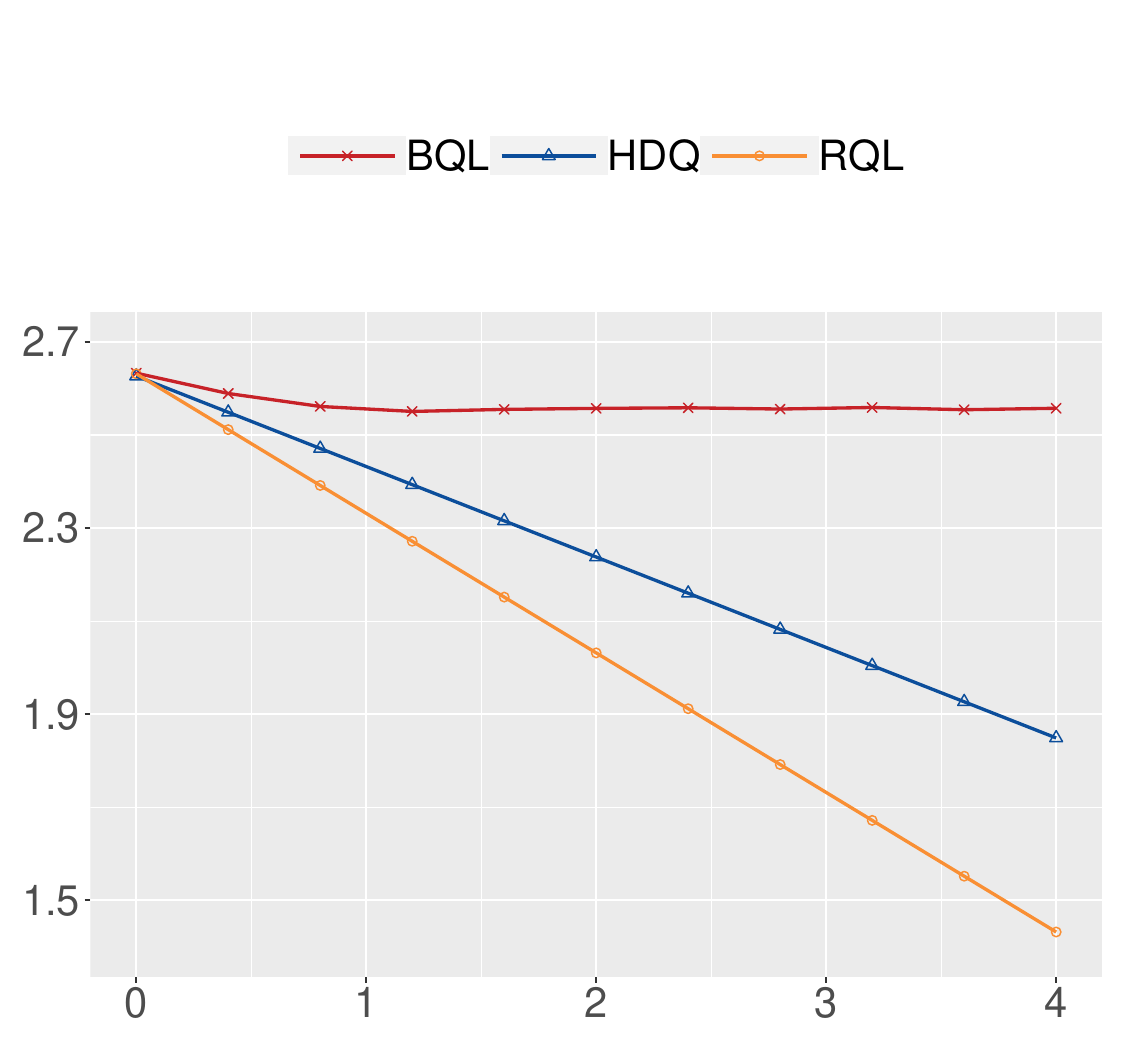}
			\captionsetup{skip=-5pt}
			\caption{Profit as $\lambda$ varies}
			\label{fig:M3_profit}
		\end{subfigure}
		\captionsetup{skip=-5pt}
		\caption{Simulation results under Model 2}
		\label{fig:Model 2}
		\end{figure}

	Under Model 2, although the candidate sets \(j_{12}\) and \(j_{23}\) include additional features, $S_{15}$ and $S_{25}$, compared to \(j_{11}\) and \(j_{22}\), respectively, incorporating these features does not improve treatment decision accuracy and is associated with higher costs. Consequently, the proposed BQL method avoids selecting the more expensive yet equally informative candidate sets \(j_{12}\) and \(j_{23}\) whenever $\lambda > 0$, where the costs are taken into account, as shown in Figure \ref{fig:Model 2}. By achieving a better balance between utility and cost, BQL yields higher profits than both the RQL and HDQ methods for all values of $\lambda > 0$.

	For Model 3, we consider all possible combinations for the second-stage covariate assessment. Since \(S_{21}\) does not contribute to the potential outcome, while \(S_{22}\) and \(S_{23}\) do, the BQL method tends to select \(j_{27} = \{2,3\}\) for small but strictly positive values of $\lambda$. As $\lambda$ increases, BQL favors \(j_{24} = \{3\}\), because (a) assessing a single feature incurs a lower cost compared to the 2-feature set \(j_{27}\), and (b) \(S_{23}\) is more important than \(S_{22}\), given that it is associated with a larger linear coefficient in the outcome model. When $\lambda$ becomes even larger, collecting any feature results in substantial costs. Consequently, BQL eventually shifts to \(j_{21} = \emptyset\), avoiding covariate assessments entirely. In comparison, RQL focuses on the full set \(j_{28} = \{1,2,3\}\), as it does not involve any feature selection steps. On the other hand, HDQ predominantly selects the set \(j_{27} = \{2,3\}\), regardless of the value of $\lambda$. As shown in Figure \ref{fig:Model 3}, BQL provides comparable profits to HDQ and RQL when $\lambda$ is small, and outperforms them as $\lambda$ increases.
	
		\begin{figure}[ht]
		\centering
		\begin{subfigure}[b]{0.4\textwidth}
			\centering
			\includegraphics[width=\textwidth]{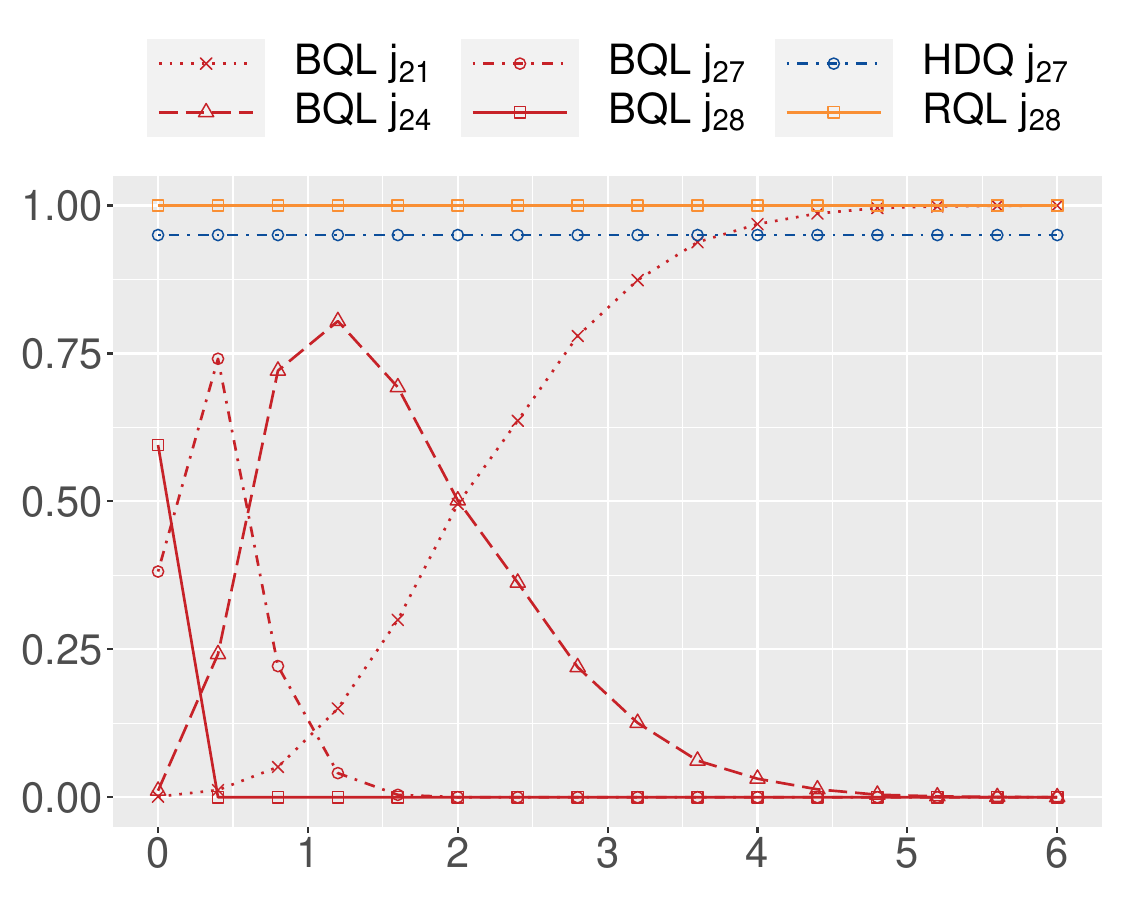}
			\captionsetup{skip=-5pt}
			\caption{Frequency as $\lambda$ varies}
			\label{fig:M4_freq}
		\end{subfigure}
		\hspace{0.03\textwidth}
		\begin{subfigure}[b]{0.4\textwidth}
			\centering
			\includegraphics[width=\textwidth]{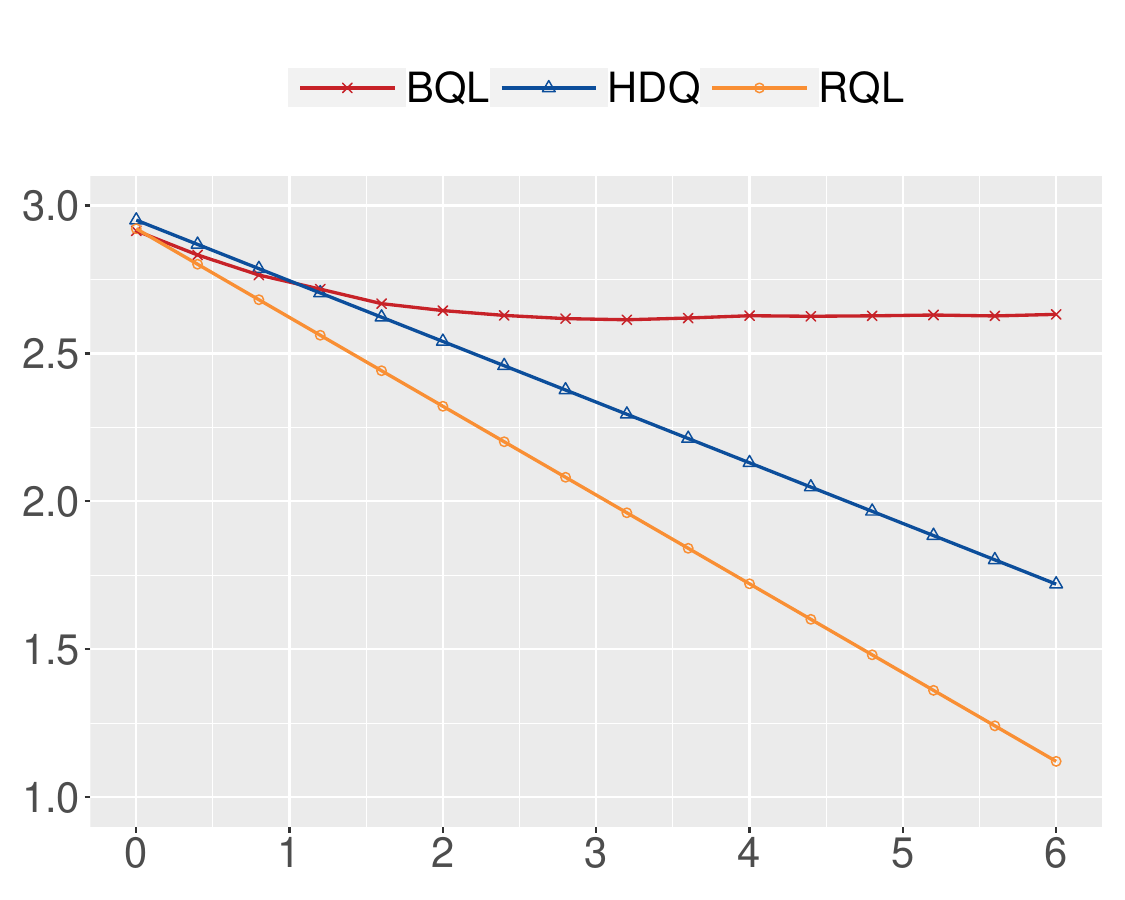}
			\captionsetup{skip=-5pt}
			\caption{Profit as $\lambda$ varies}
			\label{fig:M4_profit}
		\end{subfigure}
		\captionsetup{skip=-5pt}
		\caption{Simulation results under Model 3}
		\vspace{1em}
		\label{fig:Model 3}
	\end{figure}

	\subsection{Application to the MIMIC-III Database}\label{sec:real}

	We demonstrate the performance of the proposed method through an application to a dataset from the Medical Information Mart for Intensive Care version III (MIMIC-III) Database \citep{johnson2016mimic}, a large, publicly available critical care resource containing de-identified health records from over 40,000 ICU patients at Beth Israel Deaconess Medical Center between 2001 and 2012. Following the data processing pipeline established by \cite{komorowski2018artificial}, we identify a sepsis patient cohort consisting of 18,911 subjects, with clinical records encoded as discrete time-series at 4-hour intervals for temporal analysis. The data are then consolidated into two 24-hour stages using average values, resulting in a final cohort of 17,914 patients who have records observed at least once during both day one and day two.

		\begin{table}[t]
		\vspace{0.5em}
		\renewcommand{\arraystretch}{0.5} 
		\centering
		\begin{tabular}{ccc}
			\toprule
			Candidate set & Combinations of tests & Cost \\
			\midrule
			$j_{t1}$ & No test & $0$ \\
			$j_{t2}$ & CBC only & $24 \times 6$ \\
			$j_{t3}$ & BT, CBC, and CF & $113 \times 6$ \\
			$j_{t4}$ & ABG only & $255 \times 6$ \\
			$j_{t5}$ & CBC and ABG & $279 \times 6$ \\
			$j_{t6}$ & All tests & $368 \times 6$ \\
			\bottomrule
		\end{tabular}
		\caption{Candidate combinations for conducting tests and their associated costs in USD}
		\label{tab:test_costs}
	\end{table}

	We use the 90-day survival status of patients as the outcome variable, $Y\in\{0,1\}$. The treatment variables $A_1, A_2 \in \{0, 1\}$ are defined as the average intravenous fluid volume administered within each 24-hour stage, discretized based on the median. The baseline treatment costs are set to $C^{1t}(0) = C^{2t}(0) = 0$. The corresponding increased costs for choosing $A_1 = 1$ and $A_2 = 1$ are denoted by $C^{1t}(1)$ and $C^{2t}(1)$, respectively. Using an intravenous fluid price of \$30 per 100 mL, the resulting cost values are $C^{1t}(1) = 973.95$ and $C^{2t}(1) = 701.24$, calculated based on the average differences in intravenous fluid volume between the treatment groups.

	There are 20 and 17 baseline covariates at the first and second stages, respectively, for which we ignore the corresponding assessment costs. At each stage, there are 24 laboratory variables grouped into four categories: (1) Biochemical Tests (BT), with 11 features; (2) Complete Blood Count (CBC), with 3 features; (3) Coagulation Function (CF), with 3 features; and (4) Arterial Blood Gas (ABG) analysis, with 7 features. The prices of the four laboratory tests are \$67, \$24, \$22, and \$255\footnote{The test prices used in this study are referenced from the online healthcare marketplace \url{mdsave.com}.}, respectively. Additionally, there are two derived covariates, the SOFA and SIRS scores, which can only be calculated when the related tests are conducted. Specifically, collecting the SIRS score requires conducting the CBC and ABG tests, while collecting the SOFA score requires conducting the BT, CBC, and ABG tests. All covariates used in this study and their corresponding descriptions are listed in Section~\ref{sec:sm_rd} of the Supplementary Material. We assume that tests are conducted every 4 hours. For each stage $t = 1, 2$, we consider six combinations of tests: $\mathcal{J}_t = \{j_{t1}, j_{t2}, \dots, j_{t6}\}$, with the detailed composition and associated costs of the candidate sets listed in Table \ref{tab:test_costs}.

			\begin{table}[ht]
		\centering
		\caption{Estimated profit and utility under different $\tau$ values} 
		\label{table:profit}
		\renewcommand{\arraystretch}{0.5} 
		\resizebox{1\textwidth}{!}{
			\begin{tabular}{c c *{9}{c}} 
				\toprule
				\multicolumn{2}{c}{$\tau$} & \multicolumn{1}{c}{$\infty$} & \multicolumn{1}{c}{100000} & \multicolumn{1}{c}{10000} & \multicolumn{1}{c}{5000} & \multicolumn{1}{c}{2000} & \multicolumn{1}{c}{1000} & \multicolumn{1}{c}{500} & \multicolumn{1}{c}{200} & \multicolumn{1}{c}{100} \\
				\midrule
				\multirow{3}{*}{Profit} & BQL & 0.8487 & 0.8502 & 0.8448 & 0.8428 & 0.8401 & 0.8335 & 0.8253 & 0.8145 & 0.8151 \\
				& RQL & 0.8487 & 0.8483 & 0.8435 & 0.8381 & 0.8221 & 0.7963 & 0.7441 & 0.5955 & 0.3673 \\
				& HDQ & 0.8319 & 0.8314 & 0.8259 & 0.8205 & 0.8022 & 0.7744 & 0.7237 & 0.5825 & 0.3643 \\
				\midrule
				\multirow{3}{*}{Utility} & BQL & 0.8487 & 0.8507 & 0.8491 & 0.8488 & 0.8493 & 0.8475 & 0.8393 & 0.8295 & 0.8265 \\
				& RQL & 0.8487 & 0.8489 & 0.8488 & 0.8485 & 0.8481 & 0.8479 & 0.8457 & 0.8392 & 0.8314 \\
				& HDQ & 0.8319 & 0.8319 & 0.8311 & 0.8309 & 0.8278 & 0.8244 & 0.8198 & 0.8125 & 0.8166 \\
				\bottomrule
		\end{tabular}}
		\vspace{1em}
	\end{table}

	Our goal is to maximize the adjusted profit \eqref{def:profit-lambda} using a different parametrization by setting $\lambda = 1\%/\tau$, where $\tau$ represents the cost a patient is willing to pay for a one-percentage-point increase in the survival rate. We use 14,000 patient records as the training set to estimate treatment regimes using the three algorithms: BQL, RQL, and HDQ, as outlined in Section \ref{sec:sim}. The remaining data is used as the test set. Tables \ref{table:profit}-\ref{table:cov_set} summarize the test-set performance averaged over 100 random splits. These tables report the selection frequency of each candidate combination for the laboratory tests, the frequency of treatment assignments, the estimated utility, and the estimated profit, where the utility is calculated using the inverse probability weighting method described in \cite{zhao2015new}.

	\begin{table}[b]
		\centering
		\caption{Selection frequency of treatments under different $\tau$ values} 
		\label{table:trt}
		\renewcommand{\arraystretch}{0.5} 
		\resizebox{1\textwidth}{!}{
			\begin{tabular}{c c *{9}{c}} 
				\toprule
				\multicolumn{2}{c}{$\tau$} & \multicolumn{1}{c}{$\infty$} & \multicolumn{1}{c}{100000} & \multicolumn{1}{c}{10000} & \multicolumn{1}{c}{5000} & \multicolumn{1}{c}{2000} & \multicolumn{1}{c}{1000} & \multicolumn{1}{c}{500} & \multicolumn{1}{c}{200} & \multicolumn{1}{c}{100} \\
				\midrule
				\multirow{3}{*}{$A_1=1$} & BQL & 0.5281 & 0.5299 & 0.5217 & 0.5102 & 0.4833 & 0.4489 & 0.3255 & 0.1383 & 0.0295 \\
				& RQL & 0.5259 & 0.5254 & 0.5201 & 0.5145 & 0.4972 & 0.4689 & 0.4112 & 0.2537 & 0.0930 \\
				& HDQ & 0.6158 & 0.6137 & 0.5936 & 0.5698 & 0.4943 & 0.3652 & 0.1780 & 0.0112 & 0.0005 \\
				\midrule
				\multirow{3}{*}{$A_2=1$} & BQL & 0.4419 & 0.4408 & 0.4351 & 0.4310 & 0.4136 & 0.3866 & 0.3337 & 0.2295 & 0.1215 \\
				& RQL & 0.4422 & 0.4419 & 0.4392 & 0.4361 & 0.4269 & 0.4118 & 0.3825 & 0.3008 & 0.1927 \\
				& HDQ & 0.3490 & 0.3489 & 0.3482 & 0.3476 & 0.3456 & 0.3430 & 0.3307 & 0.2683 & 0.1730 \\
				\bottomrule
		\end{tabular}}
		\vspace{1em}
	\end{table}

		\begin{table}[ht]
		\centering
		\caption{Selection frequency of test combinations under different $\tau$ values}
		\label{table:cov_set}
		\renewcommand{\arraystretch}{0.8} 
		\resizebox{1\textwidth}{!}{
			\begin{tabular}{cccccccccccccc}
				\toprule
				 & $\tau$ & \textbf{$j_{11}$} & \textbf{$j_{12}$} & \textbf{$j_{13}$} & \textbf{$j_{14}$} & \textbf{$j_{15}$} & \textbf{$j_{16}$} & \textbf{$j_{21}$} & \textbf{$j_{22}$} & \textbf{$j_{23}$} & \textbf{$j_{24}$} & \textbf{$j_{25}$} & \textbf{$j_{26}$} \\
				\midrule
				\multirow{9}{*}{BQL} 
				& $\infty$ & 0.0004 & 0.0032 & 0.0244 & 0.0036 & 0.0033 & 0.9650 & 0.0046 & 0.0053 & 0.0172 & 0.0090 & 0.0081 & 0.9558 \\ 
				& 100000 & 0.0007 & 0.0007 & 0.0336 & 0.0009 & 0.0040 & 0.9601 & 0.0066 & 0.0065 & 0.0258 & 0.0100 & 0.0071 & 0.9440\\ 
				& 10000 & 0.0026 & 0.0027 & 0.1012 & 0.0027 & 0.0038 & 0.8870 & 0.0277 & 0.0284 & 0.3832 & 0.0088 & 0.0052 & 0.5486\\ 
				& 5000 & 0.0019 & 0.0029 & 0.4723 & 0.0018 & 0.0020 & 0.5190 & 0.0644 & 0.0640 & 0.7757 & 0.0036 & 0.0018 & 0.0904 \\ 
				& 2000 & 0.0061 & 0.0099 & 0.9478 & 0.0007 & 0.0017 & 0.0337 & 0.3269 & 0.1933 & 0.4796 & 0 & 0 & 0.0001\\ 
				& 1000 & 0.0450 & 0.0626 & 0.8801 & 0.0021 & 0.0003 & 0.0098 & 0.7715 & 0.1807 & 0.0478 & 0 & 0 & 0 \\ 
				& 500 & 0.5062 & 0.3535 & 0.1402 & 0 & 0 & 0 & 0.9683 & 0.0314 & 0.0003 & 0 & 0 & 0\\ 
				& 200 & 0.9751 & 0.0235 & 0.0014 & 0 & 0 & 0 & 0.9996 & 0.0004 & 0 & 0 & 0 & 0 \\ 
				& 100 & 0.9975 & 0.0025 & 0 & 0 & 0 & 0 & 1 & 0 & 0 & 0 & 0 & 0 \\ 
				\midrule
				\multicolumn{2}{c}{RQL} & 0 & 0 & 0 & 0 & 0 & 1 & 0 & 0 & 0 & 0 & 0 & 1 \\ 
				\midrule
				\multicolumn{2}{c}{HDQ} & 0 & 0 & 0 & 0 & 0 & 1 & 0 & 0 & 0.01 & 0 & 0 & 0.99 \\ 
				\bottomrule
		\end{tabular}}
		\vspace{1em}
	\end{table}

	As shown in Table \ref{table:profit}, when $\tau = \infty$, meaning that costs are not considered, the treatment regimes derived by the BQL and RQL methods yield similar utilities, both outperforming HDQ due to the more accurate non-parametric nuisance estimates. For all three considered methods, as $\tau$ decreases and cost becomes more heavily weighted, the selection frequencies of more expensive treatment assignments decline, resulting in reductions in both utility and profit, as shown in Tables \ref{table:profit}-\ref{table:trt}. Among them, the BQL method demonstrates superior performance by selecting test combinations with greater cost-effectiveness, enabling a more effective utility-cost trade-off and leading to higher overall profits, particularly when $\tau$ is relatively small. As observed in Table \ref{table:cov_set}, both the RQL and HDQ methods tend to select the full sets $j_{16}$ and $j_{26}$. This is because RQL does not involve any feature selection procedure, and although HDQ introduces regularizations, all the considered tests contain at least one covariate that is informative for the outcome. In contrast, as $\tau$ decreases, the frequencies with which the proposed BQL method selects $j_{16}$ and $j_{26}$ diminish. For intermediate values of $\tau$, BQL favors the combinations $j_{t2}$ and $j_{t3}$ for $t = 1, 2$. When $\tau$ is extremely small and costs are heavily emphasized, the BQL method concentrates on the cheapest option, $j_{t1} = \emptyset$. Notably, the combinations $j_{t4}$ and $j_{t5}$, both of which include the Arterial Blood Gas (ABG) analysis, are rarely selected across all $\tau$ values. This suggests that the ABG analysis may be less cost-effective compared to other tests, and it might be more appropriate to conduct ABG analysis only when all tests are administered.

	\section{Discussion}\label{sec:dis}
	
	This paper addresses the problem of estimating dynamic treatment regimes under cost considerations, where both treatment assignments and, more importantly, covariate assessments may incur costs. We propose a novel framework, Balanced Q-learning, to estimate the optimal dynamic treatment regime by balancing utility with the associated costs.

	Future work can proceed in several directions. A key extension involves high-dimensional settings, where estimating contrast functions between Q-functions for all covariate assessment combinations becomes computationally prohibitive. Additionally, addressing settings where observed data suffers from missingness -- an issue common in dynamic studies -- warrants further exploration. Furthermore, it would be valuable to extend the framework to infinite horizon settings under Markov assumptions and to other reinforcement learning problems.

	\section*{Supplementary Material}
	
	\textbf{Supplement to ``Balancing utility and cost in dynamic treatment regimes''.} Section \ref{sec:sm_sim} presents additional simulation results. Section \ref{sec:sm_rd} offers supplementary details on the real-data application, while Section \ref{sec:sm_target} introduces the specifics of the working models for the Q-functions and target parameters. Section \ref{sec:sm_proof} contains the proof of the main results.

	\bibliographystyle{plainnat}
	\bibliography{ref}

\begin{thebibliography}{40}
\providecommand{\natexlab}[1]{#1}
\providecommand{\url}[1]{\texttt{#1}}
\expandafter\ifx\csname urlstyle\endcsname\relax
  \providecommand{\doi}[1]{doi: #1}\else
  \providecommand{\doi}{doi: \begingroup \urlstyle{rm}\Url}\fi

\bibitem[Athey and Wager(2021)]{athey2021policy}
Susan Athey and Stefan Wager.
\newblock Policy learning with observational data.
\newblock \emph{Econometrica}, 89\penalty0 (1):\penalty0 133--161, 2021.

\bibitem[Bakker et~al.(2019)Bakker, de~Lange, Pijnappel, Mann, Peeters,
  Monninkhof, Emaus, Loo, Bisschops, Lobbes, et~al.]{bakker2019supplemental}
Marije~F Bakker, St{\'e}phanie~V de~Lange, Ruud~M Pijnappel, Ritse~M Mann,
  Petra~HM Peeters, Evelyn~M Monninkhof, Marleen~J Emaus, Claudette~E Loo,
  Robertus~HC Bisschops, Marc~BI Lobbes, et~al.
\newblock Supplemental {MRI} screening for women with extremely dense breast
  tissue.
\newblock \emph{New England Journal of Medicine}, 381\penalty0 (22):\penalty0
  2091--2102, 2019.

\bibitem[Bastani and Bayati(2020)]{bastani2020online}
Hamsa Bastani and Mohsen Bayati.
\newblock Online decision making with high-dimensional covariates.
\newblock \emph{Operations Research}, 68\penalty0 (1):\penalty0 276--294, 2020.

\bibitem[Caniglia et~al.(2016)Caniglia, Sabin, Robins, Logan, Cain, Abgrall,
  Mugavero, Hernandez-Diaz, Meyer, Seng, et~al.]{caniglia2016monitor}
Ellen~C Caniglia, Caroline Sabin, James~M Robins, Roger Logan, Lauren~E Cain,
  Sophie Abgrall, Michael~J Mugavero, Sonia Hernandez-Diaz, Laurence Meyer,
  Remonie Seng, et~al.
\newblock When to monitor {CD}4 cell count and {HIV} {RNA} to reduce mortality
  and {AIDS}-defining illness in virologically suppressed {HIV}-positive
  persons on antiretroviral therapy in high-income countries: a prospective
  observational study.
\newblock \emph{JAIDS Journal of Acquired Immune Deficiency Syndromes},
  72\penalty0 (2):\penalty0 214--221, 2016.

\bibitem[Chernozhukov et~al.(2018)Chernozhukov, Chetverikov, Demirer, Duflo,
  Hansen, Newey, and Robins]{chernozhukov2018double}
Victor Chernozhukov, Denis Chetverikov, Mert Demirer, Esther Duflo, Christian
  Hansen, Whitney Newey, and James Robins.
\newblock Double/debiased machine learning for treatment and structural
  parameters.
\newblock \emph{The Econometrics Journal}, 21\penalty0 (1):\penalty0 C1--C68,
  2018.

\bibitem[Clifton and Laber(2020)]{clifton2020q}
Jesse Clifton and Eric Laber.
\newblock Q-learning: Theory and applications.
\newblock \emph{Annual Review of Statistics and Its Application}, 7\penalty0
  (1):\penalty0 279--301, 2020.

\bibitem[Ertefaie et~al.(2021)Ertefaie, McKay, Oslin, and
  Strawderman]{ertefaie2021robust}
Ashkan Ertefaie, James~R McKay, David Oslin, and Robert~L Strawderman.
\newblock Robust {Q}-learning.
\newblock \emph{Journal of the American Statistical Association}, 116\penalty0
  (533):\penalty0 368--381, 2021.

\bibitem[Foster and Syrgkanis(2023)]{foster2023orthogonal}
Dylan~J Foster and Vasilis Syrgkanis.
\newblock Orthogonal statistical learning.
\newblock \emph{The Annals of Statistics}, 51\penalty0 (3):\penalty0 879--908,
  2023.

\bibitem[G{\"o}rgec et~al.(2024)G{\"o}rgec, Hansen, Kemmerich, Syversveen,
  Hilal, Belt, Bosscha, Burgmans, Cappendijk, D'Hondt, et~al.]{gorgec2024mri}
Burak G{\"o}rgec, Ingrid~S Hansen, Gunter Kemmerich, Trygve Syversveen,
  Mohammed~Abu Hilal, Eric~JT Belt, Koop Bosscha, Mark~C Burgmans, Vincent~C
  Cappendijk, Mathieu D'Hondt, et~al.
\newblock {MRI} in addition to {CT} in patients scheduled for local therapy of
  colorectal liver metastases ({CAMINO}): an international, multicentre,
  prospective, diagnostic accuracy trial.
\newblock \emph{The Lancet Oncology}, 25\penalty0 (1):\penalty0 137--146, 2024.

\bibitem[Hu and Kallus(2020)]{hu2020dtr}
Yichun Hu and Nathan Kallus.
\newblock {DTR} bandit: Learning to make response-adaptive decisions with low
  regret.
\newblock \emph{arXiv preprint arXiv:2005.02791}, 2020.

\bibitem[Hu et~al.(2025)Hu, Kallus, and Uehara]{hu2025fast}
Yichun Hu, Nathan Kallus, and Masatoshi Uehara.
\newblock Fast rates for the regret of offline reinforcement learning.
\newblock \emph{Mathematics of Operations Research}, 50\penalty0 (1):\penalty0
  633--655, 2025.

\bibitem[Huang and Xu(2020)]{huang2020estimating}
Xinyang Huang and Jin Xu.
\newblock Estimating individualized treatment rules with risk constraint.
\newblock \emph{Biometrics}, 76\penalty0 (4):\penalty0 1310--1318, 2020.

\bibitem[Jamal-Hanjani et~al.(2017)Jamal-Hanjani, Wilson, McGranahan, Birkbak,
  Watkins, Veeriah, Shafi, Johnson, Mitter, Rosenthal,
  et~al.]{jamal2017tracking}
Mariam Jamal-Hanjani, Gareth~A Wilson, Nicholas McGranahan, Nicolai~J Birkbak,
  Thomas~BK Watkins, Selvaraju Veeriah, Seema Shafi, Diana~H Johnson, Richard
  Mitter, Rachel Rosenthal, et~al.
\newblock Tracking the evolution of non--small-cell lung cancer.
\newblock \emph{New England Journal of Medicine}, 376\penalty0 (22):\penalty0
  2109--2121, 2017.

\bibitem[Johnson et~al.(2016)Johnson, Pollard, Shen, Lehman, Feng, Ghassemi,
  Moody, Szolovits, Anthony~Celi, and Mark]{johnson2016mimic}
Alistair~EW Johnson, Tom~J Pollard, Lu~Shen, Li-wei~H Lehman, Mengling Feng,
  Mohammad Ghassemi, Benjamin Moody, Peter Szolovits, Leo Anthony~Celi, and
  Roger~G Mark.
\newblock {MIMIC-III}, a freely accessible critical care database.
\newblock \emph{Scientific Data}, 3\penalty0 (1):\penalty0 1--9, 2016.

\bibitem[Kennedy(2023)]{kennedy2023towards}
Edward~H Kennedy.
\newblock Towards optimal doubly robust estimation of heterogeneous causal
  effects.
\newblock \emph{Electronic Journal of Statistics}, 17\penalty0 (2):\penalty0
  3008--3049, 2023.

\bibitem[Kitagawa and Tetenov(2018)]{kitagawa2018should}
Toru Kitagawa and Aleksey Tetenov.
\newblock Who should be treated? {E}mpirical welfare maximization methods for
  treatment choice.
\newblock \emph{Econometrica}, 86\penalty0 (2):\penalty0 591--616, 2018.

\bibitem[Komorowski et~al.(2018)Komorowski, Celi, Badawi, Gordon, and
  Faisal]{komorowski2018artificial}
Matthieu Komorowski, Leo~A Celi, Omar Badawi, Anthony~C Gordon, and A~Aldo
  Faisal.
\newblock The artificial intelligence clinician learns optimal treatment
  strategies for sepsis in intensive care.
\newblock \emph{Nature Medicine}, 24\penalty0 (11):\penalty0 1716--1720, 2018.

\bibitem[Kosorok and Laber(2019)]{kosorok2019precision}
Michael~R Kosorok and Eric~B Laber.
\newblock Precision medicine.
\newblock \emph{Annual Review of Statistics and Its Application}, 6\penalty0
  (1):\penalty0 263--286, 2019.

\bibitem[Kreif et~al.(2021)Kreif, Sofrygin, Schmittdiel, Adams, Grant, Zhu,
  van~der Laan, and Neugebauer]{kreif2021exploiting}
No{\'e}mi Kreif, Oleg Sofrygin, Julie~A Schmittdiel, Alyce~S Adams, Richard~W
  Grant, Zheng Zhu, Mark~J van~der Laan, and Romain Neugebauer.
\newblock Exploiting nonsystematic covariate monitoring to broaden the scope of
  evidence about the causal effects of adaptive treatment strategies.
\newblock \emph{Biometrics}, 77\penalty0 (1):\penalty0 329--342, 2021.

\bibitem[K{\"u}nzel et~al.(2019)K{\"u}nzel, Sekhon, Bickel, and
  Yu]{kunzel2019metalearners}
S{\"o}ren~R K{\"u}nzel, Jasjeet~S Sekhon, Peter~J Bickel, and Bin Yu.
\newblock Metalearners for estimating heterogeneous treatment effects using
  machine learning.
\newblock \emph{Proceedings of the National Academy of Sciences}, 116\penalty0
  (10):\penalty0 4156--4165, 2019.

\bibitem[Laber et~al.(2014)Laber, Lizotte, and Ferguson]{laber2014set}
Eric~B Laber, Daniel~J Lizotte, and Bradley Ferguson.
\newblock Set-valued dynamic treatment regimes for competing outcomes.
\newblock \emph{Biometrics}, 70\penalty0 (1):\penalty0 53--61, 2014.

\bibitem[Liu et~al.(2024)Liu, Wang, Fu, and Zeng]{liu2024controlling}
Mochuan Liu, Yuanjia Wang, Haoda Fu, and Donglin Zeng.
\newblock Controlling cumulative adverse risk in learning optimal dynamic
  treatment regimens.
\newblock \emph{Journal of the American Statistical Association}, 119\penalty0
  (548):\penalty0 2622--2633, 2024.

\bibitem[Liu et~al.(2018)Liu, Wang, Kosorok, Zhao, and Zeng]{liu2018augmented}
Ying Liu, Yuanjia Wang, Michael~R Kosorok, Yingqi Zhao, and Donglin Zeng.
\newblock Augmented outcome-weighted learning for estimating optimal dynamic
  treatment regimens.
\newblock \emph{Statistics in Medicine}, 37\penalty0 (26):\penalty0 3776--3788,
  2018.

\bibitem[Lizotte et~al.(2012)Lizotte, Bowling, and Murphy]{lizotte2012linear}
Daniel~J Lizotte, Michael Bowling, and Susan~A Murphy.
\newblock Linear fitted-{Q} iteration with multiple reward functions.
\newblock \emph{Journal of Machine Learning Research}, 13\penalty0
  (1):\penalty0 3253--3295, 2012.

\bibitem[Luckett et~al.(2021)Luckett, Laber, Kim, and
  Kosorok]{luckett2021estimation}
Daniel~J Luckett, Eric~B Laber, Siyeon Kim, and Michael~R Kosorok.
\newblock Estimation and optimization of composite outcomes.
\newblock \emph{Journal of Machine Learning Research}, 22\penalty0
  (167):\penalty0 1--40, 2021.

\bibitem[Mammen and Tsybakov(1999)]{mammen1999smooth}
Enno Mammen and Alexandre~B Tsybakov.
\newblock Smooth discrimination analysis.
\newblock \emph{The Annals of Statistics}, 27\penalty0 (6):\penalty0
  1808--1829, 1999.

\bibitem[Murphy(2003)]{murphy2003optimal}
Susan~A Murphy.
\newblock Optimal dynamic treatment regimes.
\newblock \emph{Journal of the Royal Statistical Society Series B: Statistical
  Methodology}, 65\penalty0 (2):\penalty0 331--355, 2003.

\bibitem[Neugebauer et~al.(2017)Neugebauer, Schmittdiel, Adams, Grant, and
  van~der Laan]{neugebauer2017identification}
Romain Neugebauer, Julie~A Schmittdiel, Alyce~S Adams, Richard~W Grant, and
  Mark~J van~der Laan.
\newblock Identification of the joint effect of a dynamic treatment
  intervention and a stochastic monitoring intervention under the no direct
  effect assumption.
\newblock \emph{Journal of Causal Inference}, 5\penalty0 (1):\penalty0
  20160015, 2017.

\bibitem[Nie and Wager(2021)]{nie2021quasi}
Xinkun Nie and Stefan Wager.
\newblock Quasi-oracle estimation of heterogeneous treatment effects.
\newblock \emph{Biometrika}, 108\penalty0 (2):\penalty0 299--319, 2021.

\bibitem[Robinson(1988)]{robinson1988root}
Peter~M Robinson.
\newblock Root-n-consistent semiparametric regression.
\newblock \emph{Econometrica: Journal of the Econometric Society}, pages
  931--954, 1988.

\bibitem[Schulte et~al.(2015)Schulte, Tsiatis, Laber, and
  Davidian]{schulte2015q}
Phillip~J Schulte, Anastasios~A Tsiatis, Eric~B Laber, and Marie Davidian.
\newblock Q-and {A}-learning methods for estimating optimal dynamic treatment
  regimes.
\newblock \emph{Statistical Science}, 29\penalty0 (4):\penalty0 640, 2015.

\bibitem[Shi et~al.(2018)Shi, Fan, Song, and Lu]{shi2018high}
Chengchun Shi, Alin Fan, Rui Song, and Wenbin Lu.
\newblock High-dimensional {A}-learning for optimal dynamic treatment regimes.
\newblock \emph{The Annals of Statistics}, 46\penalty0 (3):\penalty0 925, 2018.

\bibitem[Van~der Laan et~al.(2007)Van~der Laan, Polley, and
  Hubbard]{van2007super}
Mark~J Van~der Laan, Eric~C Polley, and Alan~E Hubbard.
\newblock Super learner.
\newblock \emph{Statistical Applications in Genetics and Molecular Biology},
  6\penalty0 (1), 2007.

\bibitem[Wang et~al.(2018)Wang, Fu, and Zeng]{wang2018learning}
Yuanjia Wang, Haoda Fu, and Donglin Zeng.
\newblock Learning optimal personalized treatment rules in consideration of
  benefit and risk: with an application to treating type 2 diabetes patients
  with insulin therapies.
\newblock \emph{Journal of the American Statistical Association}, 113\penalty0
  (521):\penalty0 1--13, 2018.

\bibitem[Wright et~al.(2023)Wright, Campbell, Eberhardt, Aitken, Perrett,
  Brent, Danecek, Gardner, Chundru, Lindsay, et~al.]{wright2023genomic}
Caroline~F Wright, Patrick Campbell, Ruth~Y Eberhardt, Stuart Aitken, Daniel
  Perrett, Simon Brent, Petr Danecek, Eugene~J Gardner, V~Kartik Chundru,
  Sarah~J Lindsay, et~al.
\newblock Genomic diagnosis of rare pediatric disease in the united kingdom and
  ireland.
\newblock \emph{New England Journal of Medicine}, 388\penalty0 (17):\penalty0
  1559--1571, 2023.

\bibitem[Xu et~al.(2024)Xu, Fu, and Qu]{xu2024optimal}
Qi~Xu, Haoda Fu, and Annie Qu.
\newblock Optimal individualized treatment rule for combination treatments
  under budget constraints.
\newblock \emph{Journal of the Royal Statistical Society Series B: Statistical
  Methodology}, 86\penalty0 (3):\penalty0 714--741, 2024.

\bibitem[Yu et~al.(2021)Yu, Liu, Nemati, and Yin]{yu2021reinforcement}
Chao Yu, Jiming Liu, Shamim Nemati, and Guosheng Yin.
\newblock Reinforcement learning in healthcare: A survey.
\newblock \emph{ACM Computing Surveys (CSUR)}, 55\penalty0 (1):\penalty0 1--36,
  2021.

\bibitem[Zhang et~al.(2012)Zhang, Tsiatis, Laber, and
  Davidian]{zhang2012robust}
Baqun Zhang, Anastasios~A Tsiatis, Eric~B Laber, and Marie Davidian.
\newblock A robust method for estimating optimal treatment regimes.
\newblock \emph{Biometrics}, 68\penalty0 (4):\penalty0 1010--1018, 2012.

\bibitem[Zhao et~al.(2015)Zhao, Zeng, Laber, and Kosorok]{zhao2015new}
Ying-Qi Zhao, Donglin Zeng, Eric~B Laber, and Michael~R Kosorok.
\newblock New statistical learning methods for estimating optimal dynamic
  treatment regimes.
\newblock \emph{Journal of the American Statistical Association}, 110\penalty0
  (510):\penalty0 583--598, 2015.

\bibitem[Zhu et~al.(2019)Zhu, Zeng, and Song]{zhu2019proper}
Wensheng Zhu, Donglin Zeng, and Rui Song.
\newblock Proper inference for value function in high-dimensional {Q}-learning
  for dynamic treatment regimes.
\newblock \emph{Journal of the American Statistical Association}, 114\penalty0
  (527):\penalty0 1404--1417, 2019.

\end{thebibliography}
	
	\appendix
	
	\clearpage\newpage
	
	\renewcommand\thesection{S\arabic{section}}
	\renewcommand\thesubsection{\thesection.\arabic{subsection}}
	\renewcommand{\theequation}{\thesection.\arabic{equation}}
	\renewcommand{\thetheorem}{\thesection.\arabic{theorem}}
	\renewcommand{\thelemma}{\thesection.\arabic{lemma}}
	\setcounter{figure}{0}
	\setcounter{table}{0}
	\renewcommand{\thefigure}{S\arabic{figure}}
	\renewcommand{\thetable}{S\arabic{table}}

	\par\bigskip
	
	\begin{center}
		\textbf{\uppercase{Supplementary Material to ``Balancing utility and cost in dynamic treatment regimes''}}
	\end{center}

	This supplementary document includes additional results and justifications, as well as the proofs of the main results. Specifically, Section \ref{sec:sm_sim} presents additional simulation results. Section \ref{sec:sm_rd} offers supplementary details on the real-data application, while Section \ref{sec:sm_target} introduces the specifics of the working models for the Q-functions and target parameters. Section \ref{sec:sm_proof} contains the proof of the main results. All of the results and notation are numbered as in the main text unless stated otherwise.
	
	\section{Additional simulation results}\label{sec:sm_sim}
	We illustrate the performance of considered estimators in Section \ref{sec:sim} under additional models.
	
	Model 4: Fix $\lambda = 0$ and vary $n$, while keeping all other settings same as those in Model 1.
	
	Model 5: Let $\balpha_{1}=(0.5,0.5^2,\cdots,0.5^{5})^\top$, $\balpha_{2}=(0.5,0.5^2,\cdots,0.5^{11})^\top$, $\bbeta_{1}=(1.2,0.4\cdot\mathbf{1}_{4},0.2\cdot\mathbf{1}_{4},1.5,1)^\top$, $\bbeta_{2}=(0.5,0.2\cdot\mathbf{1}_4,\mathbf{0}_6)^\top$, $\bbeta_{3}=(0.2\cdot\mathbf{1}_9,1.5,1)^\top$, $n=500$, and $p=5$. We set $l_1=[p]$, $\mathcal{J}_1=\emptyset$, $l_2=\{1\}$, and $\mathcal{J}_2=\{j_{21},j_{22},j_{23}\}$, where $j_{21}=\{2,3\}$, $j_{22}=\{2,3,4\}$, $j_{23}=\{2,3,4,5\}$. Consider $C^{1t}(a) = C^{2t}(a) = 0$ for $a \in \{0,1\}$, $C^{2c}(j_{21})=0$, $C^{2c}(j_{22})=0.1$, and $C^{2c}(j_{23})=0.2$, with a varying $\lambda$. In this model, all the elements of $\bS_2$ contribute to the heterogeneity in potential outcomes between treatment groups.
	
	Model 6. The setup is the same as in Model 1, except that now we consider $C^{1c}(j_{11})=C^{1c}(j_{12})=C^{2c}(j_{21})=C^{2c}(j_{22})=0$, $C^{1t}(a)=0$ for $a\in\{0,1\}$, $C^{2t}(0)=7.5$, and $\lambda=1$, with a varying $C^{2t}(1)$.
	
	Model 7. The setup is the same as in Model 1, except that now we consider $C^{1c}(j_{11})=C^{1c}(j_{12})=C^{2c}(j_{21})=C^{2c}(j_{22})=0$, $C^{2t}(a)=0$ for $a\in\{0,1\}$, $C^{1t}(0)=7.5$, and $\lambda=1$, with a varying $C^{1t}(1)$.
	
	\begin{figure}[ht]
		\centering
		\begin{subfigure}[b]{0.4\textwidth}
		\centering
		\includegraphics[width=\textwidth]{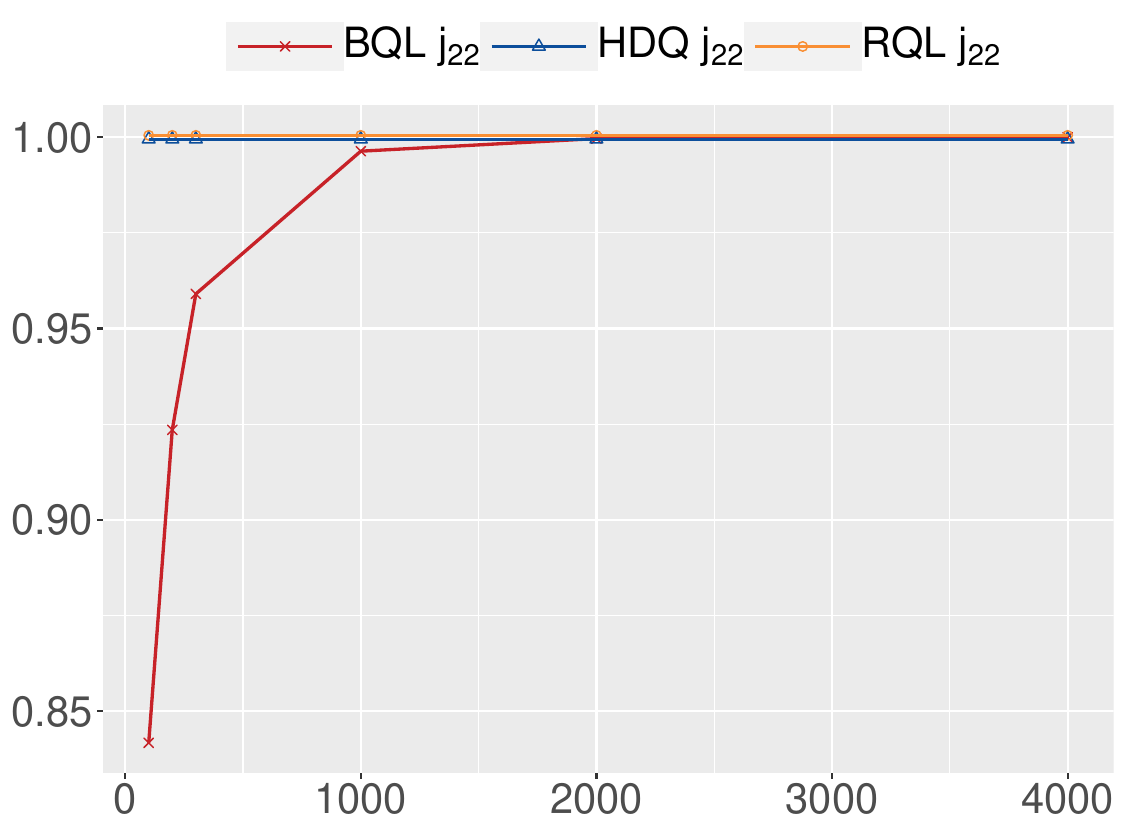}
		\captionsetup{skip=-5pt}
		\caption{Frequency as $n$ varies}
		\label{fig:M2_freq}
	\end{subfigure}
	\hspace{0.03\textwidth}
	\begin{subfigure}[b]{0.4\textwidth}
		\centering
		\includegraphics[width=\textwidth]{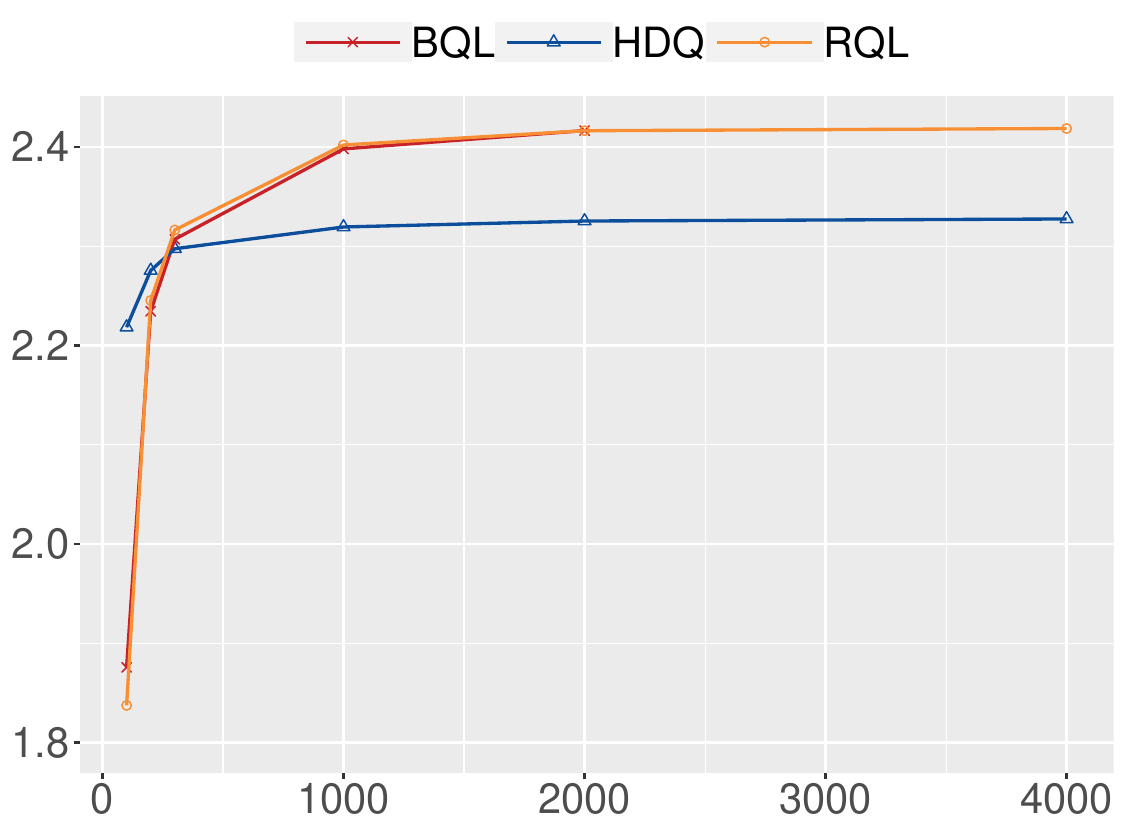}
		\captionsetup{skip=-5pt}
		\caption{Profit as $n$ varies}
		\label{fig:M2_profit}
	\end{subfigure}
	\captionsetup{skip=-5pt}
	\caption{Simulation results under Model 4}
	\label{fig:Model 4}
	\end{figure}
	
		\begin{figure}[ht]
		\centering
		\begin{subfigure}[b]{0.4\textwidth}
			\centering
			\includegraphics[width=\textwidth]{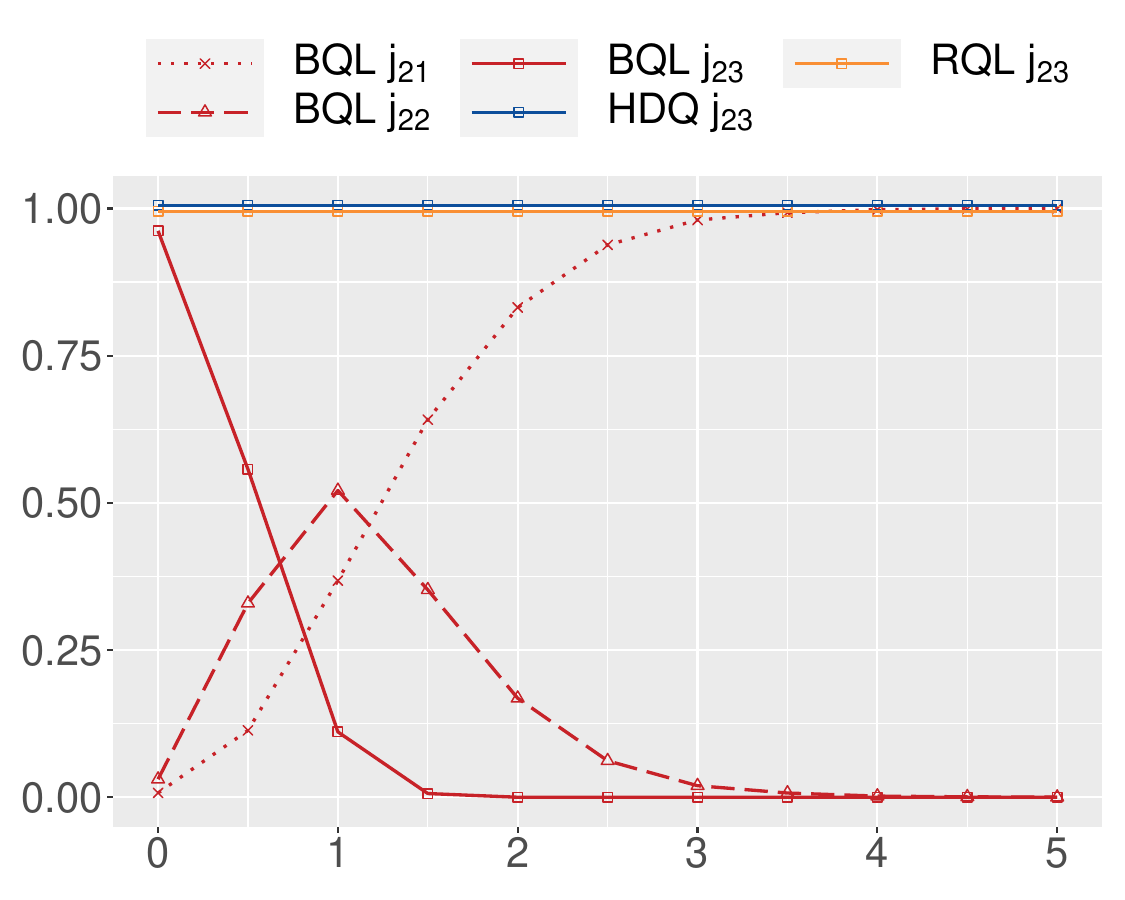}
			\captionsetup{skip=-5pt}
			\caption{Frequency as $\lambda$ varies}
			\label{fig:M5_freq}
		\end{subfigure}
		\hspace{0.03\textwidth}
		\begin{subfigure}[b]{0.4\textwidth}
			\centering
			\includegraphics[width=\textwidth]{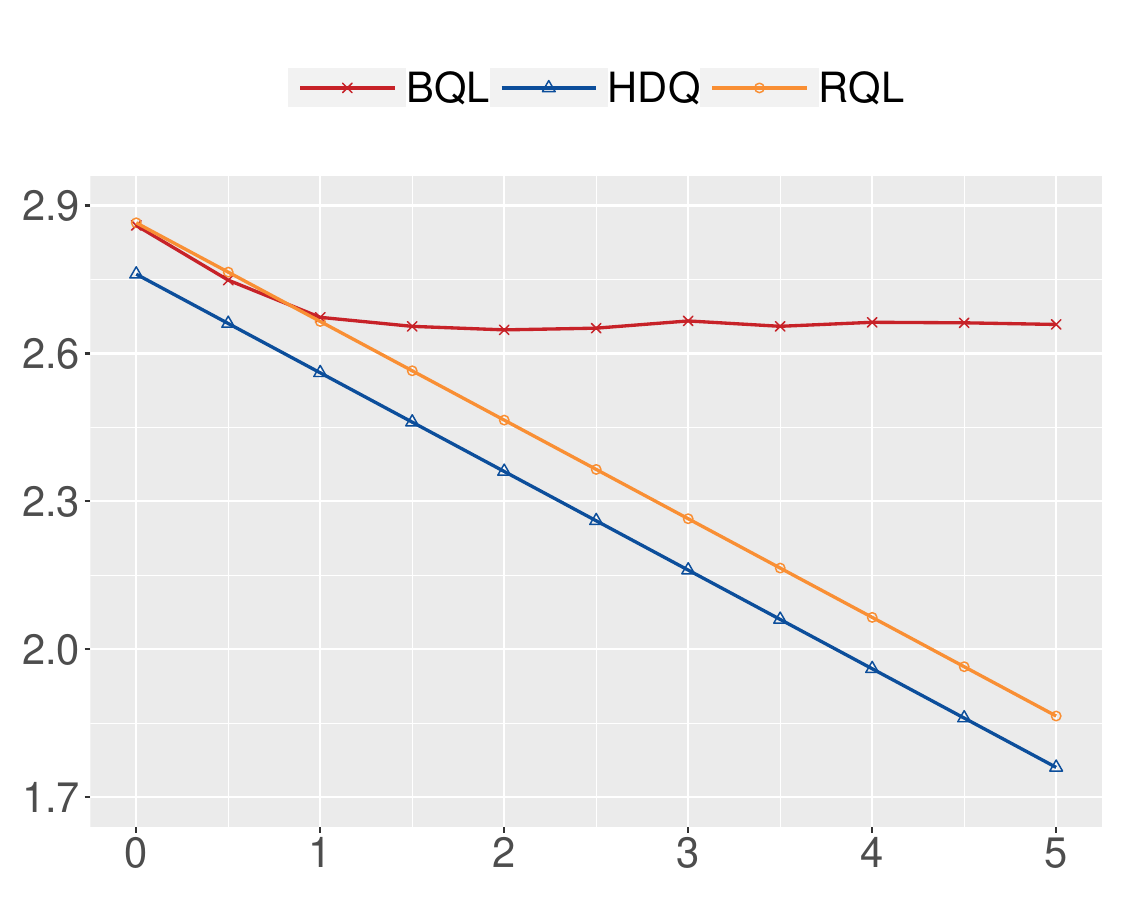}
			\captionsetup{skip=-5pt}
			\caption{Profit as $\lambda$ varies}
			\label{fig:M5_profit}
		\end{subfigure}
		\captionsetup{skip=-5pt}
		\caption{Simulation results under Model 5}
		\label{fig:Model 5}
		\end{figure}
		
		\begin{figure}[ht]
		\centering
		\begin{subfigure}[b]{0.4\textwidth}
			\centering
			\includegraphics[width=\textwidth]{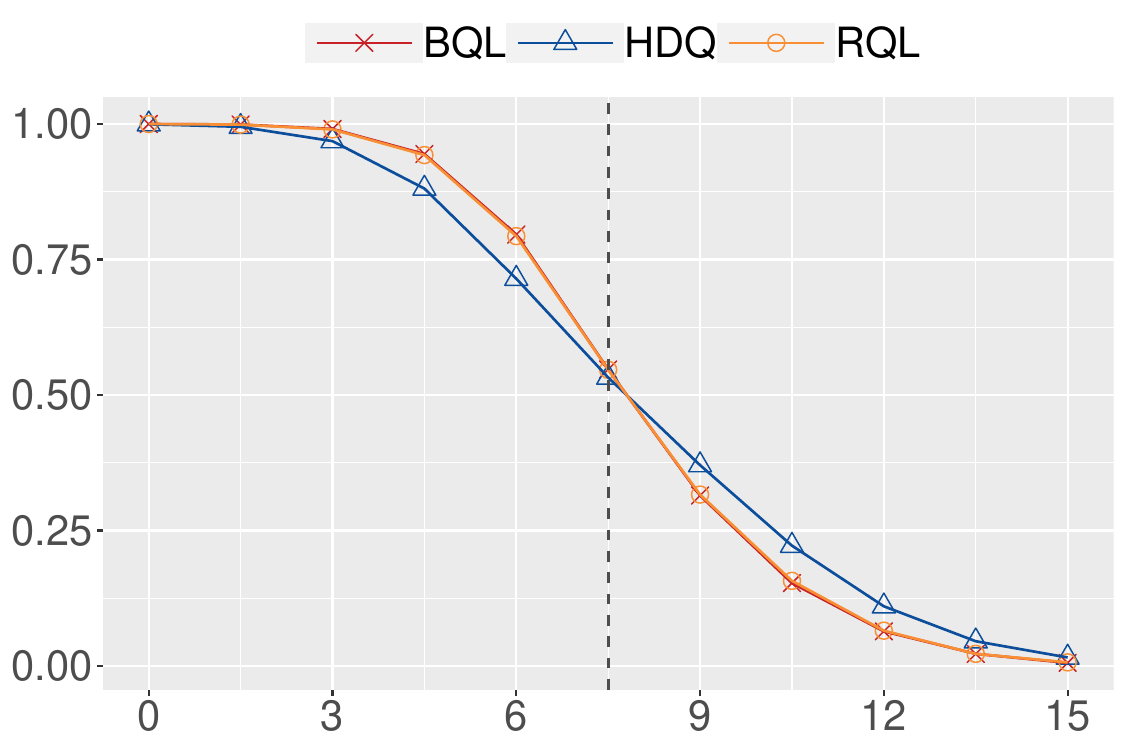}
			\captionsetup{skip=-5pt}
			\caption{Frequency as $C^{2t}(1)$ varies}
			\label{fig:M6_freq}
		\end{subfigure}
		\hspace{0.03\textwidth}
		\begin{subfigure}[b]{0.4\textwidth}
		\centering
		\includegraphics[width=\textwidth]{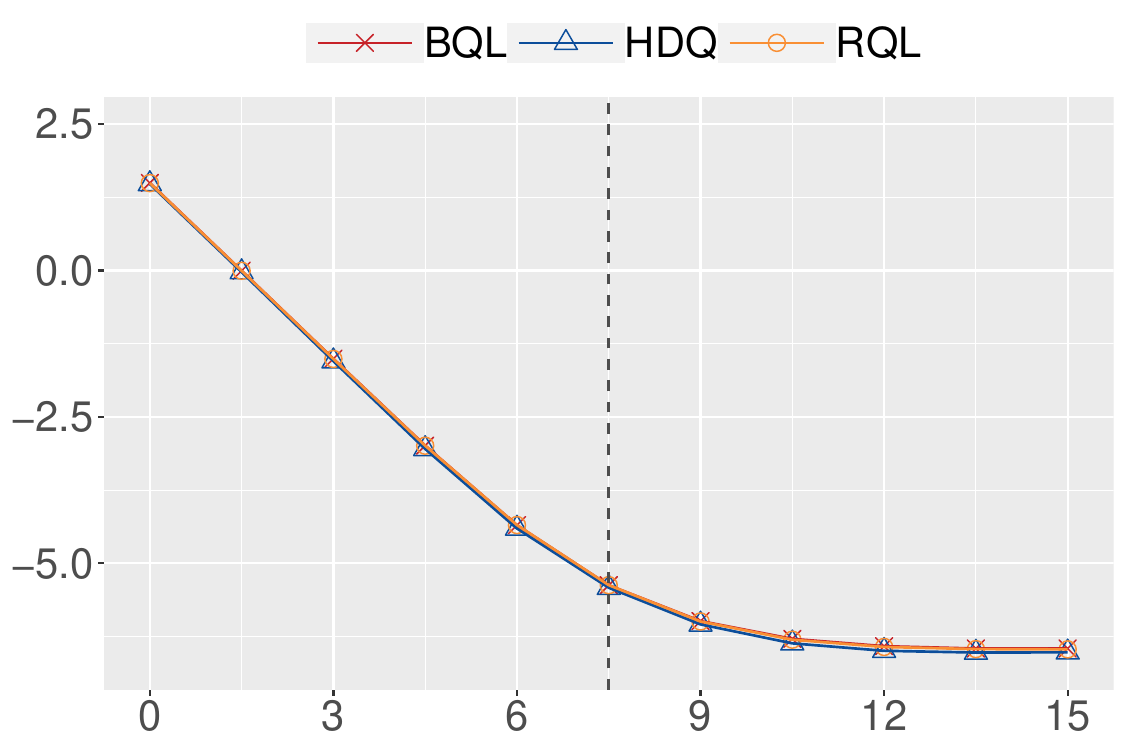}
		\captionsetup{skip=-5pt}
		\caption{Profit as $C^{2t}(1)$ varies}
		\label{fig:M6_profit}
		\end{subfigure}
		\captionsetup{skip=-5pt}
		\caption{Simulation results under Model 6}
		\label{fig:Model 6}
		\end{figure}
		
		\begin{figure}[ht]
		\centering
		\begin{subfigure}[b]{0.4\textwidth}
			\centering
			\includegraphics[width=\textwidth]{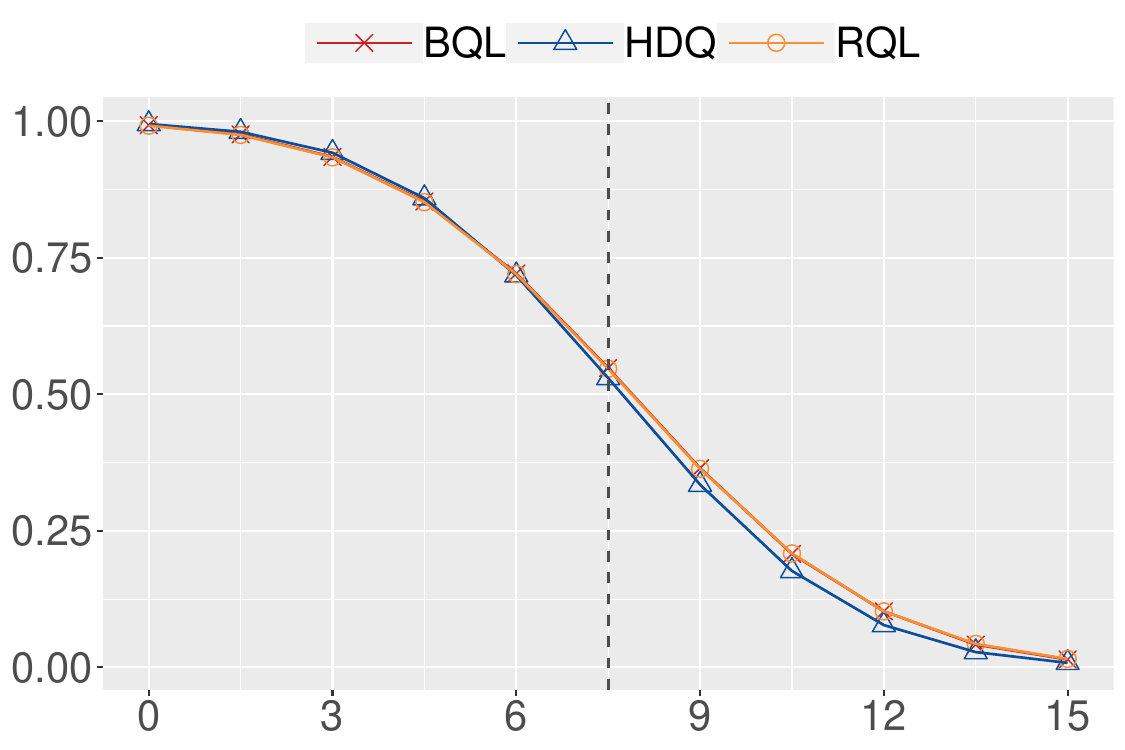}
			\captionsetup{skip=-5pt}
			\caption{Frequency as $C^{1t}(1)$ varies}
			\label{fig:M7_freq}
		\end{subfigure}
		\hspace{0.03\textwidth}
		\begin{subfigure}[b]{0.4\textwidth}
			\centering
			\includegraphics[width=\textwidth]{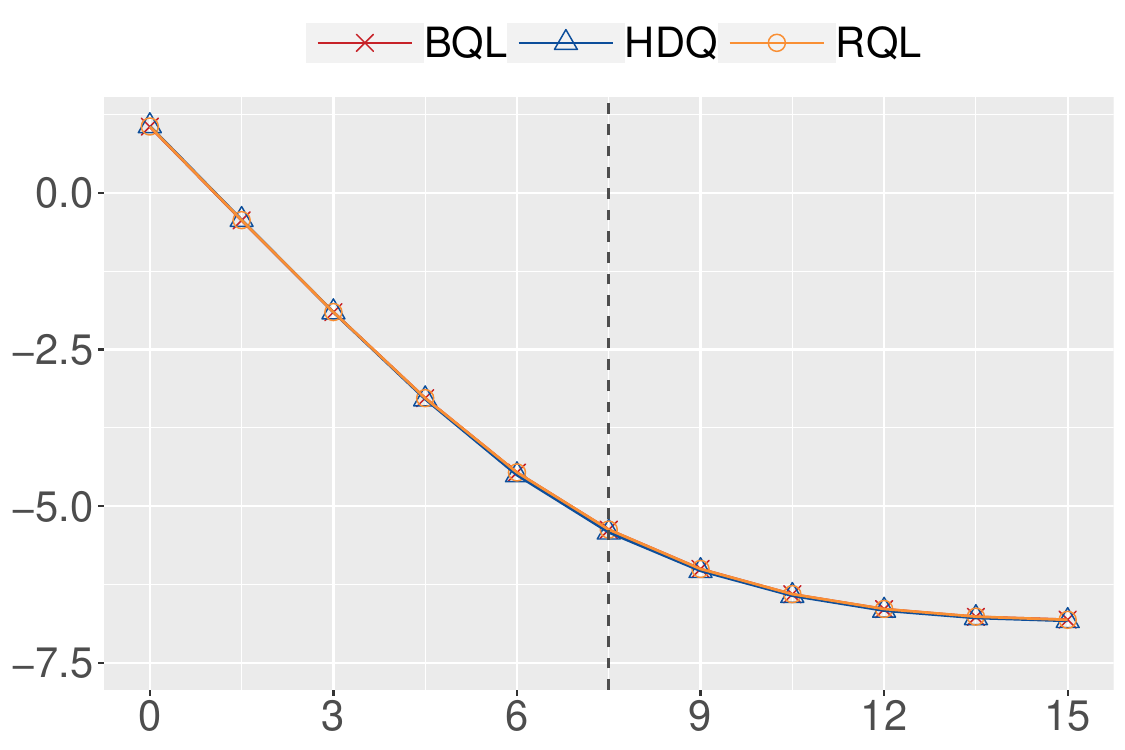}
			\captionsetup{skip=-5pt}
			\caption{Profit as $C^{1t}(1)$ varies}
			\label{fig:M7_profit}
		\end{subfigure}
		\captionsetup{skip=-5pt}
		\caption{Simulation results under Model 7}
		\label{fig:Model 7}
	\end{figure}
	
	Model 4 examines the special case where $\lambda = 0$, meaning no cost considerations are taken into account. This scenario represents the least favorable case for the BQL method, as our approach is specifically designed to address cost-related issues. Nevertheless, as shown in Figure \ref{fig:Model 4}, BQL performs similarly to the RQL method as the sample size $n$ increases. Both BQL and RQL outperform HDQ for large $n$, primarily due to the more accurate non-parametric nuisance estimates.
	
	Under Model 5, we focus on the costs associated with the second-stage covariate assessment only. Since all covariates contribute to the outcome, when $\lambda$ is small, BQL tends to select the most informative but expensive set, $j_{23} = \{2,3,4,5\}$. As $\lambda$ increases, BQL first favors the intermediate-sized set $j_{22} = \{2,3,4\}$, and later selects the smallest set, $j_{21} = \{2,3\}$. As shown in Figure \ref{fig:Model 5}, BQL and RQL deliver similar performance when $\lambda$ is small, but when $\lambda$ is relatively large, BQL clearly outperforms RQL in the sense of the overall profit. Moreover, both BQL and RQL outperform HDQ.

	Figure \ref{fig:Model 6} displays the results when the second-stage treatment cost, \(C^{2t}(0)\), is fixed at $7.5$, and \(C^{2t}(1)\) varies from $0$ to $15$. As the treatment assignment cost \(C^{2t}(1)\) increases, all considered methods show a decrease in the frequency of assigning \(A_2 = 1\) and a reduction in overall profit. Furthermore, we observe similar performance across all methods, suggesting that our approach is also adaptable to treatment assignment costs, despite our main focus being on covariate assessment costs. A similar pattern is observed in Figure \ref{fig:Model 7}, where the first-stage treatment cost is considered.

	\section{Additional details on the application to the MIMIC-III database}\label{sec:sm_rd}
	
	Table \ref{table:RD_covariate} lists the covariates and their descriptions in the considered dataset introduced in Section \ref{sec:real}. The SOFA score is calculated based on the following variables: PaO2\_FiO2, Platelets\_count, Total\_bili, MeanBP, max\_dose\_vaso, GCS, Creatinine, output\_4hourly. These variables are included in patient information, Biochemical Tests, Complete Blood Count, and Arterial Blood Gas analysis. The SIRS score is calculated based on Temp\_C, HR, RR, paCO2, and WBC\_count. These variables are included in patient information, Complete Blood Count, and Arterial Blood Gas analysis.

\begin{table}[htbp]
	\centering
	\renewcommand{\arraystretch}{0.43} 
	\small 
	\caption{Covariates and their descriptions}
	\label{table:RD_covariate}
	\begin{tabularx}{\textwidth}{lX} 
		\toprule
		\textbf{Covariate} & \textbf{Description}\\ 
		\midrule
		\multicolumn{2}{l}{\textbf{Patient Information} (dimension: 20; gender, age, re\_admission are only included in Stage 1)} \\
		gender & Gender (binary)\\
		age & Age in days\\
		re\_admission & ICU re-admission during current hospital stay (binary)\\
		Weight\_kg & Weight in kilograms \\
		mechvent & Mechanical ventilation (binary) \\
		GCS & Glasgow Coma Scale score \\
		HR & Heart rate (beats per minute) \\
		SysBP & Systolic blood pressure (mmHg) \\
		MeanBP & Mean arterial blood pressure (mmHg) \\
		DiaBP & Diastolic blood pressure (mmHg) \\
		RR & Respiratory rate (breaths per minute) \\
		Temp\_C & Temperature (Celsius) \\
		FiO2\_1 & Inspired oxygen fraction \\
		Shock\_Index & Shock index (HR/SysBP) \\
		SpO2 & Oxygen saturation (\%) \\
		max\_dose\_vaso & Maximum vasopressor dose (mcg/kg/min) \\
		input\_total & Total fluid input since hospital or ICU admission (mL) \\
		output\_total & Total urine output since hospital or ICU admission (mL) \\
		output\_4hourly & Urine output over 4-hour interval (mL) \\
		cumulated\_balance & Difference between input total and output total (mL) \\
		
		\midrule
		\multicolumn{2}{l}{\textbf{Biochemical Tests} (dimension: 11)} \\
		Potassium & Potassium (meq/L) \\
		Sodium & Sodium (meq/L) \\
		Chloride & Chloride (meq/L) \\
		Glucose & Blood glucose (mg/dL) \\
		Magnesium & Magnesium (mg/dL) \\
		Calcium & Calcium (mg/dL) \\
		BUN & Blood urea nitrogen (mg/dL) \\
		Creatinine & Creatinine (mg/dL) \\
		SGOT & Aspartate aminotransferase (U/L) \\
		SGPT & Alanine aminotransferase (U/L) \\
		Total\_bili & Total bilirubin (mg/dL) \\
		
		\midrule
		\multicolumn{2}{l}{\textbf{Complete Blood Count} (dimension: 3)} \\
		Hb & Hemoglobin (g/dL) \\
		WBC\_count & White blood cell count (E9/L) \\
		Platelets\_count & Platelet count (E9/L) \\

		\midrule
		\multicolumn{2}{l}{\textbf{Coagulation Function} (dimension: 3)} \\
		PTT & Partial thromboplastin time (seconds) \\
		PT & Prothrombin time (seconds) \\
		INR & International normalized ratio \\
		
		\midrule
		\multicolumn{2}{l}{\textbf{Arterial Blood Gas analysis} (dimension: 7)} \\
		Arterial\_pH & Arterial blood pH \\
		paO2 & Partial pressure of oxygen in arterial blood (mmHg) \\
		paCO2 & Partial pressure of carbon dioxide in arterial blood (mmHg) \\
		Arterial\_BE & Base excess in arterial blood (meq/L) \\
		Arterial\_lactate & Lactate concentration in arterial blood (mmol/L) \\
		HCO3 & Bicarbonate concentration (meq/L) \\
		PaO2\_FiO2 & Ratio of paO2 to FiO2 (mmHg) \\
		
		\midrule
		\multicolumn{2}{l}{\textbf{Other Scores} (dimension: 2)} \\
		SOFA & Sequential Organ Failure Assessment score \\
		SIRS & Systemic Inflammatory Response Syndrome score \\
		\bottomrule
	\end{tabularx}
\end{table}
	
	\newpage
	\section{Definition of the target parameters}\label{sec:sm_target}
	We begin by introducing the target parameters defined via optimal linear regimes. Let $f_{2}(\bar{\bs}_{2},a_1)=E(Y\mid \bar{\bS}_{2}=\bar{\bs}_{2},A_1=a_1)$ and $g_{2}(\bar{\bs}_{2},a_1)=E(A_2\mid \bar{\bS}_{2}=\bar{\bs}_{2},A_1=a_1)$, we consider $\bar{\Delta}^{2t*}(\bar{\bs}_{2},a_1)=(\bar{\bs}_{2},a_1)^\top\balpha^{*}$ as a population-level linear approximation of the true contrast function $\bar{\Delta}^{2t}(\bar{\bs}_{2},a_1)$, where
	\begin{align*}
	\bar\balpha^{*}=\mathop{\arg\min}_{\balpha}E\left([Y-f_{2}(\bar{\bS}_{2},A_1)-\{A_{2}-g_{2}(\bar{\bS}_{2},A_1)\}\{(\bar{\bS}_{2},A_1)^\top\balpha+C^{2t}(1)-C^{2t}(0)\}]^2\right).
	\end{align*}
	For each $a_1\in \{0,1\}$, $j_1\in\mathcal{J}_1$, and $j_2\in\mathcal{J}_2$, we introduce the following linear working model for $\Delta^{2t*}(\bs_{\bar{j}_2},j_1,a_1,j_1)$: 
	\begin{align}
	\Delta^{2t*}(\bs_{\bar{j}_2},j_1,a_1,j_2)=\bs_{\bar{j}_2}^\top\balpha_{j_1a_1j_2}^*,\quad\balpha_{j_1a_1j_2}^*=\mathop{\arg\min}_{\balpha} E\left[\{(\bar{\bS}_{2},a_1)^\top \bar\balpha^*-\bS_{\bar{j}_2}^\top\balpha\}^2\right].\label{def:alpha-star}
	\end{align}
	The corresponding linear decision rule is $\pi^{2t*}(\bs_{\bar{j}_2},j_1,a_1,j_2)=I(\bs_{\bar{j}_2}^\top\balpha_{j_1a_1j_2}^*>0)$. Define the unrestricted Q-function based on the linear working models:
	$$\bar{Q}^{2c*}(\bs_{\bar{l}_2^{\,\rm f}},j_1,a_1,j_2)=E[Q^{2t*}\{\bar{\bS}_2,a_1,\pi^{2t*}(\bS_{\bar{j}_2},j_1,a_1,j_2)\}-C^{2c}(j_2)\mid \bS_{\bar{l}_2^{\,\rm f}}=\bs_{\bar{l}_2^{\,\rm f}},A_1=a_1].$$
	Let $\bDelta^{C,2c}_{j_2}=C^{2c}(j_2)-C^{2c}(j_2^{\rm f})$, the associated contrast outcome is given by
	\begin{align*}
		Y_{j_1j_2}^{2c*}&=\bar{Q}^{2t*}\{\bar{\bS}_2,A_1,\pi^{2t*}(\bS_{\bar{J}_2},j_1,A_1,j_2)\}-\bar{Q}^{2t*}\{\bar{\bS}_2,A_1,\pi^{2t*}(\bS_{\bar{j}_2^{\rm f}},j_1,A_1,j_2^{\rm f})\}-\bDelta^{C,2c}_{j_2}\\
		&=\bar{\Delta}^{2t*}(\bar{\bS}_{2},A_1)[I\{\Delta^{2t*}(\bS_{\bar{j}_2},j_1,A_1,j_2)>0\}-I(\Delta^{2t*}(\bS_{\bar{j}_2^{\rm f}},j_1,A_1,j_2^{\rm f})>0)\}-\bDelta^{C,2c}_{j_2}.
	\end{align*} 
	To approximate the restricted contrast function $\bar{\Delta}^{2c}(\bs_{\bar{l}_2},j_1, a_1, j_2)$, we adopt a linear working model:
	\begin{align*}
		\bar{\Delta}^{2c*}(\bs_{\bar{l}_2^{\,\rm f}},j_1,a_1,j_2)&=(\bs_{\bar{l}_2^{\, \rm f}},a_1)^\top\bar{\bbeta}_{j_1j_2}^{*},\quad \bar{\bbeta}_{j_1j_2}^{*}=\mathop{\arg\min}_{\bbeta}E\{(	Y_{j_1j_2}^{2c}-(\bS_{\bar{l}_2^{\,\rm f}},A_1)^\top\bbeta)^2\}.
	\end{align*}
	Consider a linear working model for the unrestricted contrast $\Delta^{2c}(\bs_{\bar{l}_2^{\,\rm f}},j_1,a_1,j_2)$:
	\begin{align}
		\Delta^{2c*}(\bs_{\bar{l}_2},j_1,a_1,j_2)&=\bs_{\bar{l}_2}^\top\bbeta_{j_1a_1j_2}^{*}\quad	\bbeta_{j_1a_1j_2}^{*}=\mathop{\arg\min}_{\bbeta}E[\{(\bS_{\bar{l}_2^{\,\rm f}},a_1)^\top\bar{\bbeta}_{j_1j_2}^{*}-\bS_{\bar{l}_2}^\top\bbeta\}^2].
	\end{align}
	
	For the first-stage treatment assignment, we define the working Q-function as 
	\begin{align}
	\bar{Q}^{1t*}(\bs_{1},j_1,a_1)=E[\bar{Q}^{2c*}\{\bS_{\bar{l}_2^{\,\rm f}},j_1,a_1,\pi^{2c*}(\bS_{\bar{l}_2},j_1,a_1)\}-C^{1t}(a_1)\mid \bS_1=\bs_1,A_1=a_1],\label{Q1t-star}
	\end{align}
	where the linear covariate assessment rule at the second stage is given by $\pi^{2c*}(\bs_{\bar{l}_2},j_1,a_1)=\mathop{\arg\max}_{j_2\in\mathcal{J}_2}\bs_{\bar l_2}^\top\bbeta_{j_1a_1j_2}^*$. 
	
	Consider the contrast function $\bar Q^{1t*}(\bs_1,j_1,1)-\bar Q^{1t*}(\bs_1,j_1,0)$. We denote 
	$$\widetilde{Y}_{j_1}^{1t}=\bar Q^{2c*}\{\bS_{\bar{l}_2^{\,\rm f}},j_1,A_1,\pi^{2c*}(\bS_{\bar{l}_2},j_1,A_1)\}-C^{1t}(A_1).$$
	Estimating the contrast function is equivalent to estimate the following parital linear model
	$$
	\widetilde{Y}_{j_1}^{1t*}=\bar Q^{1t*}(\bS_1,j_1,0)+A_1\{\bar Q^{1t*}(\bS_1,j_1,1)-\bar Q^{1t*}(\bS_1,j_1,0)\}+\bar\epsilon,\quad E(\bar\epsilon\mid \bS_1,A_1)=0.
	$$
	We denote $\Delta^{C,2t}=C^{2t}(1)-C^{2t}(0)$ and $\Delta_{j_2}^{C,2c}=C^{2c}(j_2)-C^{2c}(j_2^{\rm f})$, note the outcome can be expressed as 
	\begin{align*}
		\widetilde{Y}_{j_1}^{1t*}
		&=E[\bar{Q}^{2t*}\{\bar{\bS}_2,A_1,\pi^{2t*}(\bS_{\bar{j}_2^{\rm f}},j_1,A_1,j_2^{\rm f})\}-C^{2c}(j_{2}^{\rm f})\mid \bS_{\bar{l}_2^{\,\rm f}},A_1]\\
		&\qquad+ \sum\nolimits_{j_2\in\mathcal{J}_2}I\{\pi^{2c*}(\bS_{\bar{l}_2},j_1,A_1)=j_2\}\bar{\Delta}^{2c*}(\bS_{\bar{l}_2^{\,\rm f}},j_1,A_1,j_2)-C^{1t}(A_1).
	\end{align*}
	Moreover, note that
	\begin{align*}
	&\bar{Q}^{2t*}\{\bar{\bS}_2,A_1,\pi^{2t*}(\bS_{\bar{j}_2^{\rm f}},j_1,A_1,j_2^{\rm f})\}\\
	&\qquad=\bar{Q}^{2t*}(\bar{\bS}_2,A_1,0)+\bar{\Delta}^{2t*}(\bar{\bS}_2,A_1)I\{\Delta^{2t*}(\bS_{\bar{j}_2^{\rm f}},j_1,A_1,j_2^{\rm f})>0\}\\
	&\qquad=f_2(\bar{\bS}_2,A_1)-C^{2t}(0)-g_2(\bar{\bS}_2,A_1)\Delta^{C,2t}\\
	&\qquad\qquad+\bar{\Delta}^{2t*}(\bar{\bS}_2,A_1)[I\{\Delta^{2t*}(\bS_{\bar{j}_2^{\rm f}},j_1,A_1,j_2^{\rm f})>0\}-g_2(\bar{\bS}_2,A_1)].
	\end{align*}
	This motivates the definition of
	\begin{align*}
		Y_{j_1}^{1t*}&=Y-C^{2t}(A_{2})+(\bar{\bS}_{2},A_1)^\top\bar\balpha^*\{I(\bS_{\bar{j_2^{\rm f}}}^\top\balpha^*_{j_1A_1j_2^{\rm f}}> 0)-A_{2}\}\nonumber\\
		&\quad-C^{2c}(j_2^{\rm f})+\sum\nolimits_{j_2\in\mathcal{J}_2}I(j_2=\mathop{\arg\max}\nolimits_{j_2^\prime }\bS_{\bar{l}_2}^\top\bbeta^*_{j_1 A_1j_2^\prime})(\bS_{\bar{l}_2^{\,\rm f}},A_{1i})^\top\bar{\bbeta}^{*}_{j_1j_2}-C^{1t}(A_{1}).
	\end{align*}
	It follows that $E(Y_{j_1}^{1t*}-\widetilde{Y}_{j_1}^{1t*}\mid \bS_1,A_1)=0$.
	The contrast can be estimated via the following model equivalently:
	$$
	Y_{j_1}^{1t*}=\bar{Q}^{1t*}(\bS_1,j_1,0)+A_1\{Q^{1t*}(\bS_1,j_1,1)-\bar{Q}^{1t*}(\bS_1,j_1,0)\}+\varepsilon^{1t*},\quad E(\varepsilon^{1t*}\mid \bS_1,A_1)=0.
	$$
	Let $g_1(\bs_1)=E(A_{1}\mid \bS_1=\bs_1)$ denote the propensity score model and
	\begin{align}\label{fj1*}
	f_{j_1}^*(\bs_1)=E(Y_{j_1}^{1t*}\mid \bS_1=\bs_1)
	\end{align} denote the outcome model. We introduce a working model for $\bar{\Delta}^{1t}(\bs_{1},j_1)$ as
	\begin{align*}
		\bar{\Delta}^{1t*}(\bs_{1},j_1)=\bs_{1}^\top\bar{\bgamma}_{j_1}^{*},\quad \bar{\bgamma}_{j_1}^{*}&=\mathop{\arg\min}_{\balpha}E\left([Y_{j_1}^{1t*}-f_{j_1}^*(\bS_{1})-\{A_{1}-g_1(\bS_{1})\}\bS_{1}^\top\bgamma]^2\right).
	\end{align*}
	Additionally, we define the nested best linear approximation for $\Delta^{1t}(\bs_{\bar{j}_1},j_1)$ as
	\begin{align}
		\Delta^{1t*}(\bs_{\bar{j}_1},j_1)=\bs_{\bar{j}_1}^\top\bgamma_{j_1}^{*},\quad \bgamma_{j_1}^{*}&=\mathop{\arg\min}_{\balpha}E\{(\bS_1^\top\bar{\bgamma}_{j_1}^{*}-\bS_{\bar{j}_1}^\top\bgamma)^2\}.
	\end{align}
	Finally, we define $Q^{1c*}(\bs_{l_1},j_1)=E[\bar{Q}^{1t*}\{\bS_1,j_1,\pi^{1t*}(\bS_{\bar{j}_1},j_1)\}-C^{1c}(j_1)\mid\bS_{l_1}=\bs_{l_1}]$, where the optimal linear treatment rule is given by $\pi^{1t*}(\bs_{\bar j_1},j_1)=I(\bs_{\bar j_1}^\top\bgamma_{j_1}^*>0)$. Now, consider the contrast $Q^{1c*}(\bs_{l_1},j_1)-Q^{1c*}(\bs_{l_1},j_1^{\rm f})$. Define 
	\begin{align*}
		\widetilde{Y}_{j_1}^{1c*}&=\bar{Q}^{1t*}\{\bS_1,j_1,\pi^{1t*}(\bS_{\bar{j}_1},j_1)\}-C^{1c}(j_1)\\
		&=\bar{Q}^{1t*}(\bS_1,j_1,0)+\bar{\Delta}^{1t*}(\bS_1,j_1)I\{\Delta^{1t*}(\bS_{\bar{J}_1},j_1)>0\}-C^{1c}(j_1)\\
		&=f_{j_1}(\bS_1)+\bar{\Delta}^{1t*}(\bS_1,j_1)[I\{\Delta^{1t*}(\bS_{\bar{j}_1},j_1)>0\}-g_1(\bS_1)]-C^{1c}(j_1),
	\end{align*}
	where the last equality is by $f_{j_1}^*(\bS_1)=\bar{Q}^{1t*}(\bS_1,j_1,0)+g_1(\bS_1)\bar \Delta^{1t*}(\bS_1,j_1)$. We introduce an equivalent unbiased representation:
		$$
	Y_{j_1}^{1c*}=Y_{j_1}^{1t*}+\bar{\Delta}^{1t*}(\bS_1,j_1)[I\{\Delta^{1t*}(\bS_{\bar{j}_1},j_1)>0\}-A_1]-C^{1c}(j_1).
	$$
	 Moreover, it holds that $E(\widetilde{Y}_{j_1}^{1c*}-Y_{j_1}^{1c*}\mid \bS_{l_1})=E\{E(\widetilde{Y}_{j_1}^{1c*}-Y_{j_1}^{1c*}\mid \bS_{1})\mid \bS_{l_1}\}=0,$
	hence $Q^{1c*}(\bs_{l_1},j_1)-Q^{1c*}(\bs_{l_1},j_1^{\rm f})=E(Y_{j_1}^{1c}-Y_{j_1^{\rm f}}^{1c}\mid \bS_{l_1})$. We specify the working model for $\Delta^{1c}(\bs_{l_1},j_1)$ as:
	\begin{align}
		\Delta^{1c*}(\bs_{l_1},j_1)=\bs_{l_1}^\top\bdelta_{j_1}^*, \quad \bdelta_{j_1}^*&=\mathop{\arg\min}_{\bdelta} E\{(Y_{j_1}^{1c*}-Y_{j_1^{\rm f}}^{1c*}-\bS_{ l_1}^\top\bdelta)^2\}.\label{def:delta-star}
	\end{align}

	\section{Proof of the main results}\label{sec:sm_proof}
	
	\begin{proof}[Proof of Theorem \ref{thm:optimal_DTR}]
	Recall that $J_1=\pi^{1c}(\bS_{l_1})$, $A_1=\pi^{1t}(\bS_{\bar{J}_1},J_1)$, $J_2=\pi^{2c}(\bS_{\bar{L}_2},J_1,\pi^{1c}(\bS_{l_1}))$, and $A_2=\pi^{2t}(\bS_{\bar{J}_2},J_1,\pi^{1t}(\bS_{\bar{J}_1},J_1),J_2)$, for simplicity, we denote $C^{1c}(\pi^{1c})=C^{1c}\{\pi^{1c}(\bS_{l_1})\}$, $C^{1t}(\pi^{1t})=C^{1t}\{\pi^{1t}(\bS_{\bar{J}_1},J_1)\}$, $C^{2c}(\pi^{2c})=C^{2c}\{\pi^{2c}(\bS_{\bar{L}_2},J_1,\pi^{1t}(\bS_{\bar{J}_1},J_1))\}$, as well as $C^{2t}(\pi^{2t})=C^{2t}\{\pi^{2t}(\bS_{\bar{J}_2},J_1,\pi^{1t}(\bS_{\bar{J}_1},J_1),J_2)\}$. Note the profit
	\begin{align*}
	&{\rm Profit}(\pi^{1c},\pi^{1t},\pi^{2c},\pi^{2t})\\
	&=E[Y\{\pi^{1t}(\bS_{\bar{J}_1},J_1),\pi^{2t}(\bS_{\bar{J}_2},J_1,\pi^{1t}(\bS_{\bar{J}_1},J_1),J_2)\}-C^{1c}(\pi^{1c})-C^{1t}(\pi^{1t})-C^{2c}(\pi^{2c})-C^{2t}(\pi^{2t})]\\
	&=E[E[Y\{\pi^{1t}(\bS_{\bar{J}_1},J_1),\pi^{2t}(\bS_{\bar{J}_2},J_1,\pi^{1t}(\bS_{\bar{J}_1},J_1),J_2)\}\\
	&\qquad\qquad-C^{1c}(\pi^{1c})-C^{1t}(\pi^{1t})-C^{2c}(\pi^{2c})-C^{2t}(\pi^{2t})\mid\bS_1]]\\
	&=E[E[Y\{\pi^{1t}(\bS_{\bar{J}_1},J_1),\pi^{2t}(\bS_{\bar{J}_2},J_1,\pi^{1t}(\bS_{\bar{J}_1},J_1),J_2)\}\\
	&\qquad\qquad-C^{1c}(\pi^{1c})-C^{1t}(\pi^{1t})-C^{2c}(\pi^{2c})-C^{2t}(\pi^{2t})\mid\bS_1,A_1=\pi^{1t}(\bS_{\bar{J}_1},J_1)]]\\
	&=E[E[E[E[Y\{\pi^{1t}(\bS_{\bar{J}_1},J_1),\pi^{2t}(\bS_{\bar{J}_2},J_1,\pi^{1t}(\bS_{\bar{J}_1},J_1),J_2)\}-C^{2t}(\pi^{2t})\mid \bar{\bS}_2,A_1=\pi^{1t}(\bS_{\bar{J}_1},J_1)]\\
	&\qquad\qquad-C^{2c}(\pi^{2c})\mid \bS_{\bar{l}_2^{\,\rm f}},A_1=\pi^{1t}(\bS_{\bar{J}_1},J_1)]-C^{1t}(\pi^{1t})-C^{1c}(\pi^{1c})\mid \bS_1,A_1=\pi^{1t}(\bS_{\bar{J}_1},J_1)]]\\
	&=E[E[E[E[Y\{\pi^{1t}(\bS_{\bar{J}_1},J_1),\pi^{2t}(\bS_{\bar{J}_2},J_1,\pi^{1t}(\bS_{\bar{J}_1},J_1),J_2)\}-C^{2t}(\pi^{2t})\\
	&\qquad\qquad\mid \bar{\bS}_2,A_1=\pi^{1t}(\bS_{\bar{J}_1},J_1),A_2=\pi^{2t}(\bS_{\bar{J}_2},J_1,\pi^{1t}(\bS_{\bar{J}_1},J_1),J_2)]-C^{2c}(\pi^{2c})\\
	&\qquad\qquad\mid \bS_{\bar{l}_2^{\,\rm f}},A_1=\pi^{1t}(\bS_{\bar{J}_1},J_1)]-C^{1t}(\pi^{1t})-C^{1c}(\pi^{1c})\mid \bS_1,A_1=\pi^{1t}(\bS_{\bar{J}_1},J_1)]],
	\end{align*}
	where the second and fourth equalities are by the tower rule, the third equality holds, since we observe that
	\begin{align*}
	\pi^{2t}(\bS_{\bar{J}_2},J_1,\pi^{1t}(\bS_{\bar{J}_1},J_1),J_2)&=\pi^{2t}(\bS_{\bar{J}_2}(\pi^{1t}(\bS_{\bar{J}_1},J_1)),J_1,\pi^{1t}(\bS_{\bar{J}_1},J_1),J_2),\\
	\pi^{2c}(\bS_{\bar{L}_2},J_1,\pi^{1c}(\bS_{l_1}))&=\pi^{2c}(\bS_{\bar{L}_2}(\pi^{1t}(\bS_{\bar{J}_1},J_1)),J_1,\pi^{1c}(\bS_{l_1})),
	\end{align*}
	where 
	\begin{align*}
	\bS_{\bar{L}_2}(\pi^{1t}(\bS_{\bar{J}_1},J_1))&=(\bS_{\bar{J}_1},\bS_{L_2}(\pi^{1t}(\bS_{\bar{J}_1},J_1)),\\ \bS_{\bar{J}_2}(\pi^{1t}(\bS_{\bar{J}_1},J_1))&=(\bS_{\bar{L}_2}(\pi^{1t}(\bS_{\bar{J}_1},J_1)),\bS_{J_2}(\pi^{1t}(\bS_{\bar{J}_1},J_1))),
	\end{align*}
 	and $\bS_{L_2}(\pi^{1t}(\bS_{\bar{J}_1}, J_1))$ and $\bS_{J_2}(\pi^{1t}(\bS_{\bar{J}_1}, J_1))$ denote the potential values of the second-stage covariates, restricted to sets $L_2$ and $J_2$, respectively, under the first-stage treatment $\pi^{1t}(\bS_{\bar{J}_1}, J_1)$. Hence the random variable inside the conditional expectation is equal to
	\begin{align*}
	&Y\{\pi^{1t}(\bS_{\bar{J}_1},J_1),\pi^{2t}(\bS_{\bar{J}_2}(\pi^{1t}(\bS_{\bar{J}_1},J_1)),J_1,\pi^{1t}(\bS_{\bar{J}_1},J_1),J_2)\}-C^{1c}\{\pi^{1c}(\bS_{l_1})\}-C^{1t}\{\pi^{1t}(\bS_{\bar{J}_1},J_1)\}\\
	&-C^{2c}\{\pi^{2c}(\bS_{\bar{L}_2}(\pi^{1t}(\bS_{\bar{J}_1},J_1)),J_1,\pi^{1t}(\bS_{\bar{J}_1},J_1))\}-C^{2t}\{\pi^{2t}(\bS_{\bar{J}_2}(\pi^{1t}(\bS_{\bar{J}_1},J_1)),J_1,\pi^{1t}(\bS_{\bar{J}_1},J_1),J_2)\},
	\end{align*}
	which is independent of $A_1$ conditional on $\bS_1$ by sequential
	ignorability of Assumption \ref{ass:identification},  the fifth equality is also by sequential
	ignorability of Assumption \ref{ass:identification}. By consistency of Assumption \ref{ass:identification}, the above term is equal to 
	\begin{align*}
		&E[E[E[E[Y-C^{2t}(\pi^{2t})\mid \bar{\bS}_2,A_1=\pi^{1t}(\bS_{\bar{J}_1},J_1),A_2=\pi^{2t}(\bS_{\bar{J}_2},J_1,\pi^{1t}(\bS_{\bar{J}_1},J_1),J_2)]-C^{2c}(\pi^{2c})\\
	&\qquad\qquad\mid \bS_{\bar{l}_2^{\,\rm f}},A_1=\pi^{1t}(\bS_{\bar{J}_1},J_1)]-C^{1t}(\pi^{1t})-C^{1c}(\pi^{1c})\mid \bS_1,A_1=\pi^{1t}(\bS_{\bar{J}_1},J_1)]]\\
	&=E[E[E[E[E[Y-C^{2t}(\pi^{2t})\mid \bar{\bS}_2,A_1=\pi^{1t}(\bS_{\bar{J}_1},J_1),A_2=\pi^{2t}(\bS_{\bar{J}_2},J_1,\pi^{1t}(\bS_{\bar{J}_1},J_1),J_2)]-C^{2c}(\pi^{2c})\\
	&\qquad\qquad\mid \bS_{\bar{l}_2^{\,\rm f}},A_1=\pi^{1t}(\bS_{\bar{J}_1},J_1)]-C^{1t}(\pi^{1t})\mid \bS_1,A_1=\pi^{1t}(\bS_{\bar{J}_1},J_1)]-C^{1c}(\pi^{1c})\mid \bS_{l_1}]].
	\end{align*}
	
	By the definition of $\bar{Q}^{2t}$, for any dynamic treatment regime $\{\pi^{1c},\pi^{1t},\pi^{2c},\pi^{2t}\}$, the expectation of the corresponding profit can be expressed as 
	\begin{align*}
		&{\rm Profit}(\pi^{1c},\pi^{1t},\pi^{2c},\pi^{2t})\\
		&=E[E[E[E[\bar{Q}^{2t}\{\bar{\bS}_2,\pi^{1t}(\bS_{\bar{J}_1},J_1),\pi^{2t}(\bS_{\bar{J}_2},J_1,\pi^{1t}(\bS_{\bar{J}_1},J_1),J_2)\}-C^{2c}(\pi^{2c})\\
		&\qquad\qquad\mid \bS_{\bar{l}_2^{\,\rm f}},A_1=\pi^{1t}(\bS_{\bar{J}_1},J_1)]-C^{1t}(\pi^{1t})\mid \bS_1,A_1=\pi^{1t}(\bS_{\bar{J}_1},J_1)]-C^{1c}(\pi^{1c})\mid \bS_{l_1}]]\\
		&=E[E[E[E[\bar{Q}^{2t}\{\bar{\bS}_2,\pi^{1t}(\bS_{\bar{J}_1},J_1),\check{\pi}^{2t}(\bS_{\bar{J}_2},J_1,\pi^{1t}(\bS_{\bar{J}_1},J_1),J_2)\}-C^{2c}(\pi^{2c})\\
		&\qquad\qquad\mid \bS_{\bar{l}_2^{\,\rm f}},A_1=\pi^{1t}(\bS_{\bar{J}_1},J_1)]-C^{1t}(\pi^{1t})\mid \bS_1,A_1=\pi^{1t}(\bS_{\bar{J}_1},J_1)]-C^{1c}(\pi^{1c})\mid \bS_{l_1}]]-G^{2t}\\
		&=E[E[E[\bar{Q}^{2c}\{\bS_{\bar{l}_2^{\,\rm f}},J_1,\pi^{1t}(\bS_{\bar{J}_1},J_1),\pi^{2c}(\bS_{\bar{L}_2},J_1,\pi^{1t}(\bS_{\bar{J}_1},J_1))\}-C^{1t}(\pi^{1t})\\
		&\qquad\qquad\mid \bS_1,A_1=\pi^{1t}(\bS_{\bar{J}_1},J_1)]-C^{1c}(\pi^{1c})\mid \bS_{l_1}]]-G^{2t}\\
		&=E[E[E[\bar{Q}^{2c}\{\bS_{\bar{l}_2^{\,\rm f}},J_1,\pi^{1t}(\bS_{\bar{J}_1},J_1),\check{\pi}^{2c}(\bS_{\bar{L}_2},J_1,\pi^{1t}(\bS_{\bar{J}_1},J_1))\}-C^{1t}(\pi^{1t})\\
		&\qquad\qquad\mid \bS_1,A_1=\pi^{1t}(\bS_{\bar{J}_1},J_1)]-C^{1c}(\pi^{1c})\mid \bS_{l_1}]]-G^{2t}-C^{2c}\\
		&=E[E[\bar{Q}^{1t}\{\bS_1,J_1,\pi^{1t}(\bS_{\bar{J}_1},J_1)\}-C^{1c}(\pi^{1c})\mid \bS_{l_1}]]-G^{2t}-G^{2c}\\
		&=E[E[\bar{Q}^{1t}\{\bS_1,J_1,\check{\pi}^{1t}(\bS_{\bar{J}_1},J_1)\}-C^{1c}(\pi^{1c})\mid \bS_{l_1}]]-G^{2t}-G^{2c}-G^{1t}\\
		&=E[Q^{1c}\{\bS_{l_1},\pi^{1c}(\bS_{l_1})\}],
	\end{align*}
	where
	\begin{align*}
		G^{2t}&=E[E[E[E[\bar{Q}^{2t}\{\bar{\bS}_2,\pi^{1t}(\bS_{\bar{J}_1},J_1),\check{\pi}^{2t}(\bS_{\bar{J}_2},J_1,\pi^{1t}(\bS_{\bar{J}_1},J_1),J_2)\}\\
		&\qquad\qquad\qquad-\bar{Q}^{2t}\{\bar{\bS}_2,\pi^{1t}(\bS_{\bar{J}_1},J_1),\pi^{2t}(\bS_{\bar{J}_2},J_1,\pi^{1t}(\bS_{\bar{J}_1},J_1),J_2)\}\\
		&\qquad\qquad\qquad\mid \bS_{\bar{l}_2^{\,\rm f}},A_1=\pi^{1t}(\bS_{\bar{J}_1},J_1)]\mid \bS_1,A_1=\pi^{1t}(\bS_{\bar{J}_1},J_1)]\mid \bS_{l_1}]],\\
		G^{2c}&=E[E[E[\bar{Q}^{2c}\{\bS_{\bar{l}_2^{\,\rm f}},J_1,\pi^{1t}(\bS_{\bar{J}_1},J_1),\check{\pi}^{2c}(\bS_{\bar{L}_2},J_1,\pi^{1t}(\bS_{\bar{J}_1},J_1))\}\\
		&\qquad\qquad\qquad-\bar{Q}^{2c}\{\bS_{\bar{l}_2^{\,\rm f}},J_1,\pi^{1t}(\bS_{\bar{J}_1},J_1),\pi^{2c}(\bS_{\bar{L}_2},J_1,\pi^{1t}(\bS_{\bar{J}_1},J_1))\}\\
		&\qquad\qquad\qquad\mid \bS_1,A_1=\pi^{1t}(\bS_{\bar{J}_1},J_1)]\mid \bS_{l_1}]],\\
		G^{1t}&=E[E[\bar{Q}^{1t}\{\bS_1,J_1,\check{\pi}^{1t}(\bS_{\bar{J}_1},J_1)\}-\bar{Q}^{1t}\{\bS_1,J_1,\pi^{1t}(\bS_{\bar{J}_1},J_1)\}\mid \bS_{l_1}]].
	\end{align*}
		We obtain 
	\begin{align*}
		G^{2t}
		&=E[E[\bar{Q}^{2t}\{\bar{\bS}_2,\pi^{1t}(\bS_{\bar{J}_1},J_1),\check{\pi}^{2t}(\bS_{\bar{J}_2},J_1,\pi^{1t}(\bS_{\bar{J}_1},J_1),J_2)\}\\
		&\qquad\qquad\qquad-\bar{Q}^{2t}\{\bar{\bS}_2,\pi^{1t}(\bS_{\bar{J}_1},J_1),\pi^{2t}(\bS_{\bar{J}_2},J_1,\pi^{1t}(\bS_{\bar{J}_1},J_1),J_2)\}\mid \bS_1,A_1=\pi^{1t}(\bS_{\bar{\pi}^{1c}})]]\\
		&		=E[E[\bar{Q}^{2t}\{\bar{\bS}_2,\pi^{1t}(\bS_{\bar{J}_1},J_1),\check{\pi}^{2t}(\bS_{\bar{J}_2},J_1,\pi^{1t}(\bS_{\bar{J}_1},J_1),J_2)\}\\
		&\qquad\qquad\qquad-\bar{Q}^{2t}\{\bar{\bS}_2,\pi^{1t}(\bS_{\bar{J}_1},J_1),\pi^{2t}(\bS_{\bar{J}_2},J_1,\pi^{1t}(\bS_{\bar{J}_1},J_1),J_2)\}\mid \bS_1]]\\
		&		=E[E[\bar{Q}^{2t}\{\bar{\bS}_2,\pi^{1t}(\bS_{\bar{J}_1},J_1),\check{\pi}^{2t}(\bS_{\bar{J}_2},J_1,\pi^{1t}(\bS_{\bar{J}_1},J_1),J_2)\}\\
		&\qquad\qquad\qquad-\bar{Q}^{2t}\{\bar{\bS}_2,\pi^{1t}(\bS_{\bar{J}_1},J_1),\pi^{2t}(\bS_{\bar{J}_2},J_1,\pi^{1t}(\bS_{\bar{J}_1},J_1),J_2)\}\mid \bS_{\bar{J}_2}]]\\
		&		=E[Q^{2t}\{\bS_{\bar{J}_2},J_1,\pi^{1t}(\bS_{\bar{J}_1},J_1),\check{\pi}^{2t}(\bS_{\bar{J}_2},J_1,\pi^{1t}(\bS_{\bar{J}_1},J_1),J_2)\}\\
		&\qquad\qquad\qquad-Q^{2t}\{\bS_{\bar{J}_2},J_1,\pi^{1t}(\bS_{\bar{J}_1},J_1),\pi^{2t}(\bS_{\bar{J}_2},J_1,\pi^{1t}(\bS_{\bar{J}_1},J_1),J_2)\}]\ge 0,
	\end{align*}
	where the first and third are by tower rule, the second equality holds since conditional on $\bS_1$, the random variable inside the conditional expectation is a function
	measurable with respect to $\bS_2(\pi^{1t}(\bS_{\bar{J}_1},J_1))$ and is independent of $A_1$ by sequential ignorability of Assumption \ref{ass:identification}, the last inequality is by the definition of $\check{\pi}^{2t}$. For the second term
	\begin{align*}
	G^{2c}&=E[E[\bar{Q}^{2c}\{\bS_{\bar{l}_2^{\,\rm f}},J_1,\pi^{1t}(\bS_{\bar{J}_1},J_1),\check{\pi}^{2c}(\bS_{\bar{L}_2},J_1,\pi^{1t}(\bS_{\bar{J}_1},J_1))\}\\
	&\qquad\qquad\qquad-\bar{Q}^{2c}\{\bS_{\bar{l}_2^{\,\rm f}},J_1,\pi^{1t}(\bS_{\bar{J}_1},J_1),\pi^{2c}(\bS_{\bar{L}_2},J_1,\pi^{1t}(\bS_{\bar{J}_1},J_1))\}\mid \bS_1,A_1=\pi^{1t}(\bS_{\bar{J}_1},J_1)]]\\
	&=E[E[\bar{Q}^{2c}\{\bS_{\bar{l}_2^{\,\rm f}},J_1,\pi^{1t}(\bS_{\bar{J}_1},J_1),\check{\pi}^{2c}(\bS_{\bar{L}_2},J_1,\pi^{1t}(\bS_{\bar{J}_1},J_1))\}\\
	&\qquad\qquad\qquad-\bar{Q}^{2c}\{\bS_{\bar{l}_2^{\,\rm f}},J_1,\pi^{1t}(\bS_{\bar{J}_1},J_1),\pi^{2c}(\bS_{\bar{L}_2},J_1,\pi^{1t}(\bS_{\bar{J}_1},J_1))\}\mid \bS_1]]\\
	&=E[E[\bar{Q}^{2c}\{\bS_{\bar{l}_2^{\,\rm f}},J_1,\pi^{1t}(\bS_{\bar{J}_1},J_1),\check{\pi}^{2c}(\bS_{\bar{L}_2},J_1,\pi^{1t}(\bS_{\bar{J}_1},J_1))\}\\
	&\qquad\qquad\qquad-\bar{Q}^{2c}\{\bS_{\bar{l}_2^{\,\rm f}},J_1,\pi^{1t}(\bS_{\bar{J}_1},J_1),\pi^{2c}(\bS_{\bar{L}_2},J_1,\pi^{1t}(\bS_{\bar{J}_1},J_1))\}\mid \bS_{L_2}]]\\
	&=E[Q^{2c}\{\bS_{\bar{L}_2},J_1,\pi^{1t}(\bS_{\bar{J}_1},J_1),\check{\pi}^{2c}(\bS_{\bar{L}_2},J_1,\pi^{1t}(\bS_{\bar{J}_1},J_1))\}\\
	&\qquad\qquad\qquad-Q^{2c}\{\bS_{\bar{L}_2},J_1,\pi^{1t}(\bS_{\bar{J}_1},J_1),\pi^{2c}(\bS_{\bar{L}_2},J_1,\pi^{1t}(\bS_{\bar{J}_1},J_1))\}]\ge 0,
	\end{align*}
	where the first and third equalities are by the tower rule, the second equality holds, since condtional on $\bS_1$, the random variable inside the conditional expectation is a function
	measurable with respect to $\bS_{\bar{l}_2^{\rm f}}(\pi^{1t}(\bS_{\bar{\pi}^{1c}}))$ and is independent of $A_1$ by sequential ignorability of Assumption \ref{ass:identification}. By tower rule and the definition of $\check{\pi}^{1t}$,
	\begin{align*}
		G^{1t}
		&=E[\bar{Q}^{1t}\{\bS_1,J_1,\check{\pi}^{1t}(\bS_{\bar{J}_1},J_1)\}-\bar{Q}^{1t}\{\bS_1,J_1,\pi^{1t}(\bS_{\bar{J}_1},J_1)\}]\\
		&=E[E[\bar{Q}^{1t}\{\bS_1,J_1,\check{\pi}^{1t}(\bS_{\bar{J}_1},J_1)\}-\bar{Q}^{1t}\{\bS_1,J_1,\pi^{1t}(\bS_{\bar{J}_1},J_1)\}\mid \bS_{\bar{J}_1}]]\\
		&=E[Q^{1t}\{\bS_{\bar{J}_1},J_1,\check{\pi}^{1t}(\bS_{\bar{J}_1},J_1)\}-Q^{1t}\{\bS_{\bar{J}_1},J_1,\pi^{1t}(\bS_{\bar{J}_1},J_1)\}]\ge 0.
	\end{align*}

	On the other hand, for the regime $\{\check\pi^{1c},\check\pi^{1t},\check\pi^{2c},\check\pi^{2t}\}$, we use $\check{J}_1=\check{\pi}^{1c}(\bS_{l_1})$ and $\check{J}_2=\check{\pi}^{2c}(\bar{\bS}_{L_2},\check{J}_1,\check{\pi}^{1c}(\bS_{l_1}))$ denote the index set selected by the regime at the first and second stage. Then the corresponding expected profit is 
	\begin{align*}
		&{\rm Profit}(\check\pi^{1c},\check\pi^{1t},\check\pi^{2c},\check\pi^{2t})\\
		&=E[E[E[E[\bar{Q}^{2t}\{\bar{\bS}_2,\check{\pi}^{1t}(\bS_{\bar{\check{J}}_1},\check{J}_1),\check{\pi}^{2t}(\bS_{\bar{\check{J}}_2},\check{J_1},\check{\pi}^{1t}(\bS_{\bar{\check{J}}_1},\check{J}_1),\check{J}_2)\}-C^{2c}(\check{\pi}^{2c})\\
		&\qquad\qquad\mid \bS_{\bar{l}_2^{\,\rm f}},A_1=\check{\pi}^{1t}(\bS_{\bar{\check{J}}_1},\check{J}_1)]-C^{1t}(\check{\pi}^{1t})\mid \bS_1,A_1=\check{\pi}^{1t}(\bS_{\bar{\check{J}}_1},\check{J}_1)]-C^{1c}(\check{\pi}^{1c})\mid \bS_{l_1}]]\\
		&=E[E[E[\bar{Q}^{2c}\{\bS_{\bar{l}_2^{\,\rm f}},\check{J}_1,\check{\pi}^{1t}(\bS_{\bar{\check{J}}_1},\check{J}_1),\check{\pi}^{2c}(\bar{\bS}_{L_2},\check{J}_1,\check{\pi}^{1t}(\bS_{\bar{\check{J}}_1},\check{J}_1))\}-C^{1t}(\check{\pi}^{1t})\\
		&\qquad\qquad\mid \bS_1,A_1=\check{\pi}^{1t}(\bS_{\bar{\check{J}}_1},\check{J}_1)]-C^{1c}(\check{\pi}^{1c})\mid \bS_{l_1}]]\\
		&=E[E[\bar{Q}\{\bS_1,\check{J}_1,\check{\pi}^{1t}(\bS_{\bar{\check{J}}_1},\check{J}_1)\}-C^{1c}(\check{\pi}^{1c})\mid \bS_{l_1}]]\\
		&=E[Q^{1c}\{\bS_{l_1},\check{\pi}^{1c}(\bS_{l_1})\}].
	\end{align*}
	Consequently, by the definition of $\check{\pi}^{1c}$,
	\begin{align*}
		&{\rm Profit}(\check\pi^{1c},\check\pi^{1t},\check\pi^{2c},\check\pi^{2t})-{\rm Profit}(\pi^{1c},\pi^{1t},\pi^{2c},\pi^{2t})\\
		=\ &E[Q^{1c}\{\bS_{l_1},\check{\pi}^{1c}(\bS_{l_1})\}-Q^{1c}\{\bS_{l_1},\pi^{1c}(\bS_{l_1})\}]+G^{1t}+G^{2c}+G^{2t}\ge 0,
	\end{align*}
	which shows the optimality of $\{\check\pi^{1c},\check\pi^{1t},\check\pi^{2c},\check\pi^{2t}\}$.
	\end{proof}
	
	The following theorem provides the full version of Theorem \ref{thm:asymptotic_normality}, including the explicit form of the asymptotic variance.
	\begin{theorem}\label{thm:asymptotic_normality2}
		Let Assumptions \ref{ass:uni_bound_}-\ref{ass:margin2} hold with some $r>1$. Let $\bar{\bX}_2=(\bar{\bS}_{2},A_1)$, $\bX_{\bar{l}_2^{\,\rm f}}=(\bS_{\bar{l}_2^{\,\rm f}},A_1)$, and define
		\begin{align*}
		E_{j_1j_2}&=I(\bS_{\bar{j}_2}^\top\balpha^*_{j_1A_1j_2}> 0)-I(\bS_{\bar{j}_{2}^{\rm f}}^\top\balpha^*_{j_1A_1j_{2}^{\rm f}}> 0),\;\; \widetilde{E}_{j_1j_{2}}=I(\bS_{\bar{j}_2}^\top\balpha^*_{j_1A_{1}j_2}> 0)-A_{2},\\	\mathcal{B}_{j_1j_2}^*&=\{\bS_{\bar{l}_2}^\top\bbeta^*_{j_1A_1j_2}> \bS_{\bar{l}_2}^\top\bbeta^*_{j_1A_1j_2^\prime}, j_2^{\prime}\neq j_2\}.
		\end{align*}
		Then, for each $j_1\in\mathcal J_1$, $a_1\in\{0,1\}$, and $j_2\in\mathcal J_2$, as $n\to\infty$, the following results hold:
		
		(a) $\sqrt{n}(\widehat{\balpha}-\bar\balpha^*)\xrightarrow{\rm d}\mathcal{N}(0,\bM^{\bar\balpha})$ and $\sqrt{n}(\widehat{\balpha}_{j_1a_1j_2}-\balpha_{j_1a_1j_2}^*)\xrightarrow{\rm d}\mathcal{N}(0,\bM_{j_1a_1j_2}^{\balpha})$, where
		\begin{align*}
		\bM^{\bar\balpha}&=(\bV^{\balpha})^{-1}
		E(\bQ^{\bar\balpha}\bQ^{\bar\balpha\top})
		(\bV^{\balpha})^{-1},\;\; \bV^{\bar\balpha}=E\{{\rm var}(A_{2}\mid \bar{\bX}_{2})\bar{\bX}_{2}\bar{\bX}_{2}^\top\},\\
		\bQ^{\bar\balpha}&=\{A_{2}-g_{2}(\bar{\bX}_{2})\}\bar{\bX}_{2}[Y-f_{2}(\bar{\bX}_{2})-\{A_{2}-g_{2}(\bar{\bX}_{2})\}\{\bar{\bX}_{2}^\top\bar\balpha^*+C^{2t}(1)-C^{2t}(0)\}],\\
			\bM_{j_1a_1j_2}^{\balpha}&=E(\bS_{\bar{j}_2}\bS_{\bar{j}_2}^\top)^{-1}E(\bQ_{j_1a_1j_2}^{\balpha}\bQ_{j_1a_1j_2}^{\balpha\top})E(\bS_{\bar{j}_2}\bS_{\bar{j}_2}^\top)^{-1},\\
			\bQ_{j_1a_1j_2}^{\balpha}&=\bS_{\bar{j}_2}\{(\bar{\bS}_{2},a_1)^\top\bar\balpha^*-\bS_{\bar{J}_2}^\top\balpha^*_{j_1a_1j_2}\}+E\{\bS_{\bar{j}_2}(\bar{\bS}_{2i},a_1)\}(\bV^{\bar\balpha})^{-1}\bQ^{\bar\balpha};
		\end{align*}
	
	(b) $\sqrt{n}(\widetilde{\bbeta}_{j_1j_2}-\bar\bbeta^{*}_{j_1j_2})\xrightarrow{\rm d}\mathcal{N}(0,\bM_{j_1j_2}^{\bar\bbeta})$ and $\sqrt{n}(\widehat{\bbeta}_{j_1a_1j_2}-\bbeta_{j_1a_1j_2}^*)\xrightarrow{\rm d}\mathcal{N}(0,\bM_{j_1a_1j_2}^{\bbeta})$, where
	\begin{align*}
		\bM_{j_1j_2}^{\bar\bbeta}&=E(\bX_{\bar{l}_2^{\,\rm f}}\bX_{\bar{l}_2^{\,\rm f}}^\top)^{-1}E(\bQ^{\bar\bbeta}_{j_1j_2}\bQ^{\bar\bbeta\top}_{j_1j_2})E(\bX_{\bar{l}_2^{\,\rm f}}\bX_{\bar{l}_2^{\,\rm f}}^\top)^{-1},\\
		\bQ^{\bar\bbeta}_{j_1j_2}&=\bX_{\bar{l}_2^{\,\rm f}}(Y_{j_1j_2}^{2c*}-\bX_{\bar{l}_2^{\,\rm f}}^\top\bbeta_{j_1j_2}^*)+E[E_{j_1j_2}\bX_{\bar{l}_2^{\,\rm f}}\bar{\bX}_2](\bV^{\bar\balpha})^{-1}\bQ^{\bar\balpha},\\
		\bM_{j_1a_1j_2}^{\bbeta}&=E(\bS_{\bar{l}_2}\bS_{\bar{l}_2}^\top)^{-1}E(\bQ_{j_1a_1j_2}^{\bbeta}\bQ_{j_1a_1j_2}^{\bbeta\top})E(\bS_{\bar{l}_2}\bS_{\bar{l}_2}^\top)^{-1},\\
		\bQ_{j_1a_1j_2}^{\bbeta}&=\bS_{\bar{l}_2}\{(\bS_{\bar{l}_2^{\,\rm f}},a_1)^\top\bar\bbeta_{j_1j_2}^{*}-\bS_{\bar{l}_2}^\top \bbeta_{j_1a_1j_2}^*\}+E\{\bS_{\bar{l}_2}(\bS_{\bar{l}_2^{\,\rm f}i},a_1)\}E(\bX_{\bar{l}_2^{\,\rm f}}\bX_{\bar{l}_2^{\,\rm f}}^\top)^{-1}\bQ^{\bar\bbeta}_{j_1j_2};
	\end{align*}
	
	(c) $\sqrt{n}(\widetilde{\bgamma}_{j_1}-\bar\bgamma_{j_1}^*)\xrightarrow{\rm d}\mathcal{N}(0,\bM_{j_1}^{\bar\bgamma})$ and $\sqrt{n}(\widehat{\bgamma}_{j_1}-\bgamma_{j_1}^{*})\xrightarrow{\rm d}\mathcal{N}(0,\bM_{j_1}^{\bgamma})$, where
	\begin{align*}
		\bM_{j_1}^{\bar\bgamma}&=(\bV^{\bar\bgamma})^{-1}E(\bQ_{j_1}^{\bar\bgamma}\bQ_{j_1}^{\bar\bgamma\top})(\bV^{\bar\bgamma})^{-1},\;\; \bV^{\bar\bgamma}=E\{{\rm var}(A_{1}\mid \bS_{1})\bS_{1}\bS_{1}^\top\},\\
		\bQ_{j_1}^{\bar\bgamma}&=\{A_{1}-g_{1}(\bS_{1})\}\bS_{1}[Y_{j_1}^{1t*}-f_{j_1}(\bS_1)-\{A_{1}-g_{1}(\bS_{1})\}\bS_{1}^\top\bar\bgamma_{j_1}^{*}]\\
		&\quad+E[\{A_{1}-g_{1}(\bS_{1})\} \widetilde{E}_{j_1j_{2}^{\rm f}}\bS_{1}\bar{\bX}_{2}^\top](\bV^{\bar\balpha})^{-1}\bQ^{\bar\balpha}\\
		&\quad+\sum\nolimits_{j_2\in\mathcal{J}_2}E\left[\{A_{1}-g_{1}(\bS_{1})\}I(\mathcal{B}_{j_1j_2}^*)\bS_{1}\bX_{\bar{l}_2^{\,\rm f}}^\top\right]E(\bX_{\bar{l}_2^{\,\rm f}}\bX_{\bar{l}_2^{\,\rm f}}^\top)^{-1}\bQ_{j_1j_2}^{\bar\bbeta},\\
		\bM_{j_1}^{\bgamma}&=E(\bS_{\bar{j}_1}\bS_{\bar{j}_1}^\top)^{-1}E(\bQ_{j_1}^{\bgamma}\bQ_{j_1}^{\bgamma\top})E(\bS_{\bar{j}_1}\bS_{\bar{j}_1}^\top)^{-1},\\ \bQ_{j_1}^{\bgamma}&=\bS_{\bar{j}_1}(\bS_1^\top\bar\bgamma^{*}_{j_1}-\bS_{\bar{j}_1i}^\top\bgamma^*_{j_1})+E\{\bS_{\bar{j}_1}\bS_{1}^\top\}(\bV^{\bar\bgamma})^{-1}\bQ^{\bar\gamma};
	\end{align*}
	
	(d) $\sqrt{n}(\widehat{\bdelta}_{j_1}-\bdelta^*_{j_1})\xrightarrow{\rm d}\mathcal{N}(0,\bM_{j_1}^{\bdelta})$, where
	\begin{align*}
		\bM_{j_1}^{\bdelta}&=E(\bS_{l_1}\bS_{l_1}^\top)^{-1}E(\bQ^{\bdelta}_{j_1}\bQ^{\bdelta\top}_{j_1})E(\bS_{l_1}\bS_{l_1}^\top)^{-1},\\
		\bQ^{\bdelta}_{j_1}&=\bS_{l_1i}(Y_{j_1}^{1c*}-Y_{j_1^{\rm f}}^{1c*}-\bS_{l_1}^\top\bdelta_{j_1}^*)+E\{(\widetilde{E}_{j_1j_2^{\rm f}}-\widetilde{E}_{j_1^{\rm f}j_2^{\rm f}})\bS_{l_1}\bar{\bX}_2^\top\}(\bV^{\bar\balpha})^{-1}\bQ^{\bar\balpha}\\
		&\qquad +\sum\nolimits_{j_2\in\mathcal{J}_2}E\left[I(\mathcal{B}_{j_1j_2}^*)\bS_{l_1}\bX_{\bar{l}_2^{\,\rm f}}^\top\right]E(\bX_{\bar{l}_2^{\,\rm f}}\bX_{\bar{l}_2^{\,\rm f}}^\top)^{-1}\bQ_{j_1j_2}^{\bar\bbeta}\\
		&\qquad-\sum\nolimits_{j_2\in\mathcal{J}_2}E\left[I(\mathcal{B}_{j_1^{\rm f}j_2}^*)\bS_{l_1}\bX_{\bar{l}_2^{\,\rm f}}^\top\right]E(\bX_{\bar{l}_2^{\,\rm f}}\bX_{\bar{l}_2^{\,\rm f}}^\top)^{-1}\bQ_{j_1^{\rm f}j_2}^{\bar\bbeta}\\
		&\qquad+E[\{I(\bS_{\bar{j}_1}^\top\bgamma_{j_1}^*>0)-A_1\}\bS_{l_1}\bS_1^\top](\bV^{\bar\bgamma})^{-1}\bQ_{j_1}^{\bar\bgamma}\\
		&\qquad-E[\{I(\bS_{1}^\top\bgamma_{j_1^{\rm f}}^*>0)-A_1\}\bS_{l_1}\bS_1^\top](\bV^{\bar\bgamma})^{-1}\bQ_{j_1^{\rm f}}^{\bar\bgamma}.
	\end{align*}
	
	Moreover, under Assumptions \ref{ass:uni_bound_}-\ref{ass:margin2}, the following results hold as long as $r\geq1$:
	\begin{align*}
	&\|\widehat{\balpha}-\bar\balpha^*\|_2=O_p(1/\sqrt{n}),\;\; \|\widehat{\balpha}_{j_1a_1j_2}-\balpha_{j_1a_1j_2}^*\|_2=O_p(1/\sqrt{n}),\\
	&\|\widetilde{\bbeta}_{j_1j_2}-\bar\bbeta_{j_1j_2}^{*}\|_2=O_p(1/\sqrt{n}),\;\; \|\widehat{\bbeta}_{j_1a_1j_2}-\bbeta_{j_1a_1j_2}^*\|_2=O_p(1/\sqrt{n}),\\
	& \|\widetilde{\bgamma}_{j_1}-\bar\bgamma_{j_1}^{*}\|_2=O_p(1/\sqrt{n}),\;\;\|\widehat{\bgamma}_{j_1}-\bgamma_{j_1}^*\|_2=O_p(1/\sqrt{n}),\;\;\|\widehat{\bdelta}_{j_1}-\bdelta_{j_1}^*\|_2=O_p(1/\sqrt{n}).
	\end{align*} 
	\end{theorem}
	
	Before proving Theorem \ref{thm:asymptotic_normality2}, we first introduce the following lemma.
	
	\begin{lemma}[Lemma 2 of \cite{ertefaie2021robust}]\label{lemma:matrix_norm}
		Let $\bV_n$ and $\widehat{\bV}_n$ be two sequences of square matrices and let $\|\cdot\|$ be any proper matrix norm. Suppose there exists $n_0<\infty$ such that (i) $\bV_n^{-1}$ and $\bV_n$ exist with $0<C_{V1}\le \|\bV_n^{-1}\|\le C_{V2}<\infty$ and (ii) $\|\widehat{\bV}_n-\bV_n\|\le (2\|\bV_n^{-1}\|)^{-1}$. Then $$\|\widehat{\bV}_n-\bV_n^{-1}\|\le 2C_{V2}^2\|\widehat{\bV}_n-\bV_n\|.$$
	\end{lemma}

	\begin{proof}[Proof of Theorem \ref{thm:asymptotic_normality2}]
	In the following, we establish the asymptotic normality of the estimators $\widetilde{\bbeta}_{j_1j_2}$, $\widehat{\bbeta}_{j_1a_1j_2}$, $\widetilde{\bgamma}_{j_1}$, $\widehat{\bgamma}_{j_1}$, and $\widehat{\bdelta}_{j_1}$ under Assumption \ref{ass:margin2} with $r>1$, and also prove their consistency under Assumption \ref{ass:margin2} with $r\ge 1$.
	
	\noindent\textbf{Step 1: Asymptotic Normality of $\widehat{\balpha}$ and $\widehat{\balpha}_{j_1a_1j_2}$} 
	
	We first prove the asymptotic normality of $\widehat{\balpha}$, denote $\bar{\bX}_2=(\bar{\bS}_2,A_1)$ and $\Delta^{C,2t}=C^{2t}(1)-C^{2t}(0)$, since $Y_i-f_2(\bar{\bX}_{2i})=\{A_{2i}-g_2(\bar{\bX}_{2i})\}\{\Delta^{2t}(\bar{\bX}_{2i})+\Delta^{C,2t}\}+\varepsilon_{2ti}$, the parameter of interest $\balpha^{*}$ can be expressed as 
	\begin{align*}
	\bar\balpha^*&=\mathop{\arg\min}_{\balpha}E\left([Y_i-f_{2}(\bar{\bX}_{2i})-\{A_{2i}-g_{2}(\bar{\bX}_{2i})\}\{\bar{\bX}_{2i}^\top\balpha+\Delta^{C,2t}\}]^2\right)\\
	&=\mathop{\arg\min}_{\balpha} E\left[\{A_{2i}-g_{2}(\bar{\bX}_{2i})\}^2\{\bar\Delta^{2t}(\bar{\bX}_{2i})-\bar{\bX}_{2i}^\top\balpha\}^2\right].
	\end{align*}
	We define
	\begin{align*}
		\breve{\balpha}&=\mathop{\arg\min}_{\balpha} n^{-1}\sum_{k\in [K]}\sum_{i\in\mathcal{I}_k}\{A_{2i}-\widehat{g}_{2}^{-k}(\bar{\bX}_{2i})\}^2\{\bar\Delta^{2t}(\bar{\bX}_{2i})-\bar{\bX}_{2i}^\top\balpha\}^2\\
		\check{\balpha}&=\mathop{\arg\min}_{\balpha} n^{-1}\sum_{i\in [n]}\{A_{2i}-g_2(\bar{\bX}_{2i})\}^2\{\bar\Delta^{2t}(\bar{\bX}_{2i})-\bar{\bX}_{2i}^\top\balpha\}^2.
	\end{align*}
	We next control the terms $\sqrt{n}(\widetilde{\balpha}-\breve{\balpha})$, $\sqrt{n}(\breve{\balpha}-\check{\balpha})$, as well as $\sqrt{n}(\check{\balpha}-\bar\balpha^{*})$ in order to bound $\sqrt{n}(\widetilde{\balpha}-\bar\balpha^*)$. We denote 
	\begin{align*}
		\widehat{\bV}=n^{-1}\sum_{k\in[K]}\sum_{i\in\mathcal{I}_k}\{A_{2i}-\widehat{g}_{2}^{-k}(\bar{\bX}_{2i})\}^2\bar{\bX}_{2i}\bar{\bX}_{2i}^\top,\;\; 	\bV=n^{-1}\sum_{i\in[n]}\{A_{2i}-g_{2}(\bar{\bX}_{2i})\}^2\bar{\bX}_{2i}\bar{\bX}_{2i}^\top.
	\end{align*}
	\textbf{Control $\sqrt{n}(\widetilde{\balpha}-\breve{\balpha})$.} Note that
	\begin{align*}
		\sqrt{n}(\widetilde{\balpha}-\breve{\balpha})&=\widehat{\bV}^{-1}\dfrac{1}{\sqrt{n}}\sum_{k}\sum_{i\in\mathcal{I}_k}\bar{\bX}_{2i}\{A_{2i}-\widehat{g}_{2}^{-k}(\bar{\bX}_{2i})\}\\
		&\qquad\times\left[Y_i-\widehat{f}_{2}^{-k}(\bar{\bX}_{2i})-\{A_{2i}-\widehat{g}_{2}^{-k}(\bar{\bX}_{2i})\}\{\bar{\Delta}^{2t}(\bar{\bX}_{2i})+\Delta^{C,2t}\}\right]\\
		&=\widehat{\bV}^{-1}(B_1+B_2+\cdots+B_6),
	\end{align*}
	where
	\begin{align*}
		B_{1}&=(1/\sqrt{n})\sum\nolimits_{i\in[n]}\bar{\bX}_{2i}\{A_{2i}-g_{2}(\bar{\bX}_{2i})\}\varepsilon_{2ti},\\
		B_{2}&=(1/\sqrt{n})\sum\nolimits_k\sum\nolimits_{i\in\mathcal{I}_k}\bar{\bX}_{2i}\{A_{2i}-g_{2}(\bar{\bX}_{2i})\}\{f_{2}(\bar{\bX}_{2i})-\widehat{f}_{2}^{-k}(\bar{\bX}_{2i})\},\\
		B_{3}&=(1/\sqrt{n})\sum\nolimits_k\sum\nolimits_{i\in\mathcal{I}_k}\bar{\bX}_{2i}\{A_{2i}-g_{2}(\bar{\bX}_{2i})\}\{\widehat{g}_{2}^{-k}(\bar{\bX}_{2i})-g_{2}(\bar{\bX}_{2i})\}\{\bar\Delta^{2t}(\bar{\bX}_{2i})+\Delta^{C,2t}\},\\
		B_{4}&=(1/\sqrt{n})\sum\nolimits_k\sum\nolimits_{i\in\mathcal{I}_k}\bar{\bX}_{2i}\{g_{2}(\bar{\bX}_{2i})-\widehat{g}_{2}^{-k}(\bar{\bX}_{2i})\}\varepsilon_{2ti},\\
		B_{5}&=(1/\sqrt{n})\sum\nolimits_k\sum\nolimits_{i\in\mathcal{I}_k}\bar{\bX}_{2i}\{g_{2}(\bar{\bX}_{2i})-\widehat{g}_{2}^{-k}(\bar{\bX}_{2i})\}\{f_{2}(\bar{\bX}_{2i})-\widehat{f}_{2}^{-k}(\bar{\bX}_{2i})\},\\
		B_{6}&=-(1/\sqrt{n})\sum\nolimits_k\sum\nolimits_{i\in\mathcal{I}_k}\bar{\bX}_{2i}\{\widehat{g}_{2}^{-k}(\bar{\bX}_{2i})-g_{2}(\bar{\bX}_{2i})\}^2\{\bar\Delta^{2t}(\bar{\bX}_{2i})+\Delta^{C,2t}\}.
	\end{align*}
	We will show that $B_2,B_3,\cdots,B_6$ are all $o_p(1)$ terms. For eack $k\in [K]$. let
	$$
	B_{2k}=|\mathcal{I}_k|^{-1}\sum_{i\in\mathcal{I}_k}\bar{\bX}_{2i}\{A_{2i}-g_{2}(\bar{\bX}_{2i})\}\{f_{2}(\bar{\bX}_{2i})-\widehat{f}_{2}^{-k}(\bar{\bX}_{2i})\}.
	$$
	We denote the dimension of $\bar{\bX}_{2i}$ as $\bar{d}$, for each $o\in [\bar{d}]$, consider the $o$-th element of $B_{2k}$, by Assumption \ref{ass:uni_bound_}, we obtain
	\begin{align*}
		E(B_{2ko}\mid \mathcal{D}_{-k})&=E\{B_{2ko}\mid \mathcal{D}_{-k},(\bar{\bX}_{2i})_{i\in\mathcal{I}_k}\}=0,\\
		{\rm Var}(B_{2ko}\mid \mathcal{D}_{-k})&=O\left\{n^{-1}\cdot \|g_{2}(\cdot)-\widehat{g}_{2}^{-k}(\cdot)\|_{P,2}^2\right\}.
	\end{align*}
	Together with Assumption \ref{ass:nuisance1_} and Chebyshev’s inequality, we obtain $B_{2ko}=O_p(n^{-1/2}\cdot \|g_{2}(\cdot)-\widehat{g}_{2}^{-k}(\cdot)\|_{P,2})=o_p(n^{-1/2})$. Hence,
	$$
	B_{2}=\sqrt{n}\sum\nolimits_{k}\dfrac{|\mathcal{I}_k|}{n}B_{2k}\le \sqrt{n}\sum\nolimits_k B_{2k}=O\left\{ \|g_{2}(\cdot)-\widehat{g}_{2}^{-k}(\cdot)\|_{P,2}\right\}=o_p(1).
	$$
	By Assumption \ref{ass:uni_bound_}, $\bar\Delta^{2t}(\bar{\bX}_{2i})+\Delta^{C,2t}$ is uniformly bounded. Using the same argument as for $B_2$, we can similarly show that
	\begin{align*}	
		B_{3}&=O\left\{\|\widehat{f}_{2}^{-k}(\cdot)-f_{2}(\cdot)\|_{P,2}\right\}=o_p(1),\\ B_{4}&=O\left\{\|\widehat{g}_{2}^{-k}(\cdot)-g_{2}(\cdot)\|_{P,2}\right\}=o_p(1).
	\end{align*}
	For each $k\in [K]$, let 
	\begin{align*}
		B_{5k}=|\mathcal{I}_k|^{-1}\sum_{i\in\mathcal{I}_k}\bar{\bX}_{2i}\{g_{2}(\bar{\bX}_{2i})-\widehat{g}_{2}^{-k}(\bar{\bX}_{2i})\}\{f_{2}(\bar{\bX}_{2i})-\widehat{f}_{2}^{-k}(\bar{\bX}_{2i})\}.
	\end{align*}
	Consider the $o$-th element of $B_{5k}$, it follows from Assumption \ref{ass:uni_bound_} and Cauchy-Schwarz inequality that
	\begin{align*}
		B_{5ko}&= O\Bigg(\bigg[|\mathcal{I}_k|^{-1}\sum_{i\in\mathcal{I}_k}\{g_{2}(\bar{\bX}_{2i})-\widehat{g}_{2}^{-k}(\bar{\bX}_{2i})\}^2\bigg]^{1/2} \bigg[|\mathcal{I}_k|^{-1}\sum_{i\in\mathcal{I}_k}\{f_{2}(\bar{\bX}_{2i})-\widehat{f}_{2}^{-k}(\bar{\bX}_{2i})\}^2\bigg]^{1/2}\Bigg)\\
		&=O_p\left\{\|\widehat{g}_{2}^{-k}(\cdot)-g_{2}(\cdot)\|_{P,2}\cdot\|\widehat{f}_{2}^{-k}(\cdot)-f_{2}(\cdot)\|_{P,2}\right\}=o_p(n^{-1/2}).
	\end{align*}
	Hence, 
	$$
	B_{5}=\sqrt{n}\sum\nolimits_{k}\dfrac{|\mathcal{I}_k|}{n}B_{5k}\le \sqrt{n}\sum\nolimits_k B_{5k}=o_p(1).
	$$
	Using a similar argument as above, we can also control
	$$
	B_6=O_p\left\{\sqrt{n}\cdot\|\widehat{g}_{2}^{-k}(\cdot)-g_{2}(\cdot)\|_{P,2}^2\right\}=o_p(1).
	$$
	We next bound $\widehat{\bV}^{-1}-\bV^{-1}$. For any $o_1,o_2\in [\bar{d}]$, the difference of the $(o1,o2)$-th element of $\widehat{\bV}-\bV$ 
	\begin{align*}
		\widehat{\bV}_{o_1o_2}-\bV_{o_1o_2}&=n^{-1}\sum_{k}\sum_{i\in\mathcal{I}_k}\{\widehat{g}_{2}^{-k}(\bar{\bX}_{2i})-g_{2}(\bar{\bX}_{2i})\}^2\bar{X}_{2io_1}\bar{X}_{2io_2}\\
		&\qquad + 2n^{-1}\sum_{k}\sum_{i\in\mathcal{I}_k}\{A_{2i}-g_{2}(\bar{\bX}_{2i})\}\{\widehat{g}_{2}^{-k}(\bar{\bX}_{2i})-g_{2}(\bar{\bX}_{2i})\}\bar{X}_{2io_1}\bar{X}_{2io_2}\\
		&=B_7+B_8.
	\end{align*}
	By repeating the argument used to control $B_5$, we establish $B_7=O_p\{\|\widehat{g}_{2}^{-k}(\cdot)-g_{2}(\cdot)\|_{P,2}^2\}=o_p(n^{-1/2})$. Using the same argument as for $B_2$, we obtain $B_8=O_p\{\|\widehat{g}_{2}^{-k}(\cdot)-g_{2}(\cdot)\|_2\cdot n^{-1/2}\}=o_p(n^{-1/2})$. Hence 
	$$\|\widehat{\bV}-\bV\|_\infty\le \bar{d}^2\cdot \max_{o1,o2}|	\widehat{\bV}_{o_1o_2}-\bV_{o_1o_2}|=o_p(n^{-1/2}).$$
	By the law of large numbers, we obtain $\|\bV-E\{{\rm var}(A_{2i}\mid \bar{\bX}_{2i})\bar{\bX}_{2i}\bar{\bX}_{2i}^\top\}\|_\infty=o_p(1)$. It then follows from Lemma~\ref{lemma:matrix_norm} and Assumption~\ref{ass:PD_matrix_} that $\|\bV^{-1}-E\{{\rm var}(A_{2i}\mid \bar{\bX}_{2i})\bar{\bX}_{2i}\bar{\bX}_{2i}^\top\}^{-1}\|_\infty=o_p(1)$ and $\|\bV^{-1}\|_\infty=O_p(1)$. Moreover, by combining Lemma~\ref{lemma:matrix_norm} with the bound $\| \widehat{\bV} - \bV \|_\infty = o_p(n^{-1/2})$, we further obtain $\|\widehat{\bV}^{-1}-\bV^{-1}\|_\infty=o_p(n^{-1/2})$. In conclusion, we have
	\begin{equation}\label{hat_alpha}
		\sqrt{n}(\widetilde{\balpha}-\breve{\balpha})=\bV^{-1}B_1+o_p(1).
	\end{equation}
	\textbf{Show $\sqrt{n}(\breve{\balpha}-\check{\balpha})=o_p(1)$.} We have 
	\begin{align*}
		\sqrt{n}(\breve{\balpha}-\check{\balpha})
		&=(\widehat{\bV}^{-1}-\bV^{-1})\dfrac{1}{\sqrt{n}}\sum_{i\in[n]}\bar{\bX}_{2i}\{A_{2i}-g_{2}(\bar{\bX}_{2i})\}^2\Delta^{2t}(\bar{\bX}_{2i})\\
		&\quad + \widehat{\bV}^{-1}\dfrac{1}{\sqrt{n}}\sum_k\sum_{i\in\mathcal{I}_k}\bar{\bX}_{2i}\{g_{2}(\bar{\bX}_{2i})-\widehat{g}_{2}^{-k}(\bar{\bX}_{2i})\}^2\Delta^{2t}(\bar{\bX}_{2i})\\
		&\quad + 2\widehat{\bV}^{-1}\dfrac{1}{\sqrt{n}}\sum_{k}\sum_{i\in\mathcal{I}_k}\bar{\bX}_{2i}\{A_{2i}-g_{2}(\bar{\bX}_{2i})\}\{g_{2}(\bar{\bX}_{2i})-\widehat{g}_{2}^{-k}(\bar{\bX}_{2i})\}\Delta^{2t}(\bar{\bX}_{2i})\\
		&=B_9+B_{10}+B_{11}.
	\end{align*}
	By the law of large numbers,
	\begin{align*}
		\|B_9\|_\infty&\le \sqrt{n}\|\widehat{\bV}^{-1}-\bV^{-1}\|_\infty\\
		&\qquad\times\|E[\bar{\bX}_{2i}\{A_{2i}-g_{2}(\bar{\bX}_{2i})\}^2\Delta^{2t}(\bar{\bX}_{2i})]\{1+o_p(1)\}\|_\infty=o_p(1).
	\end{align*}
	Hence $B_9=o_p(1)$. Using the fact that \( \|\widehat{\bV}^{-1}\|_\infty = O_p(1) \), and applying a similar argument as for term \( B_5 \), we obtain $
	B_{10}=O_p\left\{\sqrt{n}\cdot\|\widehat{g}_{2}^{-k}(\cdot)-g_{2}(\cdot)\|_{P,2}^2\right\}=o_p(1)$.
	Similarly, by repeating the argument used to control $B_2$, we find $B_{11}=O_p\left\{\|\widehat{g}_{2}^{-k}(\cdot)-g_{2}(\cdot)\|_{P,2}\right\}=o_p(1)$. Combining the upper bounds for $B_9$, $B_{10}$, and $B_{11}$, we conclude that
	\begin{align}\label{tilde_alpha_1}
	\sqrt{n}(\breve{\balpha}-\check{\balpha})=o_p(1).
	\end{align}
	\textbf{Control $\sqrt{n}(\check{\balpha}-\bar\balpha^*)$}. We have
	\begin{equation}\label{tilde_alpha_2}
		\sqrt{n}(\check{\balpha}-\bar\balpha^{*})=\bV^{-1}\dfrac{1}{\sqrt{n}}\sum_{i\in[n]}\{A_{2i}-g_{2}(\bar{\bX}_{2i})\}^2\bar{\bX}_{2i}\{\bar{\Delta}^{2t}(\bar{\bX}_{2i})-\bar{\bX}_{2i}^\top\bar\balpha^{*}\}.
	\end{equation}
	Combining \eqref{hat_alpha}, \eqref{tilde_alpha_1}, and \eqref{tilde_alpha_2}, we have
	\begin{align*}
		&\sqrt{n}(\widehat{\balpha}-\bar\balpha^{*})\\
		&\quad=\bV^{-1}\dfrac{1}{\sqrt{n}}\sum_{i\in[n]}\{A_{2i}-g_{2}(\bar{\bX}_{2i})\}\bar{\bX}_{2i}[\varepsilon_{2ti}+\{A_{2i}-g_{2}(\bar{\bX}_{2i})\}\{\bar\Delta^{2t}(\bar{\bX}_{2i})-\bar{\bX}_{2i}^\top\bar\balpha^{*}\}]+o_p(1)\\
		&\quad=\bV^{-1}\dfrac{1}{\sqrt{n}}\sum_{i\in[n]}\{A_{2i}-g_{2}(\bar{\bX}_{2i})\}\bar{\bX}_{2i}[Y_i-f_{2}(\bar{\bX}_{2i})-\{A_{2i}-g_{2}(\bar{\bX}_{2i})\}\{\bar{\bX}_{2i}^\top\bar\balpha^{*}+\Delta^{C,2t}\}]+o_p(1)\\
		&\quad\xrightarrow{\rm d}\mathcal{N}(0,\bM^{\bar\balpha}),
	\end{align*}
	where 
	\begin{align*}
		\bM^{\bar\balpha}&=E\{{\rm var}(A_{2}\mid \bar{\bX}_{2})\bar{\bX}_{2}\bar{\bX}_{2}^\top\}^{-1}
		E(\bQ^{\bar\balpha}\bQ^{\bar\balpha\top})
		E\{{\rm var}(A_{2}\mid \bar{\bX}_{2})\bar{\bX}_{2}\bar{\bX}_{2}^\top\}^{-1},\\
		\bQ^{\bar\balpha}&=\{A_{2}-g_{2}(\bar{\bX}_{2})\}\bar{\bX}_{2}[Y-f_{2}(\bar{\bX}_{2})-\{A_{2}-g_{2}(\bar{\bX}_{2})\}\{\bar{\bX}_{2}^\top\bar\balpha^*+\Delta^{C,2t}\}].
	\end{align*}
	We next prove the asymptotic normality of $\widehat{\balpha}_{j_1a_1j_2}$ for any $j_1\in\mathcal{J}_1$, $j_2\in\mathcal{J}_2$, and $a_1\in\{0,1\}$. Note that
	\begin{align*}
	\sqrt{n}(\widehat{\balpha}_{j_1a_1j_2}-\balpha_{j_1a_1j_2}^*)&=\sqrt{n}\bigg(\dfrac{1}{n}\sum_{i\in [n]}\bS_{\bar{J}_2i}\bS_{\bar{j}_2i}^\top\bigg)^{-1}\\
	&\qquad\times \dfrac{1}{n}\sum_{i\in [n]}\bS_{\bar{j}_2i}\{(\bar{\bS}_{2i},a_1)^\top\bar\balpha^*-\bS_{\bar{j}_2i}^\top\balpha^*_{j_1a_1j_2}+(\bar{\bS}_{2i},a_1)^\top(\widehat{\balpha}-\bar\balpha^*)\}\\
	&\xrightarrow{\rm d}\mathcal{N}(0,\bM_{j_1a_1j_2}^{\balpha}),
	\end{align*}
	where
	\begin{align*}
	\bM_{j_1a_1j_2}^{\balpha}&=E(\bS_{\bar{j}_2}\bS_{\bar{j}_2}^\top)^{-1}E(\bQ_{j_1a_1j_2}^{\balpha}\bQ_{j_1a_1j_2}^{\balpha\top})E(\bS_{\bar{j}_2}\bS_{\bar{j}_2}^\top)^{-1},\\
	\bQ_{j_1a_1j_2}^{\balpha}&=\bS_{\bar{j}_2}\{(\bar{\bS}_{2i},a_1)^\top\bar\balpha^*-\bS_{\bar{j}_2i}^\top\balpha^*_{j_1a_1j_2}\}+E\{\bS_{\bar{j}_2}(\bar{\bS}_{2i},a_1)\}E\{{\rm var}(A_{2}\mid \bar{\bX}_{2})\bar{\bX}_{2}\bar{\bX}_{2}^\top\}^{-1}	\bQ^{\bar\balpha}.
	\end{align*}
 	
	\noindent\textbf{Step 2a: Asymptotic Normality of $\widetilde{\bbeta}_{j_1j_2}$ and $\widehat{\bbeta}_{j_1a_1j_2}$} 
	
	For any $j_1\in\mathcal{J}_1$ and $j_2\in\mathcal{J}_2$, denote $\Delta^{C,2c}_{j_2}=C^{2c}(j_2)-C^{2c}(j_2^\mathrm{f})$, $\bar{\bX}_{2i}=(\bar{\bS}_{2i},A_{1i})$, $\bX_{\bar{l}_{2}^{\,\rm f}i}=(\bS_{\bar{l}_{2}^{\,\rm f}i},A_{1i})$,
	\begin{align*}
		Y_{j_1j_2i}^{2c*}&=\bar{\bX}_{2i}^\top\bar\balpha^*\{I(\bS_{\bar{j}_2i}^\top\balpha^*_{j_1A_{1i}j_2}> 0)-I(\bS_{\bar{j}_{2}^{\rm f}i}^\top\balpha^*_{j_1A_{1i}j_{2}^{\rm f}}> 0)\}-\Delta^{C,2c}_{j_2},\\
		\widehat{Y}_{j_1j_2i}^{2c}&=\bar{\bX}_{2i}^\top\widehat{\balpha}\{I(\bS_{\bar{j}_2i}^\top\widehat{\balpha}_{j_1A_{1i}j_2}> 0)-I(\bS_{\bar{j}_{2}^{\rm f}i}^\top\widehat{\balpha}_{j_1A_{1i}j_{2}^{\rm f}}> 0)\}-\Delta^{C,2c}_{j_2}.
	\end{align*}
	We next prove the asymptotic normality of $\widetilde{\bbeta}_{j_1j_2}$. We obtain
	\begin{align*}
		\widetilde{\bbeta}_{j_1j_2}-\bar{\bbeta}^*_{j_1j_2}
		&=\bU^{-1}n^{-1}\sum_{i\in[n]}\bX_{\bar{l}_2^{\,\rm f}i}\left(Y_{j_1j_2i}^{2c*}-\bX_{\bar{l}_2^{\,\rm f}i}^\top\bar{\bbeta}_{j_1j_2}^{*}+\widehat{Y}_{j_1j_2i}^{2c}-Y_{j_1j_2i}^{2c*}\right)\\
		&=\bU^{-1}n^{-1}\sum_{i\in[n]}\bX_{\bar{l}_2^{\,\rm f}i}\left\{Y_{j_1j_2i}^{2c*}-\bX_{\bar{l}_2^{\,\rm f}i}^\top\bar{\bbeta}_{j_1j_2}^{*}+\bar{\bX}_{2i}^\top(\widehat{\balpha}-\bar\balpha^*)E_{j_1j_2i}\right\}\\
		&\qquad + \bU^{-1}n^{-1}\sum_{i\in[n]}\bX_{\bar{l}_2^{\,\rm f}i}\bar{\bX}_{2i}^\top\bar\balpha^*(\Delta_{j_1j_2i}^E-\Delta_{j_1j_{2}^{\rm f}i}^E)\\
		&\qquad + \bU^{-1}n^{-1}\sum_{i\in[n]}\bX_{\bar{l}_2^{\,\rm f}i}\bar{\bX}_{2i}^\top(\widehat{\balpha}-\bar\balpha^*)(\Delta_{j_1j_2i}^E-\Delta_{j_1j_{2}^{\rm f}i}^E)\\
		&=\bU^{-1}(D_1+D_2+D_3),
	\end{align*}
	where $\bU=n^{-1}\sum_{i\in [n]}\bX_{\bar{l}_2^{\,\rm f}i}\bX_{\bar{l}_2^{\,\rm f}i}^\top$, 
	\begin{align*}
		E_{j_1j_2i}&=I(\bS_{\bar{j}_2i}^\top\balpha^*_{j_1A_{1i}j_2}> 0)-I(\bS_{\bar{j}_{2}^{\rm f}i}^\top\balpha^*_{j_1A_{1i}j_{2}^{\rm f}}> 0),\\
		\Delta_{j_1j_2i}^E&=I(\bS_{\bar{j}_2i}^\top\widehat{\balpha}_{j_1A_{1i}j_2}> 0)-I(\bS_{\bar{j}_2i}^\top\balpha^*_{j_1A_{1i}j_2}> 0),\\
		{\rm and}\;\; \Delta_{j_1j_{2}^{\rm f}i}^E&=I(\bS_{\bar{j}_{2}^{\rm f}i}^\top\widehat{\balpha}_{j_1A_{1i}j_{2}^{\rm f}}> 0)-I(\bS_{\bar{j}_{2}^{\rm f}i}^\top\balpha^*_{j_1A_{1i}j_{2}^{\rm f}}> 0).
	\end{align*}
	We now show that both $D_2$ and $D_3$ are $o_p(n^{-1/2})$ terms under Assumption \ref{ass:margin2} with $r>1$. By Assumptions \ref{ass:uni_bound_} and \ref{ass:PD_matrix_}, there exists a constant $C>0$ such that
	\begin{align*}
		c_2\|\bar\balpha^*\|_2^2&\le \lambda_{\rm min}[E\{{\rm var}(A_{2i}\mid\bar{\bX}_{2i})\bar{\bX}_{2i}\bar{\bX}_{2i}^\top\}]\cdot \|\bar\balpha^*\|_2^2\\
		&\le E[\{A_{2i}-g_{2}(\bar{\bX}_{2i})\}^2(\bar{\bX}_{2i}^\top\bar\balpha^*)^2]\le E[\{A_{2i}-g_{2}(\bar{\bX}_{2i})\}^2\bar{\Delta}^{2t}(\bar{\bX}_{2i})^2]\le C,
	\end{align*}
	which implies $\|\bar\balpha^*\|_2=O(1)$. We denote 
	$$
	\widetilde{\Delta}_{j_1j_2i}^E=I\{|\bS_{\bar{j}_2i}^\top\balpha^*_{j_1A_{1i}j_2}|\le |\bS_{\bar{j}_2i}^\top(\widehat{\balpha}_{j_1A_{1i}j_2}-\balpha^*_{j_1A_{1i}j_2})|\},
	$$ 
	and note that $|\Delta_{j_1j_2i}^E|\le \widetilde{\Delta}_{j_1j_2i}^E$, since for any $a,b>0$, we have $|a-b|\ge |a|$ whenever $ab<0$. For any $\delta\in (0,1)$, the consistency of $\widehat{\balpha}_{j_1a_1j_2}$ ensures there exists a constant $C_{\delta}$ depends on $\delta$ such that $P(\mathcal{A}_{j_1j_2}^{\delta})>1-\delta$, where 
	\begin{align}\label{A_J1J2a1}
	\mathcal{A}_{j_1j_2}^{\delta}=\left\{(\|\widehat{\balpha}_{j_10j_2}-\balpha^*_{j_10j_2}\|_2)+(\|\widehat{\balpha}_{j_11j_2}-\balpha^*_{j_11j_2}\|_2)\le C_{\delta}n^{-1/2}\right\}.
	\end{align} We denote 
	$$
	D_{21}=		n^{-1}\sum_{i\in[n]}\bX_{\bar{l}_2^{\,\rm f}i}\bar{\bX}_{2i}^\top\balpha^*\Delta_{j_1j_2i}^E.
	$$ 
	We now bound the $o$-th element of $D_{21}$, conditional on the event $\mathcal{A}_{j_1j_2}^{\delta}$, it follows from Assumption \ref{ass:uni_bound_} and $\|\bar\balpha^*\|_2=O(1)$ that there exist a constant $C^\prime>0$ such that
	\begin{equation}\label{D21}
		\begin{split}
			D_{21o}&=O\left\{n^{-1}\sum\nolimits_{i\in[n]}|\Delta_{j_1j_2i}^E|\right\}=O\left\{n^{-1}\sum\nolimits_{i\in[n]}\widetilde{\Delta}_{j_1j_2i}^E\right\}\\
			&=O\left\{n^{-1}\sum\nolimits_{i\in[n]}I(|\bS_{\bar{j}_2i}^\top\balpha^*_{j_1A_{1i}j_2}|\le C^\prime C_{\delta}n^{-1/2})\right\}.
		\end{split}
	\end{equation}
	Under Assumption \ref{ass:margin2} with $r>1$ and applying the law of large numbers, we can obtain that as $n\to \infty$,
	\begin{align*}
	&n^{-1}\sum\nolimits_{i\in[n]} I(|\bS_{\bar{j}_2i}^\top\balpha^*_{j_1A_{1i}j_2}|\le C^\prime C_{\delta}n^{-1/2})\xrightarrow{\rm p}P(|\bS_{\bar{j}_2i}^\top\balpha^*_{j_1A_{1i}j_2}|\le C^\prime C_{\delta}n^{-1/2})\\
	&\qquad = P(A_{1i}=0,|\bS_{\bar{j}_2i}^\top\balpha^*_{j_10j_2}|\le C^\prime C_{\delta}n^{-1/2})+P(A_{1i}=1,|\bS_{\bar{j}_2i}^\top\balpha^*_{j_11j_2}|\le C^\prime C_{\delta}n^{-1/2}),
	\end{align*}
	which is bounded by $\sum_{a_1=0}^1P(|\bS_{\bar{j}_2i}^\top\balpha^*_{j_1a_1j_2}|\le C^\prime C_{\delta}n^{-1/2})=o_p(n^{1/2})$ under Assumption \ref{ass:margin2} with $r>1$. Hence $D_{21}=o_p(n^{-1/2})$. Analogously, we obtain 
	$$
	n^{-1}\sum_{i\in[n]}\bX_{\bar{l}_2^{\,\rm f}i}\bar{\bX}_{2i}^\top\bar\balpha^*\Delta_{j_1j_{2}^{\rm f}i}^E=o_p(n^{-1/2}).
	$$
	Together with Assumption \ref{ass:PD_matrix_} and the fact that $\bX_{\bar{l}_2^{\,\rm f}}$ is a subvector of $\bar{\bX}_2$, it follows that $\lambda_{\min}\{E(\bX_{\bar{l}_2^{\,\rm f}}\bX_{\bar{l}_2^{\,\rm f}}^\top)\}\ge \lambda_{\min}\{E(\bar{\bX}_{2}\bar{\bX}_{2}^\top)\}\ge \lambda_{\min}[E\{{\rm var}(A_{2}\mid \bar{\bX}_{2})\bar{\bX}_{2}\bar{\bX}_{2}^\top\}]>c_2$
	Using the same argument as for showing $\|\bV^{-1}\|_\infty=O_p(1)$, we also have $\|\bU^{-1}\|_\infty=O_p(1)$, which implies $\bU^{-1}D_{2}=o_p(n^{-1/2})$. For the term $D_3$, applying the same reasoning as for $D_{21}$, 
	\begin{align*}
		n^{-1}\sum_{i\in[n]}\bX_{\bar{l}_2^{\,\rm f}i}\bar{\bX}_{2i}^\top(\widehat{\balpha}-\bar\balpha^*)\Delta_{j_1j_2i}^E&=O\bigg\{\|\widehat{\balpha}-\bar\balpha^*\|_2\cdot n^{-1}\sum_{i\in[n]}|\Delta_{j_1j_2i}^E|\bigg\}=o_p(n^{-1/2}),\\
		n^{-1}\sum_{i\in[n]}\bX_{\bar{l}_2^{\rm f}i}\bar{\bX}_{2i}^\top(\widehat{\balpha}-\bar\balpha^*)\Delta_{j_1j_{2}^{\rm f}i}^E&=O\bigg\{\|\widehat{\balpha}-\bar\balpha^*\|_2\cdot n^{-1}\sum_{i\in[n]}|\Delta_{j_1j_{2}^{\rm f}i}^E|\bigg\}=o_p(n^{-1/2}).
	\end{align*}
	Therefore, $\bU^{-1}D_3=o_p(n^{-1/2})$. Combining the upper bounds for $\bU^{-1}D_2$ and $\bU^{-1}D_3$, 
	\begin{align*}
		\sqrt{n}(\widetilde{\bbeta}_{j_1j_2}-\bar{\bbeta}^{*}_{j_1j_2})
		&=\bU^{-1}(1/\sqrt{n})\sum_{i\in[n]}\bX_{\bar{l}_2^{\,\rm f}i}\left\{Y_{j_1j_2i}^{2c*}-\bX_{\bar{l}_2^{\,\rm f}i}^\top\bar{\bbeta}_{j_1j_2}^{*}+\bar{\bX}_{2i}^\top(\widehat{\balpha}-\bar\balpha^*)E_{j_1j_2i}\right\}+o_p(1)\\
		&\xrightarrow{\rm d}\mathcal{N}(0,\bM_{j_1j_2}^{\bar{\bbeta}}),
	\end{align*}
	where
	\begin{align*}
	\bM_{j_1j_2}^{\bar{\bbeta}}&=E(\bX_{\bar{l}_2^{\,\rm f}}\bX_{\bar{l}_2^{\,\rm f}}^\top)^{-1}E(\bQ^{\bar{\bbeta}}_{j_1j_2}\bQ^{\bar{\bbeta}\top}_{j_1j_2})E(\bX_{\bar{l}_2^{\,\rm f}}\bX_{\bar{l}_2^{\,\rm f}}^\top)^{-1},\\
	\bQ^{\bar{\bbeta}}_{j_1j_2}&=\bX_{\bar{l}_2^{\,\rm f}}(Y_{j_1j_2}^{2c*}-\bX_{\bar{l}_2^{\,\rm f}}^\top\bar{\bbeta}_{j_1j_2}^*)+E[E_{j_1j_2}\bX_{\bar{l}_2^{\,\rm f}}\bar{\bX}_2]E\{{\rm var}(A_{2}\mid \bar{\bX}_2)\bar{\bX}_{2}\bar{\bX}_{2}^\top\}^{-1}\bQ^{\bar\balpha},
	\end{align*}
	and $E_{j_1j_2}=I(\bS_{\bar{j}_2}^\top\balpha^*_{j_1A_1j_2}> 0)-I(\bS_{\bar{j}_{2}^{\rm f}}^\top\balpha^*_{j_1A_1j_{2}^{\rm f}}> 0)$. We next prove the asymptotic normality of $\widehat{\bbeta}_{j_1a_1j_2}$. Note that
	\begin{align}
		\sqrt{n}(\widehat{\bbeta}_{j_1a_1j_2}-\bbeta_{j_1a_1j_2}^*)&=\sqrt{n}\bigg(\dfrac{1}{n}\bS_{\bar{l}_2i}\bS_{\bar{l}_2i}^\top\bigg)^{-1}\bigg[\dfrac{1}{n}\sum_{i\in [n]}\bS_{\bar{l}_2i}\{(\bS_{\bar{l}_2^{\,\rm f}},a_1)^\top\bar{\bbeta}_{j_1j_2}^{*}-\bS_{\bar{l}_2}^\top \bbeta_{j_1a_1j_2}^*\nonumber\\
		&\qquad\qquad\qquad\qquad\qquad\qquad\qquad+(\bS_{\bar{l}_2^{\,\rm f}},a_1)^\top(\bar\bbeta_{j_1j_2}^{*}-\widetilde{\bbeta}_{j_1j_2})\}\bigg]\label{beta_hat_deco}\\
		&\xrightarrow{\rm d}\mathcal{N}(0,\bM_{j_1a_1j_2}^{\bbeta}),\nonumber
		\end{align}
	where
	\begin{align*}
		\bM_{j_1a_1j_2}^{\bbeta}&=E(\bS_{\bar{l}_2}\bS_{\bar{l}_2}^\top)^{-1}E(\bQ_{j_1a_1j_2}^{\bbeta}\bQ_{j_1a_1j_2}^{\bbeta\top})E(\bS_{\bar{l}_2}\bS_{\bar{l}_2}^\top)^{-1},\\
		\bQ_{j_1a_1j_2}^{\bbeta}&=\bS_{\bar{l}_2}\{(\bS_{\bar{l}_2^{\,\rm f}},a_1)^\top\bar\bbeta_{j_1j_2}^{*}-\bS_{\bar{l}_2}^\top \bbeta_{j_1a_1j_2}^*\}+E\{\bS_{\bar{l}_2}(\bS_{\bar{l}_2^{\,\rm f}i},a_1)\}E(\bX_{\bar{l}_2^{\,\rm f}}\bX_{\bar{l}_2^{\,\rm f}}^\top)^{-1}\bQ^{\bar{\bbeta}}_{j_1j_2}.
	\end{align*}
	\noindent\textbf{Step 2a: Consistency of $\widetilde{\bbeta}_{j_1j_2}$ and $\widehat{\bbeta}_{j_1a_1j_2}$} 
	
	We now establish the consistency of $\widetilde{\bbeta}_{j_1j_2}$ under Assumption \ref{ass:margin2} with $r\ge 1$. Consider the $o$-th element of $D_{21}$, for any $\delta\in(0,1)$, conditional on the event $\mathcal{A}_{j_1j_2}^{\delta}$ defined in \eqref{A_J1J2a1},
	\begin{equation}\label{D21_bigO}
	\begin{split}
		D_{21o}&=O\left\{n^{-1}\sum\nolimits_{i\in[n]}|\Delta_{j_1j_2i}^E|\right\}=O\left\{n^{-1}\sum\nolimits_{i\in[n]}\widetilde{\Delta}_{j_1j_2i}^E\right\}\\
		&=O\left\{n^{-1}\sum\nolimits_{i\in[n]}I(|\bS_{\bar{j}_2i}^\top\balpha^*_{j_1A_{1i}j_2}|\le C^\prime C_{\delta}n^{-1/2})\right\}.
		\end{split}
	\end{equation}
	By Assumption \ref{ass:margin2} with $r\ge 1$ and the law of large numbers, we have $n^{-1}\sum\nolimits_{i\in[n]}I(|\bS_{\bar{j}_2i}^\top\balpha^*_{j_1A_{1i}j_2}|\le C^\prime C_{\delta}n^{-1/2})\xrightarrow{\rm p} P(|\bS_{\bar{j}_2i}^\top\balpha^*_{j_1A_{1i}j_2}|\le C^\prime C_{\delta}n^{-1/2})$, which is bounded by $\sum_{a_1=0}^1P(|\bS_{\bar{j}_2i}^\top\balpha^*_{j_1a_1j_2}|\le C^\prime C_{\delta}n^{-1/2})=O_p(n^{1/2})$ under Assumption \ref{ass:margin2} with $r\ge 1$. Consequently, $D_{21}=O_p(n^{-1/2})$. Similarily, we obtain
	$$
	n^{-1}\sum_{i\in[n]}\bX_{\bar{l}_2^{\,\rm f}i}\bar{\bX}_{2i}^\top\bar\balpha^*\Delta_{j_1j_{2}^{\rm f}i}^E=O_p(n^{-1/2}),$$ 
	and thus $\bU^{-1}D_2=O_p(n^{-1/2})$. In addition, since $D_3=O(\|\widehat{\balpha}-\bar\balpha^*\|_2)=O_p(n^{-1/2})$, it follows that $\bU^{-1}D_3=O_p(n^{-1/2})$. Since $D_1$ is sum of i.i.d. mean-zero random variables, we have $D_1=O_p(n^{-1/2})$ and therefore $\bU^{-1}D_1=O_p(n^{-1/2})$. Combining the results above, we conclude that $\widetilde{\bbeta}_{j_1j_2}-\bar{\bbeta}^{*}_{j_1j_2}=O_p(n^{-1/2})$. In this case, noting that the term $n^{-1}\sum_{i\in[n]}\bS_{\bar{l}_2}(\bS_{\bar{l}_2^{\,\rm f}},a_1)^\top(\bar\bbeta_{j_1j_2}^{*}-\widetilde{\bbeta}_{j_1j_2})$ in \eqref{beta_hat_deco} is also of order $O_p(n^{-1/2})$, we can infer that $\widehat{\bbeta}_{j_1a_1j_2}-\bbeta^{*}_{j_1a_1j_2}=O_p(n^{-1/2})$.
	
	\noindent\textbf{Step 3a: Asymptotic Normality of $\widetilde{\bgamma}_{j_1}$ and $\widehat{\bgamma}_{j_1}$} 
	
	 For any $j_1\in\mathcal{J}_1$, we first show the asymptotic normality of $\widetilde{\bgamma}_{j_1}$. Since $Y_{j_1i}^{1c*}-f^*_{j_1}(\bS_{1i})=\{A_{1i}-g_1(\bS_{1i})\}\{\bar Q^{1t*}(\bS_{1i},j_1,1)-\bar Q^{1t*}(\bS_{1i},j_1,0)\}+\varepsilon^{1t*}$, $\bar\bgamma_{j_1}^{*}$ can be expressed as 
	 $$
	 \bar\bgamma^{*}_{j_1}=\mathop{\arg\min}_{\bgamma} E[\{A_{1i}-g_1(\bS_{1i})\}^2\{\bar Q^{1t*}(\bS_{1i},j_1,1)-\bar Q^{1t*}(\bS_{1i},j_1,0)-\bS_{1i}^\top\bgamma\}^2],\\
	 $$
	 and we define
	\begin{align*}
		\breve{\bgamma}_{j_1}&=\mathop{\arg\min}_{\bgamma} n^{-1}\sum_{k}\sum_{i\in\mathcal{I}_k}\{A_{1i}-\widehat{g}_1^{-k}(\bS_{1i})\}^2\{\bar Q^{1t*}(\bS_{1},j_1,1)- \bar Q^{1t*}(\bS_{1},j_1,0)-\bS_{1i}^\top\bgamma\}^2,\\
		\check{\bgamma}_{j_1}&=\mathop{\arg\min}_{\bgamma} n^{-1}\sum_{i\in[n]}\{A_{1i}-g_1(\bS_{1i})\}^2\{ \bar Q^{1t*}(\bS_{1},j_1,1)- \bar Q^{1t*}(\bS_{1},j_1,0)-\bS_{1i}^\top\bgamma\}^2.
	\end{align*}
	We next control the terms $\sqrt{n}(\widetilde{\bgamma}_{j_1}-\breve{\bgamma}_{j_1})$, $\sqrt{n}(\breve{\bgamma}_{j_1}-\check{\bgamma}_{j_1})$, as well as $\sqrt{n}(\check{\bgamma}_{j_1}-\bar\bgamma_{j_1}^{*})$ in order to bound $\sqrt{n}(\widetilde{\bgamma}_{j_1}-\bar\bgamma_{j_1}^{*})$. We denote 
	\begin{align*}
		\widehat{\bW}=n^{-1}\sum_{k}\sum_{i\in\mathcal{I}_k}\{A_{1i}-\widehat{g}_{1}^{-k}(\bS_{1i})\}^2\bS_{1i}\bS_{1i}^\top,\quad 	\bW=n^{-1}\sum_{i\in[n]}\{A_{1i}-g_{1}(\bS_{1i})\}^2\bS_{1i}\bS_{1i}^\top.
	\end{align*}
	
	\noindent\textbf{Control $\sqrt{n}(\widetilde{\bgamma}_{j_1}-\breve{\bgamma}_{j_1})$}. Note that
	\begin{align*}
		\sqrt{n}(\widetilde{\bgamma}_{j_1}-\breve{\bgamma}_{j_1})&=\widehat{\bW}^{-1}\dfrac{1}{\sqrt{n}}\sum_k\sum_{i\in\mathcal{I}_k}\bS_{1i}\{A_{1i}-\widehat{g}_{1}^{-k}(\bS_{1i})\}\\
		&\qquad\qquad \times[\widehat{Y}_{j_1i}^{1t}-\widehat{f}_{j_1}^{-k}(\bS_{1i})-\{A_{1i}-\widehat{g}_{1}^{-k}(\bS_{1i})\}\{\bar Q^{1t*}(\bS_{1},j_1,1)-\bar Q^{1t*}(\bS_{1},j_1,0)\}]\\
		&=\widehat{\bW}^{-1}(F_1+F_2+\cdots+F_8),
	\end{align*}
	where
	\begin{align*}
		F_1&=(1/\sqrt{n})\sum\nolimits_{i\in[n]}\bS_{1i}\{A_{1i}-g_{1}(\bS_{1i})\}(\varepsilon_{1ti}^*+\widehat{Y}^{1t}_{j_1i}-Y^{1t*}_{j_1i}),\\
		F_2&=(1/\sqrt{n})\sum\nolimits_k\sum\nolimits_{i\in\mathcal{I}_k}\bS_{1i}\{g_{1}(\bS_{1i})-\widehat{g}_{1}^{-k}(\bS_{1i})\}(\widehat{Y}_{j_1i}^{1t}-Y_{j_1i}^{1t*}),\\
		F_3&=(1/\sqrt{n})\sum\nolimits_k\sum\nolimits_{i\in\mathcal{I}_k}\bS_{1i}\{A_{1i}-g_{1}(\bS_{1i})\}\{f_{j_1}(\bS_{1i})-\widehat{f}_{j_1}^{-k}(\bS_{1i})\},\\
		F_{4}&=(1/\sqrt{n})\sum\nolimits_k\sum\nolimits_{i\in\mathcal{I}_k}\bS_{1i}\{A_{1i}-g_{1}(\bS_{1i})\}\{\widehat{g}_{1}^{-k}(\bS_{1i})-g_{1}(\bS_{1i})\}\\
		&\qquad\qquad\qquad\qquad\qquad\times\{\bar Q^{1t*}(\bS_{1},j_1,1)-\bar Q^{1t*}(\bS_{1},j_1,0)\},\\
		F_5&=(1/\sqrt{n})\sum\nolimits_k\sum\nolimits_{i\in\mathcal{I}_k}\bS_{1i}\{g_{1}(\bS_{1i})-\widehat{g}_{1}^{-k}(\bS_{1i})\}\varepsilon_{1ti}^*,\\
		F_6&=(1/\sqrt{n})\sum\nolimits_k\sum\nolimits_{i\in\mathcal{I}_k}\bS_{1i}\{g_{1}(\bS_{1i})-\widehat{g}_{1}^{-k}(\bS_{1i})\}\{f_{j_1}(\bS_{1i})-\widehat{f}_{j_1}^{-k}(\bS_{1i})\},\\
		F_7&=(1/\sqrt{n})\sum\nolimits_k\sum\nolimits_{i\in\mathcal{I}_k}\bS_{1i}\{g_{1}(\bS_{1i})-\widehat{g}_{1}^{-k}(\bS_{1i})\}\{\widehat{g}_{1}^{-k}(\bS_{1i})-g_{1}(\bS_{1i})\}\\
		&\qquad\qquad\qquad\qquad\qquad\times\{\bar Q^{1t*}(\bS_{1},j_1,1)-\bar Q^{1t*}(\bS_{1},j_1,0)\}.
	\end{align*}
	Analogous to the arguments for both $B_2$ and $B_5$ are $o_p(1)$ terms, we can show that 
	\begin{align*}
		F_3&=O\{\|\widehat{f}_{j_1}^{-k}(\cdot)-f^*_{j_1}(\cdot)\|_{P,2}\}=o_p(1),\\
		F_4&=O\{\|\widehat{g}_{1}^{-k}(\cdot)-g_{1}(\cdot)\|_{P,2}\}=o_p(1),\\
		F_5&=O\{\|\widehat{g}_{1}^{-k}(\cdot)-g_{1}(\cdot)\|_{P,2}\}=o_p(1).
		\end{align*}
	Using the same argument as for $B_5$, we obtain
	\begin{align*}
		F_6&=O\{\sqrt{n}\cdot\|\widehat{f}_{j_1}^{-k}(\cdot)-f^*_{j_1}(\cdot)\|_2\cdot\|\widehat{g}_{1}^{-k}(\cdot)-g_{1}(\cdot)\|_{P,2}\}=o_p(1),\\
		F_7&=O\{\sqrt{n}\cdot\|\widehat{g}_{1}^{-k}(\cdot)-g_{1}(\cdot)\|_{P,2}^2\}=o_p(1).
	\end{align*}
	We next bound the terms $F_1$ and $F_2$. Denote $\bX_{\bar{l}_2^{\,\rm f}i}=(\bS_{\bar{l}_2^{\,\rm f}i},A_{1i})$. For any $j_1\in\mathcal{J}_1$ and $a_1\in\{0,1\}$, when the maximum value is attained by more than one element in the sets $\{\bS_{\bar{l}_2}^\top\widehat{\bbeta}_{j_1a_{1}j_2}\}_{j_2\in\mathcal{J}_2}$ and $\{\bS_{\bar{l}_2}^\top\bbeta^*_{j_1a_1j_2}\}_{j_2\in\mathcal{J}_2}$, we define the following terms:
	\begin{align*}
	\widehat{G}_{j_1j_2i}&=\bS_{\bar{l}_2i}^\top\widehat{\bbeta}_{j_1A_{1i}j_2},\;\; 	\widetilde{G}_{j_1j_2i}=\bX_{\bar{l}_2^{\,\rm f}i}^\top\widetilde{\bbeta}_{j_1j_2},\;\; 
	G^*_{j_1j_2i}=\bS_{\bar{l}_2i}^\top\bbeta^*_{j_1A_{1i}j_2},\;\;
	\bar G^{*}_{j_1j_2i}=\bX_{\bar{l}_2^{\,\rm f}i}^\top\bar\bbeta^{*}_{j_1j_2},
	\end{align*}
	and the associated events
	\begin{align*}
	\widehat{\mathcal{B}}_{j_1j_2i}&=\{\widehat{G}_{j_1j_2i}\ge \widehat{G}_{j_1j_2^\prime i}, \widehat{G}_{j_1j_2i}> \widehat{G}_{j_1j_2^{\prime\prime} i}, j_2^\prime<j_2, j_2^{\prime\prime}>j_2\},\\
	\mathcal{B}_{j_1j_2i}&=\{G_{j_1j_2i}^*\ge G^*_{j_1j_2^\prime i}, G^*_{j_1j_2i}> G^*_{j_1j_2^{\prime\prime} i}, j_2^\prime<j_2, j_2^{\prime\prime}>j_2\},\;\;
	 \mathcal{B}_{j_1j_2i}^*=\{G_{j_1j_2i}^*> G_{j_1j_2^\prime i}^*, j_2^{\prime}\neq j_2\}.
	\end{align*}
	Here, $\widehat{B}_{j_1j_2i}$ denotes the event that the maximum value among $\{\widehat{G}_{j_1j_2i}\}$ occurs lastly at position $j_2$, while $\mathcal{B}_{j_1j_2i}$ is analogously defined based on $\{G_{j_1j_2i}^*\}$. We consider
	\begin{align*}
		\widehat{Y}_{j_1i}^{1t}&=Y_i-C^{2t}(A_{2i})+\bar{\bX}_{2i}^\top\widehat{\balpha}\{I(\bS_{\bar{j}_2^{\rm f}i}^\top\widehat{\balpha}_{j_1A_{1i}j_2^{\rm f}}> 0)-A_{2i}\}-C^{2c}(j_2^{\rm f})\\
		&\qquad\qquad+\sum\nolimits_{j_2\in\mathcal{J}_2}I(\widehat{\mathcal{B}}_{j_1j_2i})\widetilde{G}_{j_1j_2i}-C^{1t}(A_{1i}),\\
		Y_{j_1i}^{1t*}&=Y_i-C^{2t}(A_{2i})+\bar{\bX}_{2i}^\top\bar\balpha^*\{I(\bS_{\bar{j}_2^{\rm f}i}^\top\balpha^*_{j_1A_{1i}j_2^{\rm f}}> 0)-A_{2i}\}-C^{2c}(j_2^{\rm f})\\
		&\qquad\qquad+\sum\nolimits_{j_2\in\mathcal{J}_2}I(\mathcal{B}_{j_1j_2i})\bar{G}^{*}_{j_1j_2i}-C^{1t}(A_{1i}),
	\end{align*}
	where $\bar{\bX}_{2i}=(\bar{\bS}_{2i},A_{1i})$. We obtain
	\begin{align}
	\widehat{Y}_{j_1i}^{1t}-Y_{j_1i}^{1t*}&=\bar{\bX}_{2i}^\top(\widehat{\balpha}-\bar\balpha^*)\widetilde{E}_{j_1j_2^{\rm f}i}+\bar{\bX}_{2i}^\top\bar\balpha^*\Delta_{j_1j_2^{\rm f}i}^E+\bar{\bX}_{2i}^\top(\widehat{\balpha}-\bar\balpha^*)\Delta_{j_1j_2^{\rm f}i}^E\nonumber\\
	&\quad+\sum_{j_2\in\mathcal{J}_2}I(\mathcal{B}_{j_1j_2i}^*)(\widetilde{G}_{j_1j_2i}-\bar{G}^{*}_{j_1j_2i})+\sum_{j_2\in\mathcal{J}_2}\{I(\widehat{\mathcal{B}}_{j_1j_2i})-I(\mathcal{B}_{j_1j_2i})\}(\widetilde{G}_{j_1j_2i}-\bar{G}^{*}_{j_1j_2i})\nonumber\\
	&\quad+\sum_{j_2\in\mathcal{J}_2}\{I(\mathcal{B}_{j_1j_2i})-I(\mathcal{B}^*_{j_1j_2i})\}(\widetilde{G}_{j_1j_2i}-\bar{G}^{*}_{j_1j_2i})+\sum_{j_2\in\mathcal{J}_2}\{I(\widehat{\mathcal{B}}_{j_1j_2i})-I(\mathcal{B}_{j_1j_2i})\}\bar{G}^{*}_{j_1j_2i}\nonumber\\
	&=\Delta^Y_{1i}+\cdots+\Delta^Y_{7i}\label{Delta_Y},
	\end{align}
	where $\widetilde{E}_{j_1j_{2}^{\rm f}i}=I(\bS_{\bar{j}_2^{\rm f}i}^\top\balpha^*_{j_1A_{1i}j_{2}^{\rm f}}> 0)-A_{2i}$ and $\Delta_{j_1j_{2}^{\rm f}i}^E=I(\bS_{\bar{j}_2^{\rm f}i}^\top\widehat{\balpha}_{j_1A_{1i}j_{2}^{\rm f}}> 0)-I(\bS_{\bar{j}_2^{\rm f}i}^\top\balpha^*_{j_1A_{1i}j_{2}^{\rm f}}> 0)$. For the term $F_2$, we denote
	\begin{align*}
	F_2&=(1/\sqrt{n})\sum\nolimits_k\sum\nolimits_{i\in\mathcal{I}_k}\bS_{1i}\{g_{1}(\bS_{1i})-\widehat{g}_{1}^{-k}(\bS_{1i})\}(\Delta^Y_{1i}+\cdots+\Delta^Y_{7i})\\
	&=F_{21}+F_{22}+\cdots+F_{27}.
	\end{align*}
	For each $k\in[K]$ and each $j\in[7]$, let $
	F_{2jk}=|\mathcal{I}_k|^{-1}\sum_{i\in\mathcal{I}_k}\bS_{1i}\{g_{1}(\bS_{1i})-\widehat{g}_{1}^{-k}(\bS_{1i})\}\Delta^Y_{ji}$. 
	Consider the $o$-th element of $F_{21k}$, it follows by Assumption \ref{ass:uni_bound_} and Cauchy-Schwarz inequality that
	\begin{align*}
		F_{21ko}
		&=O\bigg\{|\mathcal{I}_k|^{-1}\sum_{i\in\mathcal{I}_k}|g_{1}(\bS_{1i})-\widehat{g}_{1}^{-k}(\bS_{1i})|\cdot |\bar{\bX}_{2i}^\top(\widehat{\balpha}-\bar\balpha^*)|\bigg\}\\
		&=O\bigg[\bigg\{|\mathcal{I}_k|^{-1}\sum_{i\in\mathcal{I}_k}|g_{1}(\bS_{1i})-\widehat{g}_{1}^{-k}(\bS_{1i})|^2\bigg\}^{1/2}\bigg\{|\mathcal{I}_k|^{-1}\sum_{i\in\mathcal{I}_k}|\bar{\bX}_{2i}^\top(\widehat{\balpha}-\bar\balpha^*)|^2\bigg\}^{1/2}\bigg]\\
		&=O\left\{\|\widehat{g}_{1}^{-k}(\cdot)-g_{1}(\cdot)\|_{P,2}\cdot\|\widehat{\balpha}-\bar\balpha^*\|_2\right\}=o_p(n^{-1/2}).
	\end{align*}
	Hence, $F_{21o}=\sqrt{n}\sum\nolimits_{k}(|\mathcal{I}_k|/n)F_{21k}\le \sqrt{n}\sum\nolimits_k F_{51k}=o_p(1)$ and $F_{21}=o_p(1)$. Similarily, we obtain $F_{23}$, $F_{24}$, $F_{25}$, and $F_{26}$ are all $o_p(1)$ terms. Consider the $o$-th element of $F_{22k}$. By Assumption \ref{ass:uni_bound_}, we have
	\begin{align*}
		F_{22ko}
		&=O\bigg\{|\mathcal{I}_k|^{-1}\sum_{i\in\mathcal{I}_k}|g_1(\bS_{1i})-\widehat{g}_1^{-k}(\bS_{1i})|\cdot |\Delta_{j_1j_{2}^{\rm f}i}^E|\bigg\}\\
		&=O\bigg[\bigg\{|\mathcal{I}_k|^{-1}\sum_{i\in\mathcal{I}_k}|g_{1}(\bS_{1i})-\widehat{g}_{1}^{-k}(\bS_{1i})|^2\bigg\}^{1/2}\bigg\{|\mathcal{I}_k|^{-1}\sum_{i\in\mathcal{I}_k}|\Delta_{j_1j_{2}^{\rm f}i}^E|\bigg\}^{1/2}\bigg]=o_p(n^{-1/2}),
	\end{align*}
	where the second equality follows from Cauchy-Schwarz inequality and the fact that $|\Delta_{j_1j_{2}^{\rm f}i}^E|^2=|\Delta_{j_1j_{2}^{\rm f}i}^E|$. Moreover, under the condition that $|\mathcal{I}_k|^{-1}\sum_{i\in\mathcal{I}_k}|\Delta_{j_1j_{2}^{\rm f}i}^E|= O_p(n^{-1/2})$ as well as $\|\widehat{g}_{1}^{-k}(\cdot)-g_{1}(\cdot)\|_{P,2}=o_p(n^{-1/4})$ by Assumption \ref{ass:nuisance1_}, the last equality holds. Therefore, we next proceed to prove that $|\mathcal{I}_k|^{-1}\sum_{i\in\mathcal{I}_k}|\Delta_{j_1j_{2}^{\rm f}i}^E|= O_p(n^{-1/2})$. For any $\delta\in(0,1)$, there exist a constant $C_{2\delta}$ depends on $\delta$ such that $P(\mathcal{A}_{j_1j_2^{\rm f}}^\delta)>1-\delta$, where $\mathcal{A}_{j_1j_2^{\rm f}}^\delta=\{\|\widehat{\balpha}_{j_10j_2^{\rm f}}-\balpha_{j_10j_2^{\rm f}}^*\|_2+\|\widehat{\balpha}_{j_11j_2^{\rm f}}-\balpha_{j_11j_2^{\rm f}}^*\|_2\}\le C_{2\delta}n^{-1/2}$. Conditional on $\mathcal{A}_{j_1j_2^{\rm f}}^\delta$, there exists a constant $C^{\prime\prime}$
	\begin{equation}\label{hat_Delta}
	\begin{split}
	|\mathcal{I}_k|^{-1}\sum_{i\in\mathcal{I}_k}|\Delta_{j_1j_{2}^{\rm f}i}^E|&\le |\mathcal{I}_k|^{-1}\sum_{i\in\mathcal{I}_k}I\{|\bS_{\bar{j}_2^{\rm f}i}^\top\balpha^*_{j_1A_{1i}j_2^{\rm f}}|\le |\bS_{\bar{j}_2^{\rm f}i}^\top(\widehat{\balpha}_{j_1A_{1i}j_2^{\rm f}}-\balpha^*_{j_1A_{1i}j_2^{\rm f}})|\}\\
	&\le |\mathcal{I}_k|^{-1}\sum_{i\in\mathcal{I}_k}I\{|\bS_{\bar{j}_2^{\rm f}i}^\top\balpha^*_{j_1A_{1i}j_2^{\rm f}}|\le C^{\prime\prime}C_{2\delta}n^{-1/2}\},
	\end{split}
	\end{equation}
	where the first equality holds, since for any $a,b>0$, we have $|a-b| \ge |a|$, the second equality follows from Assumption \ref{ass:uni_bound_}. By the law of large number, the above term converae in probability to
	\begin{align*}
	P(A_{1i}=0, |\bS_{\bar{j}_2^{\rm f}i}^\top\balpha^*_{j_10j_2^{\rm f}}|\le C^{\prime\prime}C_{2\delta}n^{-1/2})+P(A_{1i}=1, |\bS_{\bar{j}_2^{\rm f}i}^\top\balpha^*_{j_11j_2^{\rm f}}|\le C^{\prime\prime}C_{2\delta}n^{-1/2}).
	\end{align*}
	Under Assumption \ref{ass:margin2} with $r\ge 1$, the above term is bounded by
	\begin{align*}
 	P(|\bS_{\bar{j}_2^{\rm f}i}^\top\balpha^*_{j_10j_2^{\rm f}}|\le C^{\prime\prime}C_{2\delta}n^{-1/2})+P(|\bS_{\bar{j}_2^{\rm f}i}^\top\balpha^*_{j_11j_2^{\rm f}}|\le C^{\prime\prime}C_{2\delta}n^{-1/2})=O(n^{-1/2}).
	\end{align*}
 	Therefore, we have established that $|\mathcal{I}_k|^{-1}\sum_{i\in\mathcal{I}_k}|\Delta_{j_1j_{2}^{\rm f}i}^E|= O_p(n^{-1/2})$, which implies that $F_{22}=o_p(1)$. As for the term $F_{27}$, for $o$-th element of $F_{27k}$,
	\begin{equation}\label{F26}
	\begin{split}
	F_{27ko}
	&=\dfrac{1}{|\mathcal{I}_k|}\sum_{i\in\mathcal{I}_k}\bS_{1i}\{g_{1}(\bS_{1i})-\widehat{g}_{1}^{-k}(\bS_{1i})\}\\
	&\qquad\qquad\times\sum_{j_2\in\mathcal{J}_2}\bigg\{\sum_{j_2^{\prime\prime}\in\mathcal{J}_2}I(\widehat{\mathcal{B}}_{j_1j_2^{\prime\prime} i})\bar{G}^{*}_{j_1j_2^{\prime\prime} i}-\sum_{j_2^{\prime}\in\mathcal{J}_2}I(\mathcal{B}_{j_1j_2^{\prime} i})\bar{G}^{*}_{j_1j_2^{\prime} i}\bigg\}I(\widehat{\mathcal{B}}_{j_1j_2i})\\
	&=\dfrac{1}{|\mathcal{I}_k|}\sum_{i\in\mathcal{I}_k}\bS_{1i}\{g_{1}(\bS_{1i})-\widehat{g}_{1}^{-k}(\bS_{1i})\}\sum_{j_2\in\mathcal{J}_2}\bigg\{\bar{G}^{*}_{j_1j_2i}-\sum_{j_2^{\prime}\in\mathcal{J}_2}I(\mathcal{B}_{j_1j_2^{\prime} i})\bar{G}^{*}_{j_1j_2^{\prime} i}\bigg\}I(\widehat{\mathcal{B}}_{j_1j_2i})\\
	&=\dfrac{1}{|\mathcal{I}_k|}\sum_{i\in\mathcal{I}_k}\bS_{1i}\{g_{1}(\bS_{1i})-\widehat{g}_{1}^{-k}(\bS_{1i})\}\sum_{j_2\in\mathcal{J}_2}\sum_{j_2^\prime\neq j_2}(\bar{G}^{*}_{j_1j_2i}-\bar{G}^{*}_{j_1j_2^\prime i})I(\widehat{\mathcal{B}}_{j_1j_2i}\cap\mathcal{B}_{j_1j_2^\prime i})\\
	&=O\bigg[\sum_{j_2\in\mathcal{J}_2}\sum_{j_2^\prime\neq j_2}\dfrac{1}{|\mathcal{I}_k|}\sum_{i\in\mathcal{I}_k}|g_{1}(\bS_{1i})-\widehat{g}_{1}^{-k}(\bS_{1i})|\cdot I(\widehat{\mathcal{B}}_{j_1j_2i}\cap\mathcal{B}_{j_1j_2^\prime i})\bigg]\\
	&=O\bigg[\sum_{j_2\in\mathcal{J}_2}\sum_{j_2^\prime\neq j_2}\bigg\{\dfrac{1}{|\mathcal{I}_k|}\sum_{i\in\mathcal{I}_k}|g_{1}(\bS_{1i})-\widehat{g}_{1}^{-k}(\bS_{1i})|^2\bigg\}^{1/2}\bigg\{\dfrac{1}{|\mathcal{I}_k|}\sum_{i\in\mathcal{I}_k}I(\widehat{\mathcal{B}}_{j_1j_2i}\cap\mathcal{B}_{j_1j_2^\prime i})\bigg\}^{1/2}\bigg],
	\end{split}
	\end{equation}
	where the first equality holds, since $\cup_{j_2\in\mathcal{J}_2}\widehat{\mathcal{B}}_{j_1j_2i}$ is the universal set, the second and third equalities hold, since for any $j_2\neq j_2^\prime$, both $\widehat{\mathcal{B}}_{j_1j_2i}\cap \widehat{\mathcal{B}}_{j_1j_2^\prime i}$ and $\mathcal{B}_{j_1j_2i}\cap \mathcal{B}_{j_1j_2^\prime i}$ are empty sets, the fourth equality is by Assumption \ref{ass:uni_bound_} and $\|\bar{\bbeta}_{j_1j_2}^{*}\|_2=O(1)$ (by Assumption \ref{ass:uni_bound_} and the uniformly boundness of $Y_{j_1j_2i}^{2c*}$), the last equality is by Cauchy-Schwarz inequality. We next control $|\mathcal{I}_k|^{-1}\sum_{i\in\mathcal{I}_k}I(\widehat{\mathcal{B}}_{j_1j_2i}\cap\mathcal{B}_{j_1j_2^\prime i})$. For any pair $(j_2,j_2^{\prime})$ such that $j_2\neq j_2^{\prime}$ and any $\delta\in (0,1)$, by the convergence of $\widehat{\bbeta}_{j_1a_1j_2}$ and $\widehat{\bbeta}_{j_1a_1j_2^\prime}$, there exist a constant $C_{3\delta}$ depends on $\delta$ such that $P(\mathcal{C}_{j_2j_2^{\prime}}^{\delta})>1-\delta$, where 
	$$
	\mathcal{C}_{j_2j_2^{\prime}}^{\delta}=\left\{\sum_{a_1=0}^{1}(\|\widehat{\bbeta}_{j_1a_1j_{2}}-\bbeta^*_{j_1a_1j_{2}}\|_2+\|\widehat{\bbeta}_{j_1a_1j_{2}^\prime}-\bbeta^*_{j_1a_1j_{2}^\prime }\|_2)\le C_{3\delta}n^{-1/2}\right\}.
	$$
	Conditional on $\mathcal{C}_{j_2j_2^{\prime}}^{\delta}$, there exists a constant $C^{\prime\prime\prime}$ such that
	\begin{equation}\label{B_J2_J2p}
	\begin{split}
	&|\mathcal{I}_k|^{-1}\sum_{i\in\mathcal{I}_k}I(\widehat{\mathcal{B}}_{j_1j_2i}\cap\mathcal{B}_{j_1j_2^\prime i})\le |\mathcal{I}_k|^{-1}\sum_{i\in\mathcal{I}_k}I(\widehat{G}_{j_1j_2i}-\widehat{G}_{j_1j_2^{\prime}i}\ge 0, G^*_{j_1j_2i}-G^*_{j_1j_2^{\prime}i}\le 0)\\
	&\qquad\le 	|\mathcal{I}_k|^{-1}\sum_{i\in\mathcal{I}_k}I(|G^*_{j_1j_2i}-G^*_{j_1j_2^{\prime}i}|\le |G^*_{j_1j_2i}-G^*_{j_1j_2^\prime i}-\widehat{G}_{j_1j_2i}+\widehat{G}_{j_1j_2^{\prime}i}|)\\
	&\qquad\le 	|\mathcal{I}_k|^{-1}\sum_{i\in\mathcal{I}_k}I(|G_{j_1j_2i}^*-G^*_{j_1j_2^{\prime}i}|\le C^{\prime\prime\prime}C_{3\delta}n^{-1/2}),
	\end{split}
	\end{equation}
	where the second inequality is by $I(a\le 0, b\ge 0)\le I(ab\le 0)\le I(|a|\le |a-b|)$. By the law of large numbers, the above term coverage in probability to
	\begin{align*}
	&P(A_{1i}=0,|\bS_{\bar{l}_2}^\top(\bbeta^*_{j_10j_2}-\bbeta^*_{j_10j_2^\prime})|\le C^{\prime\prime\prime}C_{3\delta}n^{-1/2})\\
	&\qquad+	P(A_{1i}=1,|\bS_{\bar{l}_2}^\top(\bbeta^*_{j_11j_2}-\bbeta^*_{j_11j_2^\prime})|\le C^{\prime\prime\prime}C_{3\delta}n^{-1/2}).
	\end{align*}
	Under Assumption \ref{ass:margin2} with $r\ge 1$, the above term is bounded by
	\begin{align*}
	&	P(|\bS_{\bar{l}_2}^\top(\bbeta^*_{j_10j_2}-\bbeta^*_{j_10j_2^\prime})|\le C^{\prime\prime\prime}C_{3\delta}n^{-1/2})\\
	&\qquad+	P(|\bS_{\bar{l}_2}^\top(\bbeta^*_{j_11j_2}-\bbeta^*_{j_11j_2^\prime})|\le C^{\prime\prime\prime}C_{3\delta}n^{-1/2})=O(n^{-1/2}).
	\end{align*}
	Hence, 	$F_{27ko}=O_p(\|\widehat{g}_{1}^{-k}(\bS_{1i})-g_{1}(\bS_{1i})\|_{P,2}\cdot n^{-1/4})=o_p(n^{-1/2})$, and thus $F_{27}=o_p(1)$. For the term $F_{1}$, we have
	\begin{align}
	F_1&=(1/\sqrt{n})\sum\nolimits_{i\in[n]}\bS_{1i}\{A_{1i}-g_{1}(\bS_{1i})\}(\varepsilon_{1ti}^*+\Delta_{1i}^Y+\Delta_{4i}^Y)\nonumber\\
	&\qquad + (1/\sqrt{n})\sum\nolimits_{i\in[n]}\bS_{1i}\{A_{1i}-g_{1}(\bS_{1i})\}\Delta_{2i}^Y+ (1/\sqrt{n})\sum\nolimits_{i\in[n]}\bS_{1i}\{A_{1i}-g_{1}(\bS_{1i})\}\Delta_{3i}^Y\nonumber\\
	&\qquad+ (1/\sqrt{n})\sum\nolimits_{i\in[n]}\bS_{1i}\{A_{1i}-g_{1}(\bS_{1i})\}\Delta_{5i}^Y + (1/\sqrt{n})\sum\nolimits_{i\in[n]}\bS_{1i}\{A_{1i}-g_{1}(\bS_{1i})\}\Delta_{6i}^Y\nonumber\\
	&\qquad+ (1/\sqrt{n})\sum\nolimits_{i\in[n]}\bS_{1i}\{A_{1i}-g_{1}(\bS_{1i})\}\Delta_{7i}^Y\nonumber\\
	&=F_{11}+F_{12}+F_{13}+F_{14}+F_{15}+F_{16},\label{F1}
	\end{align}
	where the terms $\Delta_{1i}^Y,\cdots,\Delta_{7i}^Y$ are defined in \eqref{Delta_Y}. We next show that $F_{12}, F_{13},\cdots,F_{16}$ are all $o_p(1)$ terms. As described in \eqref{hat_Delta}, under Assumption \ref{ass:margin2} with $r>1$, it follows that $
	F_{12}=O\{(1/\sqrt{n})\sum\nolimits_{i\in[n]}|\Delta_{j_1j_2^{\rm f}i}^E|\}=o_p(1)$. Together with Assumptions \ref{ass:uni_bound_} and \ref{ass:margin2}, $F_{13}=O\{(1/\sqrt{n})\sum\nolimits_{i\in[n]}|\Delta_{j_1j_2^{\rm f}i}^E|\cdot \|\widehat{\balpha}-\bar\balpha^*\|_2\}=o_p(1)$. Consider the $o$-th element of $F_{15}$, we have
	\begin{align*}
	F_{15o}&=O\bigg\{(1/\sqrt{n})\sum_{i\in[n]}\sum_{j_2\in\mathcal{J}_2}|I(\mathcal{B}_{j_1j_2i})-I(\mathcal{B}_{j_1j_2i}^*)|\cdot|\widetilde{G}_{j_1j_2i}-\bar{G}_{j_1j_2i}^{*}|\bigg\}\\
	&=O\bigg\{ (1/\sqrt{n})\sum_{i\in[n]}\sum_{j_2\in\mathcal{J}_2}I(\exists j_2^\prime <j_2,G_{j_1j_2i}^*=G_{j_1j_2^\prime i}^* )|\widetilde{G}_{j_1j_2i}-G_{j_1j_2i}^{\prime*}|\bigg\}\\
	&= O\bigg\{	(1/\sqrt{n})\sum_{i\in[n]}\sum_{j_2\in\mathcal{J}_2}\sum_{j_2^\prime<j_2}I(G_{j_1j_2i}^*=G_{j_1j_2^\prime i}^* )\cdot\|\widetilde{\bbeta}_{j_1j_2}-\bar{\bbeta}_{j_1j_2}^*\|_2\bigg\}\\
	&=O_p\bigg\{\sum_{j_2\in\mathcal{J}_2}\sum_{j_2^\prime<j_2}P(G_{j_1j_2i}^*=G_{j_1j_2^\prime i}^*)\bigg\}=o_p(1),
	\end{align*}
	the last equality holds, since under Assumption \ref{ass:margin2}, $P(G_{j_1j_2i}^*=G_{j_1j_2^\prime i}^*)=P(\bs_{\bar{l}_2}^\top\bbeta^*_{j_1A_1j_2}=\bs_{\bar{l}_2}^\top\bbeta^*_{j_1A_1j_2^\prime})=o(1)$ for any $j_2^\prime<j_2$. Similar to the step controlling the term $F_{27}$ in \eqref{F26}, the $o$-th element of $F_{16}$
	\begin{align*}
	F_{16o}&=(1/\sqrt{n})\sum\nolimits_{i\in[n]}\bS_{1i}\{A_{1i}-g_{1}(\bS_{1i})\}\sum_{j_2\in\mathcal{J}_2}\{I(\widehat{\mathcal{B}}_{j_1j_2i})-I(\mathcal{B}_{j_1j_2i})\}\bar{G}^{*}_{j_1j_2i}\\
	&=(1/\sqrt{n})\sum\nolimits_{i\in[n]}\bS_{1i}\{A_{1i}-g_{1}(\bS_{1i})\}\sum_{j_2\in\mathcal{J}_2}\sum_{j_2^\prime\neq j_2}(\bar{G}^{*}_{j_1j_2i}-\bar{G}^{*}_{j_1j_2^\prime i})I(\widehat{\mathcal{B}}_{j_1j_2i}\cap\mathcal{B}_{j_1j_2^\prime i})\\
	&=O\bigg[(1/\sqrt{n})\sum\nolimits_{i\in[n]}\sum_{j_2\in\mathcal{J}_2}\sum_{j_2^\prime\neq j_2}I(\widehat{\mathcal{B}}_{j_1j_2i}\cap\mathcal{B}_{j_1j_2^\prime i})\bigg].
	\end{align*}
	As described in \eqref{B_J2_J2p}, for any $j_2\neq j_2^\prime$, $(1/\sqrt{n})\sum\nolimits_{i\in[n]}I(\widehat{\mathcal{B}}_{j_1j_2i}\cap\mathcal{B}_{j_1j_2^\prime i})=o_p(1)$ under Assumption \ref{ass:margin2} with $r>1$, hence $F_{16}=o_p(1)$. Similarily,
	\begin{align*}
		F_{14o}&=\dfrac{1}{\sqrt{n}}\sum_{i\in[n]}\bS_{1i}\{A_{1i}-g_{1}(\bS_{1i})\}\sum_{j_2\in\mathcal{J}_2}\{I(\widehat{\mathcal{B}}_{j_1j_2i})-I(\mathcal{B}_{j_1j_2i})\}(\widetilde{G}_{j_1j_2i}-\bar{G}^{*}_{j_1j_2i})\\
		&=\dfrac{1}{\sqrt{n}}\sum_{i\in[n]}\bS_{1i}\{A_{1i}-g_{1}(\bS_{1i})\}\sum_{j_2\in\mathcal{J}_2}\sum_{j_2\neq j_2^\prime}(\widetilde{G}_{j_1j_2i}-\bar{G}^{*}_{j_1j_2i}-\widetilde{G}_{j_1j_2^\prime i}+\bar{G}^{*}_{j_1j_2^\prime i})I(\widehat{\mathcal{B}}_{j_1j_2i}\cap \mathcal{B}_{j_1j_2^\prime i})\\
		&=O\bigg\{\dfrac{1}{\sqrt{n}}\sum_{i\in[n]}\sum_{j_2\in\mathcal{J}_2}\sum_{j_2^\prime\neq j_2}I(\widehat{\mathcal{B}}_{j_1j_2i}\cap\mathcal{B}_{j_1j_2^\prime i}) \|\widetilde{\bbeta}_{j_1j_2}-\bar{\bbeta}_{j_1j_2}^{*}\|_2\bigg\}=o_p(1),
	\end{align*}
	the last equality holds, since $(1/\sqrt{n})\sum\nolimits_{i\in[n]}I(\widehat{\mathcal{B}}_{j_1j_2i}\cap\mathcal{B}_{j_1j_2^\prime i})\le (1/\sqrt{n})\sum_{i\in [n]}I(\widehat{G}_{j_1j_2i}-\widehat{G}_{j_1j_2^{\prime}i}\ge 0, G^*_{j_1j_2i}-G^*_{j_1j_2^{\prime}i}\le 0)=O_p(1)$ under Assumption \ref{ass:margin2} with $r\ge 1$ following \eqref{B_J2_J2p} and $\|\widetilde{\bbeta}_{j_1j_2}-\bar{\bbeta}_{j_1j_2}^{*}\|_2=O_p(n^{-1/2})$. Consequently, we obtain $F_{14}=o_p(1)$. By repeating the argument used to show $\|\widehat{\bV}^{-1}-\bV^{-1}\|_\infty=o_p(n^{-1/2})$, it follows from assumption \ref{ass:PD_matrix_} that $\|\widehat{\bW}^{-1}-\bW^{-1}\|_\infty=o_p(n^{-1/2})$. Together with the upper bounds for $F_2,F_3,\cdots,F_7$ and $F_{12},F_{13},\cdots, F_{16}$,
	\begin{equation}\label{tilde_gamma}
		\sqrt{n}(\widetilde{\bgamma}_{j_1}-\breve{\bgamma}_{j_1})=\bW^{-1}F_{11}+o_p(1).
	\end{equation}
	\textbf{Control $\sqrt{n}(\breve{\bgamma}_{j_1}-\check{\bgamma}_{j_1})$.} By repeating the argument used to control $\sqrt{n}(\widetilde{\balpha}-\check{\balpha})$, we have \begin{equation}\label{tilde_gamma2}
		\sqrt{n}(\breve{\bgamma}_{j_1}-\check{\bgamma}_{j_1})=o_p(1).
	\end{equation}
	\textbf{Control $\sqrt{n}(\check{\bgamma}_{j_1}-\bar{\bgamma}_{j_1}^{*})$.} We have
	\begin{equation}\label{tilde_gamma3}
		\sqrt{n}(\check{\bgamma}_{j_1}-\bar{\bgamma}_{j_1}^{*})=\bW^{-1}\dfrac{1}{\sqrt{n}}\sum_{i\in [n]}\{A_{1i}-g_{1}(\bS_{1i})\}^2\bS_{1i}\{Q^{1t*}(\bS_{1i},j_1,1)-Q^{1t*}(\bS_{1i},j_1,0)-\bS_{1i}^\top\bar\bgamma_{j_1}^*\}.
	\end{equation}
	Combining the upper bounds for \eqref{tilde_gamma}, \eqref{tilde_gamma2}, and \eqref{tilde_gamma3},
	\begin{align*}
		\sqrt{n}(\widetilde{\bgamma}_{j_1}-\bar\bgamma_{j_1}^{*})&=\bW^{-1}\dfrac{1}{\sqrt{n}}\sum_{i\in[n]}\bS_{1i}\{A_{1i}-g_{1}(\bS_{1i})\}\bigg[Y_{j_1i}^{1t*}-f_{j_1}(\bS_{1i})-\{A_{1i}-g_{1}(\bS_{1i})\}\bS_{1i}^\top\bar\bgamma_{j_1}^*\\
		&\qquad+\bar{\bX}_{2i}^\top(\widehat{\balpha}-\bar\balpha^*)\widetilde{E}_{j_1j_2^{\rm f}}+\sum\nolimits_{j_2\in\mathcal{J}_2}I(\mathcal{B}_{j_1j_2i}^*)\bX_{\bar{l}_2^{\,\rm f}}^\top(\widetilde{\bbeta}_{j_1j_2}-\bar\bbeta^{*}_{j_1j_2})\bigg]+o_p(1)\\
		&\xrightarrow{\rm d}\mathcal{N}(0,\bM_{j_1}^{\bar\bgamma}),
	\end{align*}
	where
	\begin{align*}
		\bM_{j_1}^{\bar\bgamma}&=E\{{\rm var}(A_{1}\mid \bS_{1})\bS_{1}\bS_{1}^\top\}^{-1}E(\bQ_{j_1}^{\bar\bgamma}\bQ_{j_1}^{\bar\bgamma\top})E\{{\rm var}(A_{1}\mid \bS_1)\bS_{1}\bS_{1}^\top\}^{-1},\\
		\bQ_{j_1}^{\bar\bgamma}&=\{A_{1}-g_{1}(\bS_{1})\}\bS_{1}[Y_{j_1}^{1t*}-f_{j_1}(\bS_1)-\{A_{1}-g_{1}(\bS_{1})\}\bS_{1}^\top\bar\bgamma_{j_1}^{*}]\\
		&\quad+E[\{A_{1}-g_{1}(\bS_{1})\} \{I(\bS_{\bar{j}_{2}^{\rm f}}^\top\balpha^*_{j_1j_{2}^{\rm f}A_{1}}> 0)-A_{2}\}\bS_{1}\bar{\bX}_{2}^\top]E\{{\rm var}(A_{2}\mid \bar{\bX}_{2})\bar{\bX}_{2}\bar{\bX}_{2}^\top\}^{-1}\bQ^{\bar\balpha}\\
		&\quad+\sum_{j_2\in\mathcal{J}_2}E\left[\{A_{1}-g_{1}(\bS_{1})\}I(\bS_{\bar{l}_2}^\top\bbeta^*_{j_1A_1j_2}> \max_{j_2^\prime\neq j_2}\bS_{\bar{j}_2^\prime}^\top\bbeta_{j_1A_1j_2^\prime }^*)\bS_{1}\bX_{\bar{l}_2^{\,\rm f}}^\top\right]E(\bX_{\bar{l}_2^{\,\rm f}}\bX_{\bar{l}_2^{\,\rm f}}^\top)^{-1}\bQ_{j_1j_2}^{\bar\bbeta}.
	\end{align*}
	We next prove the asymptotic normality of $\widehat{\bgamma}_{j_1}$. Note that
	\begin{align}
	\sqrt{n}(\widehat{\bgamma}_{j_1}-\bgamma_{j_1}^*)&=\sqrt{n}\bigg(\dfrac{1}{n}\bS_{\bar{j}_1i}\bS_{\bar{j}_1i}^\top\bigg)^{-1}\bigg[\dfrac{1}{n}\sum_{i\in [n]}\bS_{\bar{j}_1i}\{\bS_1^\top\bar\bgamma^{*}_{j_1}-\bS_{\bar{j}_1i}^\top\bgamma^*_{j_1}+\bS_{1i}^\top(\widetilde{\bgamma}_{j_1}-\bar\bgamma_{j_1}^{*})\}\bigg]\label{gamma_hat_deco}\\
		&\xrightarrow{\rm d}\mathcal{N}(0,\bM_{j_1}^{\bgamma})\nonumber,
	\end{align}
	where
	\begin{align*}
	\bM_{j_1}^{\bgamma}&=E(\bS_{\bar{j}_1}\bS_{\bar{j}_1}^\top)^{-1}E(\bQ_{j_1}^{\bgamma}\bQ_{j_1}^{\bgamma\top})E(\bS_{\bar{j}_1}\bS_{\bar{j}_1}^\top)^{-1},\\
	\bQ_{j_1}^{\bgamma}&=\bS_{\bar{j}_1}(\bS_1^\top\bar\bgamma^{*}_{j_1}-\bS_{\bar{j}_1i}^\top\bgamma^*_{j_1})+E\{\bS_{\bar{j}_1}\bS_{1}^\top\}E\{{\rm var}(A_{1}\mid \bS_{1})\bS_{1}\bS_{1}^\top\}^{-1}	\bQ^{\bar\gamma}.
	\end{align*}
	
	\noindent\textbf{Step 3b: Consistency of $\widetilde{\bgamma}_{j_1}$ and $\widehat{\bgamma}_{j_1}$} 
	
	Now we prove the consistency of $\widetilde{\bgamma}_{j_1}$. As described in \eqref{hat_Delta}, it follows that $
	F_{12}=O\{(1/\sqrt{n})\sum\nolimits_{i\in[n]}|\Delta_{j_1j_2^{\rm f}i}^E|\}=O_p(1)$ under Assumption \ref{ass:margin2} with $r\ge 1$. We also obtain $F_{13}=O\{\sqrt{n}\|\widehat{\balpha}-\bar\balpha^*\|_2\}=O_p(1)$. From \eqref{B_J2_J2p}, for any $j_2\neq j_2^\prime$, $(1/\sqrt{n})\sum\nolimits_{i\in[n]}I(\widehat{\mathcal{B}}_{j_1j_2i}\cap\mathcal{B}_{j_1j_2^\prime i})=O_p(1)$ under Assumption \ref{ass:margin2} with $r\ge 1$, which implies $F_{16}=O_p(1)$. Moreover, $F_{14}=\sum_{j_2\in\mathcal{J}_2}\sum_{j_2^\prime\neq j_2} \|\widetilde{\bbeta}_{j_1j_2}-\bar\bbeta_{j_1j_2}^{*}\|_2=O_p(1)$. Combining the above results, we conclude that all of $F_{12}, F_{13}, F_{14}$, and $F_{16}$ are $O_p(1)$ terms. Together with the upper bounds for $F_{15}, F_2,F_3,\cdots,F_7$, and using \eqref{tilde_gamma2} and \eqref{tilde_gamma3}, it follows that $\sqrt{n}(\widetilde{\bgamma}_{j_1}-\bar\bgamma_{j_1}^*)=O_p(1)$. At this point, the quantity $n^{-1}\sum_{i\in [n]}\bS_{\bar{j}_1i}\bS_{1i}^\top(\widetilde{\bgamma}_{j_1}-\bar\bgamma_{j_1}^{\prime*})$ appearing in \eqref{gamma_hat_deco} is $O_p(n^{-1/2})$ term, we obtain that $\sqrt{n}(\widehat{\bgamma}_{j_1}-\bgamma_{j_1}^{*})=O_p(1)$.

	\noindent\textbf{Step 4: Asymptotic Normality $\widehat{\bdelta}_{j_1}$.} 
	
	Let $\Delta^{C,1c}_{j_1}=C^{1c}(j_1)-C^{1c}(j_1^{\rm f})$, we have
	\begin{align*}
	Y_{j_1}^{1c*}-Y_{j_1^{\rm f}}^{1c*}&=Y_{j_1}^{1t*}+\bS_1^\top\bar\bgamma_{j_1}^{*}[I\{\bS_{\bar{j}_1}^\top\bgamma_{j_1}^*>0\}-A_1]-Y_{j_1^{\rm f}}^{1t*}+\bS_1^\top\bar\bgamma_{j_1^{\rm f}}^{*}[I(\bS_1^\top\bgamma_{j_1^{\rm f}}^{*}>0)-A_1]-\Delta^{C,1c}_{j_1},\\
	\widehat{Y}_{j_1}^{1c}-\widehat{Y}_{j_1^{\rm f}}^{1c}&=\widehat{Y}_{j_1}^{1t}+\bS_1^\top\widetilde{\bgamma}_{j_1}[I\{\bS_{\bar{j}_1}^\top\widehat{\bgamma}_{j_1}>0\}-A_1]-\widehat{Y}_{j_1^{\rm f}}^{1t}+\bS_1^\top\widetilde{\bgamma}_{j_1^{\rm f}}[I(\bS_1^\top\widehat{\bgamma}_{j_1^{\rm f}}>0)-A_1]-\Delta^{C,1c}_{j_1}.
	\end{align*}
	It follows that
	\begin{align*}
	&\bS_1^\top\widetilde{\bgamma}_{j_1}[I\{\bS_{\bar{j}_1}^\top\widehat{\bgamma}_{j_1}>0\}-A_1]-\bS_1^\top\bar\bgamma_{j_1}^{*}[I\{\bS_{\bar{j}_1}^\top\bgamma_{j_1}^*>0\}-A_1]\\
	&=\bS_1^\top(\widetilde{\bgamma}_{j_1}-\bar{\bgamma}_{j_1}^*)[I\{\bS_{\bar{j}_1}^\top\bgamma^*_{j_1}>0\}-A_1]\\
	&\qquad + \bS_1^\top(\widetilde{\bgamma}_{j_1}-\bar{\bgamma}_{j_1}^*)[I\{\bS_{\bar{j}_1}^\top\widehat\bgamma_{j_1}>0\}-I\{\bS_{\bar{j}_1}^\top\bgamma^*_{j_1}>0\}]\\
	&\qquad +\bS_1^\top\bar{\bgamma}_{j_1}^*[I\{\bS_{\bar{j}_1}^\top\widehat\bgamma_{j_1}>0\}-I\{\bS_{\bar{j}_1}^\top\bgamma^*_{j_1}>0\}].
	\end{align*}
 	Recall that
	\begin{align*}
		\widehat{Y}_{j_1i}^{1t}-Y_{j_1i}^{1t*}&=\bar{\bX}_{2i}^\top(\widehat{\balpha}-\bar\balpha^*)\widetilde{E}_{j_1j_2^{\rm f}i}+\bar{\bX}_{2i}^\top\bar\balpha^*\Delta_{j_1j_2^{\rm f}i}^E+\bar{\bX}_{2i}^\top(\widehat{\balpha}-\bar\balpha^*)\Delta_{j_1j_2^{\rm f}i}^E\\
		&\ \ +\sum_{j_2\in\mathcal{J}_2}I(\mathcal{B}_{j_1j_2i}^*)(\widetilde{G}_{j_1j_2i}-\bar G^{*}_{j_1j_2i})+\sum_{j_2\in\mathcal{J}_2}\{I(\widehat{\mathcal{B}}_{j_1j_2i})-I(\mathcal{B}_{j_1j_2i})\}(\widetilde{G}_{j_1j_2i}-\bar G^{*}_{j_1j_2i})\\
		&\ \ +\sum_{j_2\in\mathcal{J}_2}\{I(\mathcal{B}_{j_1j_2i})-I(\mathcal{B}_{j_1j_2i}^*)\}(\widetilde{G}_{j_1j_2i}-\bar G^{*}_{j_1j_2i})+\sum_{j_2\in\mathcal{J}_2}\{I(\widehat{\mathcal{B}}_{j_1j_2i})-I(\mathcal{B}_{j_1j_2i})\}\bar{G}^{*}_{j_1j_2i},
	\end{align*}
	where $\widetilde{E}_{j_1j_{2}^{\rm f}i}=I(\bS_{\bar{j}_2^{\rm f}i}^\top\balpha^*_{j_1j_{2}^{\rm f}A_{1i}}> 0)-A_{2i}$, $\Delta_{j_1j_{2}^{\rm f}i}^E=I(\bS_{\bar{j}_2^{\rm f}i}^\top\widehat{\balpha}_{j_1A_{1i}j_{2}^{\rm f}}> 0)-I(\bS_{\bar{j}_2^{\rm f}i}^\top\balpha^*_{j_1A_{1i}j_{2}^{\rm f}}> 0)$,
		\begin{align*}
		\widehat{G}_{j_1j_2i}&=\bS_{\bar{l}_2i}^\top\widehat{\bbeta}_{j_1A_{1i}j_2},\;\; 	\widetilde{G}_{j_1j_2i}=\bX_{\bar{l}_2^{\,\rm f}i}^\top\widetilde{\bbeta}_{j_1j_2},\;\; 
		G^*_{j_1j_2i}=\bS_{\bar{l}_2i}^\top\bbeta^*_{j_1A_{1i}j_2},\;\;
		\bar G^{*}_{j_1j_2i}=\bX_{\bar{l}_2^{\,\rm f}i}^\top\bar\bbeta^{*}_{j_1j_2},\\
		\widehat{\mathcal{B}}_{j_1j_2i}&=\{\widehat{G}_{j_1j_2i}\ge \widehat{G}_{j_1j_2^\prime i}, \widehat{G}_{j_1j_2i}> \widehat{G}_{j_1j_2^{\prime\prime} i}, j_2^\prime<j_2, j_2^{\prime\prime}>j_2\},\\
		\mathcal{B}_{j_1j_2i}&=\{G_{j_1j_2i}^*\ge G^*_{j_1j_2^\prime i}, G^*_{j_1j_2i}> G^*_{j_1j_2^{\prime\prime} i}, j_2^\prime<j_2, j_2^{\prime\prime}>j_2\},\\
		\mathcal{B}_{j_1j_2i}^*&=\{G_{j_1j_2i}^*> G_{j_1j_2^\prime i}^*, j_2^{\prime}\neq j_2\}.
	\end{align*}	
	Consequently, we obtain
		\begin{align*}
		\widehat{\bdelta}_{j_1}-\bdelta^*_{j_1}
		&=\bZ^{-1}n^{-1}\sum_{i\in[n]}\bS_{l_1i}\left\{(Y_{j_1i}^{1c*}-Y_{j_1^{\rm f}i}^{1c*}-\bS_{l_1i}^\top\bdelta_{j_1}^*)+(\widehat{Y}_{j_1i}^{1c}-\widehat{Y}_{j_1^{\rm f}i}^{1c})-(Y_{j_1i}^{1c*}-Y_{j_1^{\rm f}i}^{1c*})\right\}\\
		&=\bZ^{-1}n^{-1}\sum_{i\in[n]}\bS_{l_1i}\bigg[(Y_{j_1i}^{1c*}-Y_{j_1^{\rm f}i}^{1c*}-\bS_{l_1i}^\top\bdelta_{j_1}^*)+\bar{\bX}_{2i}^\top(\widehat{\balpha}-\bar\balpha^*)(\widetilde{E}_{j_1j_2^{\rm f}i}-\widetilde{E}_{j_1^{\rm f}j_2^{\rm f}i})\\
		&\quad\quad+\sum_{j_2\in\mathcal{J}_2}\{I(\mathcal{B}_{j_1j_2i}^*)\bX_{\bar{l}_2^{\,\rm f}i}^\top(\widetilde{\bbeta}_{j_1j_2}-\bar\bbeta^{*}_{j_1j_2})-I(\mathcal{B}_{j_1^{\rm f}j_2i}^*)\bX_{\bar{l}_2^{\,\rm f}i}^\top(\widetilde{\bbeta}_{j_1^{\rm f}j_2}-\bar\bbeta^{*}_{j_1^{\rm f}j_2})\}\\
		&\quad\quad + \bS_{1i}^\top(\widetilde{\bgamma}_{j_1}-\bar\bgamma_{j_1}^{*})\{I(\bS_{\bar{j}_1}^\top\bgamma_{j_1}^*>0)-A_{1i}\}- \bS_{1i}^\top(\widetilde{\bgamma}_{j_1^{\rm f}}-\bar\bgamma_{j_1^{\rm f}}^{*})\{I(\bS_{1i}^\top\bgamma_{j_1^{\rm f}}^*>0)-A_{1i}\}\bigg]\\
		&\quad+\bZ^{-1}(H_2+H_3),
	\end{align*}
	where $\bZ=n^{-1}\sum_{i\in[n]}\bS_{l_1i}\bS_{l_1i}^\top$,
		\begin{align*}
		H_2&=n^{-1}\sum_{i\in[n]}\bS_{l_1i}\bigg[\bS_{1i}^\top(\widetilde{\bgamma}_{j_1}-\bar{\bgamma}_{j_1}^{*})\{I(\bS_{\bar{j}_1i}^\top\widehat{\bgamma}_{j_1}>0)-I(\bS_{\bar{j}_1i}^\top\bgamma_{j_1}^*>0)\}\\
		&\qquad -\bS_{1i}^\top(\widetilde{\bgamma}_{j_1^{\rm f}}-\bar\bgamma_{j_1^{\rm f}}^{*})\{I(\bS_{1i}^\top\widehat{\bgamma}_{j_1^{\rm f}}>0)-I(\bS_{1i}^\top\bgamma_{j_1^{\rm f}}^*>0)\}\\
		&\qquad+\bS_{1i}^\top\bar\bgamma_{j_1}^{*}\{I(\bS_{\bar{j}_1i}^\top\widehat{\bgamma}_{j_1}>0)-I(\bS_{\bar{j}_1}^\top\bgamma_{j_1}^*>0)\}- \bS_{1i}^\top\bar\bgamma_{j_1^{\rm f}}^{*}\{I(\bS_{1i}^\top\widehat{\bgamma}_{j_1^{\rm f}}>0)-I(\bS_{1i}^\top\bgamma_{j_1^{\rm f}}^*>0)\}\bigg],
	\end{align*}
	and
	\begin{align*}
	H_3&=n^{-1}\sum_{i\in[n]}\bS_{l_1i}\bigg[\bar{\bX}_{2i}^\top\bar\balpha^*\Delta_{j_1j_2^{\rm f}i}^E+\bar{\bX}_{2i}^\top(\widehat{\balpha}-\bar\balpha^*)\widehat{\Delta}_{j_1j_2^{\rm f}i}^E-\bar{\bX}_{2i}^\top\bar\balpha^*{\Delta}_{j_1^{\rm f}j_2^{\rm f}i}^E-\bar{\bX}_{2i}^\top(\widehat{\balpha}-\bar\balpha^*){\Delta}_{j_1^{\rm f}j_2^{\rm f}i}^E\\
	&\quad +\sum_{j_2\in\mathcal{J}_2}\{I(\widehat{\mathcal{B}}_{j_1j_2i})-I(\mathcal{B}_{j_1j_2i})\}(\widetilde{G}_{j_1j_2i}-\bar{G}^{*}_{j_1j_2i})+\sum_{j_2\in\mathcal{J}_2}\{I(\mathcal{B}_{j_1j_2i})-I(\mathcal{B}^*_{j_1j_2i})\}(\widetilde{G}_{j_1j_2i}-\bar G^{*}_{j_1j_2i})\\
	&\quad+\sum_{j_2\in\mathcal{J}_2}\{I(\widehat{\mathcal{B}}_{j_1j_2i})-I(\mathcal{B}_{j_1j_2i})\}\bar{G}^{*}_{j_1j_2i} -\sum_{j_2\in\mathcal{J}_2}\{I(\widehat{\mathcal{B}}_{j_1^{\rm f}j_2i})-I(\mathcal{B}_{j_1^{\rm f}j_2i})\}(\widetilde{G}_{j_1^{\rm f}j_2i}-\bar{G}^{*}_{j_1^{\rm f}j_2i})\\
	&\quad-\sum_{j_2\in\mathcal{J}_2}\{I(\mathcal{B}_{j_1^{\rm f}j_2i})-I(\mathcal{B}_{j_1^{\rm f}j_2i}^*)\}(\widetilde{G}_{j_1^{\rm f}j_2i}-\bar{G}^{*}_{j_1^{\rm f}j_2i})-\sum_{j_2\in\mathcal{J}_2}\{I(\widehat{\mathcal{B}}_{j_1^{\rm f}j_2i})-I(\mathcal{B}_{j_1^{\rm f}j_2i})\}\bar{G}^{*}_{j_1^{\rm f}j_2i}\bigg].
	\end{align*}
	By Assumptions \ref{ass:uni_bound_} and \ref{ass:PD_matrix_},
	\begin{align*}
		c_3\|\bar\bgamma^{*}_{j_1}\|_2^2&\le \lambda_{\rm min}[E\{{\rm var}(A_{1i}\mid\bS_{1i})\bS_{1i}\bS_{1i}^\top\}]\cdot \|\bar\bgamma^{*}_{j_1}\|_2^2\le E[\{A_{1i}-g_1(\bS_{1i})\}^2(\bS_{1i}^\top\bar\bgamma^{*}_{j_1})^2]\\
		&\le E[\{A_{2i}-g_1(\bS_{1i})\}^2\{Q^{1t*}(\bS_{1},j_1,1)-Q^{1t*}(\bS_{1},j_1,0)\}^2]=O(1),
	\end{align*}
	hence $\|\bar\bgamma^{*}_{j_1}\|_2=O(1)$. By repeating the steps used in deriving the bounds for $D_2$ and $D_3$, we obtain $H_2=o_p(n^{-1/2})$ under Assumption \ref{ass:margin2} with $r>1$ and $H_2=O_p(n^{-1/2})$ under Assumption \ref{ass:margin2} with $r\ge 1$. By analogy with the argument used to bound $F_{12}+F_{13}+F_{14}+F_{15}+F_{16}$ in \eqref{F1}, we can establish that $H_3=o_p(n^{-1/2})$ under Assumption \ref{ass:margin2} with $r>1$, and $H_3=O_p(n^{-1/2})$ under Assumption \ref{ass:margin2} with $r\ge 1$.  Moreover, it follows by Assumption \ref{ass:PD_matrix_} and the fact that $\bS_{l_1}$ is a subvector of $\bS_1$ that $\lambda_{\min}\{E(\bS_{l_1}\bS_{l_1}^\top)\}\ge \lambda_{\min}\{E(\bS_{1}\bS_{1}^\top)\}\ge \lambda_{\min}[E\{{\rm var}(A_{1}\mid \bS_{1})\bS_{1}\bS_{1}^\top\}]>c_2$. Using the same
	argument as for showing $\|\bU^{-1}\|_\infty=O_p(1)$, we can prove that $\|\bZ^{-1}\|_\infty=O_p(1)$. Consequently, it follows that under Assumption \ref{ass:margin2} with $r>1$
	\begin{align*}
		\sqrt{n}(\widehat{\bdelta}_{j_1}-\bdelta^*_{j_1})\xrightarrow{\rm d}\mathcal{N}(0,\bM_{j_1}^{\bdelta}),
	\end{align*}
 	where
	\begin{align*}
		\bM_{j_1}^{\bdelta}&=E(\bS_{l_1}\bS_{l_1}^\top)^{-1}E(\bQ^{\bdelta}_{j_1}\bQ^{\bdelta\top}_{j_1})E(\bS_{l_1}\bS_{l_1}^\top)^{-1},\\
		\bQ^{\bdelta}_{j_1}&=\bS_{l_1i}(Y_{j_1i}^{1c*}-Y_{j_1^{\rm f}i}^{1c*}-\bS_{l_1i}^\top\bdelta_{j_1}^*)+E\{(\widetilde{E}_{j_1j_2^{\rm f}}-\widetilde{E}_{j_1^{\rm f}j_2^{\rm f}})\bS_{l_1}\bar{\bX}_2^\top\}E\{{\rm var}(A_{2}\mid \bar{\bX}_{2})\bar{\bX}_{2}\bar{\bX}_{2}^\top\}^{-1}\bQ^{\bar\balpha}\\
		& +\sum\nolimits_{j_2\in\mathcal{J}_2}E\left[I(\bS_{\bar{l}_2}^\top\bbeta^*_{j_1A_1j_2}> \max_{j_2^\prime\neq j_2}\bS_{\bar{l}_2}^\top\bbeta_{j_1A_1j_2^\prime}^*)\bS_{l_1}\bX_{\bar{l}_2^{\,\rm f}}^\top\right]E(\bX_{\bar{l}_2^{\,\rm f}}\bX_{\bar{l}_2^{\,\rm f}}^\top)^{-1}\bQ_{j_1j_2}^{\bar\bbeta}\\
		&-\sum\nolimits_{j_2\in\mathcal{J}_2}E\left[I(\bS_{1}^\top\bbeta^*_{j_1^{\rm f}A_1j_2}> \max_{j_2^\prime\neq j_2}\bS_{1}^\top\bbeta_{j_1^{\rm f}A_1j_2^\prime }^*)\bS_{l_1}\bX_{\bar{l}_2^{\,\rm f}}^\top\right]E(\bX_{\bar{l}_2^{\,\rm f}}\bX_{\bar{l}_2^{\,\rm f}}^\top)^{-1}\bQ_{j_1^{\rm f}j_2}^{\bar\bbeta}\\
		&+E[\{I(\bS_{\bar{j}_1}^\top\bgamma_{j_1}^*>0)-A_1\}\bS_{l_1}\bS_1^\top]E\{{\rm var}(A_{1}\mid \bS_1)\bS_{1}\bS_{1}^\top\}^{-1}\bQ_{j_1}^{\bar\bgamma}\\
		&-E[\{I(\bS_{1}^\top\bgamma_{j_1^{\rm f}}^*>0)-A_1\}\bS_{l_1}\bS_1^\top]E\{{\rm var}(A_{1}\mid \bS_1)\bS_{1}\bS_{1}^\top\}^{-1}\bQ_{j_1^{\rm f}}^{\bar\bgamma}.
	\end{align*}
	Moreover, under Assumption \ref{ass:margin2} with $r\ge 1$, $\|\widehat{\bdelta}_{j_1}-\bdelta_{j_1}^*\|_2=O_p(n^{1/2})$.
	\end{proof}

	\begin{proof}[Proof of Theorem \ref{thm:regret}]
		We denote $J_1=\hat{\pi}^{1c}(\bS_{l_1})$, $J_2=\hat{\pi}^{2c}(\bS_{\bar{J}_1},J_1,\hat{\pi}^{1t}(\bS_{\bar{J}_1},J_1))$. By the proof of Theorem \ref{thm:optimal_DTR}, we obtain
		\begin{align*}
			{\rm Profit}(\check\pi^{1c},\check\pi^{1t},\check\pi^{2c},\check\pi^{2t})-{\rm Profit}(\hat{\pi}^{1c},\hat{\pi}^{1t},\hat{\pi}^{2c},\hat{\pi}^{2t})=\hat{G}^{1c}+\hat{G}^{1t}+\hat{G}^{2c}+\hat{G}^{2t},
		\end{align*}
		where
		\begin{align*}
			\hat{G}^{2t}&=E[Q^{2t}\{\bS_{\bar{J}_2},J_1,\hat{\pi}^{1t}(\bS_{\bar{J}_1},J_1),J_2,\check{\pi}^{2t}(\bS_{\bar{J}_2},J_1,\hat{\pi}^{1t}(\bS_{\bar{J}_1},J_1),J_2)\}\\
			&\qquad\qquad-Q^{2t}\{\bS_{\bar{J}_2},J_1,\hat{\pi}^{1t}(\bS_{\bar{J}_1},J_1),J_2,\hat{\pi}^{2t}(\bS_{\bar{J}_2},J_1,\hat{\pi}^{1t}(\bS_{\bar{J}_1},J_1),J_2)\}],\\
			\hat{G}^{2c}&=E[Q^{2c}\{\bS_{\bar{L}_2},J_1,\hat{\pi}^{1t}(\bS_{\bar{J}_1},J_1),\check{\pi}^{2c}(\bS_{\bar{L}_2},J_1,\hat{\pi}^{1t}(\bS_{\bar{J}_1},J_1))\}\\
			&\qquad\qquad-Q^{2c}\{\bS_{\bar{L}_2},J_1,\hat{\pi}^{1t}(\bS_{\bar{J}_1},J_1),\hat{\pi}^{2c}(\bS_{\bar{L}_2},J_1,\hat{\pi}^{1t}(\bS_{\bar{J}_1},J_1))\}],\\
			\hat{G}^{1t}&=E[Q^{1t}\{\bS_{\bar{J}_1},J_1,\check{\pi}^{1t}(\bS_{\bar{J}_1},J_1)\}-Q^{1t}\{\bS_{\bar{J}_1},J_1,\hat{\pi}^{1t}(\bS_{\bar{J}_1},J_1)\}],\\
			\hat{G}^{1c}&=E[Q_1^{1c}\{\bS_{l_1},\check{\pi}^{1c}(\bS_{l_1})\}-Q_1^{1c}\{\bS_{l_1},\hat{\pi}^{1c}(\bS_{l_1})\}].
		\end{align*}
		We next control the four terms above, respectively. For simplicity, we use $\hat{\pi}^{1t}$ denote $\hat{\pi}^{1t}(\bS_{\bar{J}_1},J_1)$,
		\begin{align*}
			&E\left[\Delta^{2t}(\bS_{\bar{J}_2},J_1,\hat{\pi}^{1t},J_2)\{I(\Delta^{2t}(\bS_{\bar{J}_2},J_1,\hat{\pi}^{1t},J_2)>0)-I(\bS_{\bar{J}_2}^\top\widehat{\balpha}_{J_1\hat{\pi}^{1t}J_2}>0)\}\right]\\
			&=E\left[\Delta^{2t}(\bS_{\bar{J}_2},J_1,\hat{\pi}^{1t},J_2)\{I(\Delta^{2t}(\bS_{\bar{J}_2},J_1,\hat{\pi}^{1t},J_2)>0)-I(\bS_{\bar{J}_2}^\top\balpha^*_{J_1\hat{\pi}^{1t}J_2}>0)\}\right]\\
			&\qquad+E\left[\bS_{\bar{J}_2}^\top\balpha^*_{J_1\hat{\pi}^{1t}J_2}\{I(\bS_{\bar{J}_2}^\top\balpha^*_{J_1\hat{\pi}^{1t}J_2}>0)-I(\bS_{\bar{J}_2}^\top\widehat{\balpha}_{J_1\hat{\pi}^{1t}J_2}>0)\}\right]\\
			&\qquad+E\left[\{\Delta^{2t}(\bS_{\bar{J}_2},J_1,\hat{\pi}^{1t},J_2)-\bS_{\bar{J}_2}^\top\balpha^*_{J_1\hat{\pi}^{1t}J_2}\}\{I(\bS_{\bar{J}_2}^\top\balpha^*_{J_1\hat{\pi}^{1t}J_2}>0)-I(\bS_{\bar{J}_2}^\top\widehat{\balpha}_{J_1\hat{\pi}^{1t}J_2}>0)\}\right]\\
			&=T_1+T_2+T_3.
		\end{align*}
		Since $\bS_{\bar{J}_2}^\top\balpha^*_{J_1\hat{\pi}^{1t}J_2}\{I(\Delta^{2t}(\bS_{\bar{J}_2},J_1,\hat{\pi}^{1t},J_2)>0)-I(\bS_{\bar{J}_2}^\top\balpha^*_{J_1\hat{\pi}^{1t}J_2}>0)\}\le 0$, the first term can be bounded by
		\begin{align*}
			T_1&\le E\left[\{\Delta^{2t}(\bS_{\bar{J}_2},J_1,\hat{\pi}^{1t},J_2)-\bS_{\bar{J}_2}^\top\balpha^*_{J_1\hat{\pi}^{1t}J_2}\}\{I(\Delta^{2t}(\bS_{\bar{J}_2},J_1,\hat{\pi}^{1t},J_2)>0)-I(\bS_{\bar{J}_2}^\top\balpha^*_{J_1\hat{\pi}^{1t}J_2}>0)\}\right]\\
			&\le E\big[|\Delta^{2t}(\bS_{\bar{J}_2},J_1,\hat{\pi}^{1t},J_2)-\bS_{\bar{J}_2}^\top\balpha^*_{J_1\hat{\pi}^{1t}J_2}|\\
			&\qquad\qquad\times I\{|\bS_{\bar{J}_2}^\top\balpha^*_{J_1\hat{\pi}^{1t}J_2}|\le |\Delta^{2t}(\bS_{\bar{J}_2},J_1,\hat{\pi}^{1t},J_2)-\bS_{\bar{J}_2}^\top\balpha^*_{J_1\hat{\pi}^{1t}J_2}|\}\big]\\
			&\le \epsilon_n\cdot P\left\{|\bS_{\bar{J}_2}^\top\balpha^*_{J_1\hat{\pi}^{1t}J_2}|\le \epsilon_n\right\}\\
			&=\epsilon_n\sum_{j_1,j_2,a_1} P\{|\bS_{\bar{j}_2}^\top\balpha^*_{j_1a_1j_2}|\le \varepsilon,J_1=j_1,J_2=j_2,\hat{\pi}^{1t}=a_1\}=O(\epsilon_n^{r+1}),
		\end{align*}
		where the second inequality holds, since for any $a,b>0$, we have $|a-b|\ge |a|$ whenever $ab<0$, the third inequality is by Assumption \ref{ass:margin2}. It follows by $\|\widehat{\balpha}_{j_1a_1j_2}-\balpha^*_{j_1a_1j_2}\|_2=O_p(n^{-1/2})$, Assumptions \ref{ass:uni_bound_}, and \ref{ass:margin2} with $r\ge 1$ that
		\begin{align*}
			T_2&\le E\left[|\bS_{\bar{J}_2}^\top\balpha^*_{J_1\hat{\pi}^{1t}J_2}|\cdot|I(\bS_{\bar{J}_2}^\top\balpha^*_{J_1\hat{\pi}^{1t}J_2}>0)-I(\bS_{\bar{J}_2}^\top\widehat\balpha_{J_1\hat{\pi}^{1t}J_2}>0)|\right]\\
			&\le E\left[|\bS_{\bar{J}_2}^\top\balpha^*_{J_1\hat{\pi}^{1t}J_2}|\cdot I\{|\bS_{\bar{J}_2}^\top\balpha^*_{J_1\hat{\pi}^{1t}J_2}|\le |\bS_{\bar{J}_2}^\top(\widehat\balpha_{J_1\hat{\pi}^{1t}J_2}-\balpha^*_{J_1\hat{\pi}^{1t}J_2})|\}\right]\\
			&\le E\left[|\bS_{\bar{J}_2}^\top(\widehat\balpha_{J_1\hat{\pi}^{1t}J_2}-\balpha^*_{J_1\hat{\pi}^{1t}J_2})|\cdot I\{|\bS_{\bar{J}_2}^\top\balpha^*_{J_1\hat{\pi}^{1t}J_2}|\le |\bS_{\bar{J}_2}^\top(\widehat\balpha_{J_1\hat{\pi}^{1t}J_2}-\balpha^*_{J_1\hat{\pi}^{1t}J_2})|\}\right]\\
			&=O\left[\|\widehat\balpha_{J_1\hat{\pi}^{1t}J_2}-\balpha^*_{J_1\hat{\pi}^{1t}J_2}\|_2\cdot P\{|\bS_{\bar{J}_2}^\top\balpha^*_{J_1\hat{\pi}^{1t}J_2}|\le |\bS_{\bar{J}_2}^\top(\widehat\balpha_{J_1\hat{\pi}^{1t}J_2}-\balpha^*_{J_1\hat{\pi}^{1t}J_2})|\}\right]\\
			&=O_p\{(1/\sqrt{n})^{r+1}\}.
		\end{align*}
		Moreover, 
		\begin{align*}
			T_3&\le E\left[|\Delta^{2t}(\bS_{\bar{J}_2},J_1,\hat{\pi}^{1t},J_2)-\bS_{\bar{J}_2}^\top\balpha^*_{J_1\hat{\pi}^{1t}J_2}|\cdot |I(\bS_{\bar{J}_2}^\top\balpha^*_{J_1\hat{\pi}^{1t}J_2}>0)-I(\bS_{\bar{J}_2}^\top\widehat\balpha_{J_1\hat{\pi}^{1t}J_2}>0)|\right]\\
			&\le E\left[|\Delta^{2t}(\bS_{\bar{J}_2},J_1,\hat{\pi}^{1t},J_2)-\bS_{\bar{J}_2}^\top\balpha^*_{J_1\hat{\pi}^{1t}J_2}|\cdot I\{|\bS_{\bar{J}_2}^\top\balpha^*_{J_1\hat{\pi}^{1t}J_2}|\le |\bS_{\bar{J}_2}^\top(\widehat\balpha_{J_1\hat{\pi}^{1t}J_2}-\balpha^*_{J_1\hat{\pi}^{1t}J_2})|\}\right]\\
			&\le\epsilon_n\cdot P\{|\bS_{\bar{J}_2}^\top\balpha^*_{J_1\hat{\pi}^{1t}J_2}|\le |\bS_{\bar{J}_2}^\top(\widehat\balpha_{J_1\hat{\pi}^{1t}J_2}-\balpha^*_{J_1\hat{\pi}^{1t}J_2})|\}=O_p\{\epsilon_n(1/\sqrt{n})^r\}.
		\end{align*}
		Consequently, $\hat{G}^{2t}=O_p\{(1/\sqrt{n})^{r+1}+\epsilon_n(1/\sqrt{n})^r+\epsilon_n^{r+1}\}$. As for the second term,
		\begin{align*}
			\hat{G}^{2c}
			&= E\bigg[\sum\nolimits_{j_2\in\mathcal{J}_2}\Delta^{2c}(\bS_{\bar{L}_2},J_1,\hat{\pi}^{1t},j_2)\{I(\mathcal{D}_{J_1j_2})-I(\mathcal{E}_{J_1j_2})\}\bigg]\\
			&=E\bigg[\sum\nolimits_{j_2\in\mathcal{J}_2}\Delta^{2c}(\bS_{\bar{L}_2},J_1,\hat{\pi}^{1t},j_2)\{I(\mathcal{D}_{J_1j_2})-I(\mathcal{F}_{J_1j_2})\}\bigg]\\
			&\qquad+E\bigg[\sum\nolimits_{j_2\in\mathcal{J}_2}\bS_{\bar{L}_2}^\top\bbeta_{J_1\hat{\pi}^{1t}j_2}^*\{I(\mathcal{F}_{J_1j_2})-I(\mathcal{E}_{J_1j_2})\}\bigg]\\
			&\qquad +E\bigg[\sum\nolimits_{j_2\in\mathcal{J}_2}\{\Delta^{2c}(\bS_{\bar{L}_2},J_1,\hat{\pi}^{1t},j_2)-\bS_{\bar{L}_2}^\top\bbeta_{J_1\hat{\pi}^{1t}j_2}^*\}\{I(\mathcal{F}_{J_1j_2})-I(\mathcal{E}_{J_1j_2})\}\bigg]\\
			&=T_4+T_5+T_6,
		\end{align*}
		where
		\begin{align*}
			\mathcal{D}_{J_1j_2}&=\{\Delta^{2c}(\bS_{\bar{L}_2},J_1,\hat{\pi}^{1t},j_2)\ge \Delta^{2c}(\bS_{\bar{L}_2},J_1,\hat{\pi}^{1t},j_2^\prime),\\
			&\qquad\qquad\Delta^{2c}(\bS_{\bar{L}_2},J_1,\hat{\pi}^{1t},j_2)>\Delta^{2c}(\bS_{\bar{L}_2},J_1,\hat{\pi}^{1t},j_2^{\prime\prime}),j_2^\prime<j_2,j_2^{\prime\prime}>j_2\},\\
			\mathcal{E}_{J_1j_2}&=\{\bS_{\bar{L}_2}^\top\widehat\bbeta_{J_1\hat{\pi}^{1t}j_2}\ge \bS_{\bar{L}_2}^\top\widehat\bbeta_{J_1\hat{\pi}^{1t}j_2^\prime},\bS_{\bar{L}_2}^\top\widehat\bbeta_{J_1\hat{\pi}^{1t}j_2}>\bS_{\bar{L}_2}^\top\widehat\bbeta_{J_1\hat{\pi}^{1t}j_2^{\prime\prime}},j_2^\prime<j_2,j_2^{\prime\prime}>j_2\},\\
			\mathcal{F}_{J_1j_2}&=\{\bS_{\bar{L}_2}^\top\bbeta_{J_1\hat{\pi}^{1t}j_2}^*\ge \bS_{\bar{L}_2}^\top\bbeta_{J_1\hat{\pi}^{1t}j_2^\prime}^*,\bS_{\bar{L}_2}^\top\bbeta_{J_1\hat{\pi}^{1t}j_2}^*>\bS_{\bar{L}_2}^\top\bbeta_{J_1\hat{\pi}^{1t}j_2^{\prime\prime}}^*,j_2^\prime<j_2,j_2^{\prime\prime}>j_2\}.
		\end{align*}
		Here we consider the case where more than one element in the sequence $\{\Delta^{2c}(\bS_{\bar{L}_2},J_1,\hat{\pi}^{1t},j_2)\}_{j_2\in\mathcal{J}_2}$ attain the maximum value, and the event $\mathcal{D}_{J_1j_2}$ denotes that the last occurrence of the maximum value is at index $j_2$. The events $\mathcal{E}_{J_1j_2}$ and $\mathcal{F}_{J_1j_2}$ are defined similarly. We first obtain that the events $\mathcal{F}_{J_1j_2}$ and $\mathcal{F}_{J_1j_2^{\prime}}$ are disjoint for any pairs $j_2\neq j_2^\prime$. Consequently, $\sum_{j_2}\bS_{\bar{L}_2}^\top\bbeta_{J_1\hat{\pi}^{1t}j_2}^*I(\mathcal{F}_{J_1j_2})=\max_{j_2}\bS_{\bar{L}_2}^\top\bbeta_{J_1\hat{\pi}^{1t}j_2}^*\ge \sum_{J_2}\bS_{\bar{L}_2}^\top\bbeta_{J_1\hat{\pi}^{1t}j_2}^*I(\mathcal{D}_{J_1j_2})$, which implies that $\bS_{\bar{L}_2}^\top\bbeta_{J_1\hat{\pi}^{1t}j_2}^*\{I(	\mathcal{D}_{J_1j_2})-I(\mathcal{F}_{J_1j_2})\}\le 0$. Moreover, since both $\cup_{j_2\in\mathcal{J}_2} \mathcal{D}_{J_1j_2}$ and $\cup_{j_2\in\mathcal{J}_2} \mathcal{F}_{J_1j_2}$ are full sets, as well as $\mathcal{D}_{J_1j_2}\cap \mathcal{D}_{J_1j_2^\prime}$ and $\mathcal{F}_{J_1j_2}\cap \mathcal{F}_{J_1j_2^\prime}$ are empty sets for $j_2\neq j_2^\prime$,
		\begin{align*}
			T_4&\le E\bigg[\sum_{j_2\in\mathcal{J}_2}\{\Delta^{2c}(\bS_{\bar{L}_2},J_1,\hat{\pi}^{1t},j_2)-\bS_{\bar{L}_2}^\top\bbeta_{J_1\hat{\pi}^{1t}j_2}^*\}\{I(	\mathcal{D}_{J_1j_2})-I(\mathcal{F}_{J_1j_2})\}\bigg]\\
			&=E\bigg(\sum_{j_2\in\mathcal{J}_2}\bigg[\sum_{j_2^{\prime\prime}\in\mathcal{J}_2}\{\Delta^{2c}(\bS_{\bar{L}_2},J_1,\hat{\pi}^{1t},j_2^{\prime\prime})-\bS_{\bar{L}_2}^\top\bbeta_{J_1\hat{\pi}^{1t}j_2^{\prime\prime}}^*\}I(\mathcal{D}_{J_1j_2^{\prime\prime}})\\
			&\qquad\qquad\qquad\quad-\sum_{j_2^\prime\in\mathcal{J}_2}\{\Delta^{2c}(\bS_{\bar{L}_2},J_1,\hat{\pi}^{1t},j_2^\prime)-\bS_{\bar{L}_2}^\top\bbeta_{J_1\hat{\pi}^{1t}j_2^\prime}^*\}I(\mathcal{F}_{J_1j_2^\prime})\bigg]I(\mathcal{D}_{J_1j_2})\bigg)\\
			&=E\bigg(\sum_{j_2\in\mathcal{J}_2}\bigg[\{\Delta^{2c}(\bS_{\bar{L}_2},J_1,\hat{\pi}^{1t},j_2)-\bS_{\bar{L}_2}^\top\bbeta_{J_1\hat{\pi}^{1t}j_2}^*\}\\
			&\qquad\qquad\qquad\quad-\sum_{j_2^\prime\in\mathcal{J}_2}\{\Delta^{2c}(\bS_{\bar{L}_2},J_1,\hat{\pi}^{1t},j_2^\prime)-\bS_{\bar{L}_2}^\top\bbeta_{J_1\hat{\pi}^{1t}j_2^\prime}^*\}I(\mathcal{F}_{J_1j_2^\prime})\bigg]I(\mathcal{D}_{J_1j_2})\bigg)\\
			&=E\bigg[\sum_{j_2\in\mathcal{J}_2}\sum_{j_2^\prime\neq J_2}\{\Delta^{2c}(\bS_{\bar{L}_2},J_1,\hat{\pi}^{1t},j_2)-\bS_{\bar{L}_2}^\top\bbeta_{J_1\hat{\pi}^{1t}j_2}^*\\
			&\qquad\qquad\qquad\quad-\Delta^{2c}(\bS_{\bar{L}_2},J_1,\hat{\pi}^{1t},j_2^\prime)+\bS_{\bar{L}_2}^\top\bbeta_{J_1\hat{\pi}^{1t}j_2^\prime}^*\}\{I(	\mathcal{D}_{J_1j_2}\cap \mathcal{F}_{J_1j_2^\prime})\}\bigg]\\
			&\le E\bigg[\sum_{j_2\in\mathcal{J}_2}\sum_{j_2^\prime\neq j_2}2\epsilon_n\cdot I(	\mathcal{D}_{J_1j_2}\cap \mathcal{F}_{J_1j_2^\prime})\bigg].
		\end{align*}
		For each pairs $j_2\neq j_2^\prime$,
		\begin{align*}
		&E\{I(\mathcal{D}_{J_1j_2}\cap \mathcal{F}_{J_1j_2^\prime})\}\le P\{\Delta^{2c}(\bS_{\bar{L}_2},J_1,\hat{\pi}^{1t},j_2)\ge \Delta^{2c}(\bS_{\bar{L}_2},J_1,\hat{\pi}^{1t},j_2^\prime),\bS_{\bar{L}_2}^\top\bbeta_{J_1\hat{\pi}^{1t}j_2^\prime}^*\ge \bS_{\bar{L}_2}^\top\bbeta_{J_1\hat{\pi}^{1t}j_2}^*\}\\
		&\qquad\le P\{|\bS_{\bar{L}_2}^\top(\bbeta_{J_1\hat{\pi}^{1t}j_2}^*-\bbeta_{J_1\hat{\pi}^{1t}j_2^\prime}^*)|\\
		&\qquad\qquad\le |\bS_{\bar{L}_2}^\top(\bbeta_{\hat{J}_1J_2\hat{a}_1}^*-\bbeta_{\hat{J}_1J_2^\prime \hat{a}_1}^*)-\Delta^{2c}(\bS_{\bar{L}_2},J_1,\hat{\pi}^{1t},j_2)+\Delta^{2c}(\bS_{\bar{L}_2},J_1,\hat{\pi}^{1t},j_2^\prime)|\}\\
		&\qquad\le P\{|\bS_{\bar{L}_2}^\top(\bbeta_{J_1\hat{\pi}^{1t}j_2}^*-\bbeta_{J_1\hat{\pi}^{1t}j_2^\prime}^*)|\le 2\epsilon_n\} = O(\epsilon_n^r),
		\end{align*}
		where the second inequality holds, since $I(a\le 0, b\ge 0)\le I(ab\le 0)\le I(|a|\le |a-b|)$, the last equality is by Assumption \ref{ass:margin2}. Hence $T_4=\varepsilon^{r+1}$.
		Since $\cup_{j_2\in\mathcal{J}_2}\mathcal{E}_{J_1j_2}$ is the full set, using the same argument as for $T_4$ yields
		\begin{align*}
			T_5
			&=E\bigg\{\sum_{j_2\in\mathcal{J}_2}\sum_{j_2^\prime\neq j_2}\bS_{\bar{L}_2}^\top(\bbeta_{J_1\hat{\pi}^{1t}j_2}^*-\bbeta_{J_1\hat{\pi}^{1t}j_2^\prime}^*)I(\mathcal{F}_{J_1j_2}\cap \mathcal{E}_{J_1j_2^\prime})\bigg\}.
		\end{align*}
		For each pairs $j_2\neq j_2^\prime$, 
		\begin{align*}
			&E\bigg\{\bS_{\bar{L}_2}^\top(\bbeta_{J_1\hat{\pi}^{1t}j_2}^*-\bbeta_{J_1\hat{\pi}^{1t}j_2^\prime}^*)I(\mathcal{F}_{J_1j_2}\cap \mathcal{E}_{J_1j_2^\prime})\bigg\}\\
			\le\ & E\bigg[\bS_{\bar{L}_2}^\top(\bbeta_{J_1\hat{\pi}^{1t}j_2}^*-\bbeta_{J_1\hat{\pi}^{1t}j_2^\prime}^*)I\{\bS_{\bar{L}_2}^\top(\bbeta_{J_1\hat{\pi}^{1t}j_2^\prime}^*-\bbeta_{J_1\hat{\pi}^{1t}j_2}^*)\le 0, \bS_{\bar{L}_2}^\top(\widehat\bbeta_{J_1\hat{\pi}^{1t}j_2^\prime}-\widehat\bbeta_{J_1\hat{\pi}^{1t}j_2}^*)\ge 0 \}\bigg]\\
			\le\ & E\bigg[|\bS_{\bar{L}_2}^\top(\bbeta_{J_1\hat{\pi}^{1t}j_2}^*-\bbeta_{J_1\hat{\pi}^{1t}j_2^\prime}^*)|\\
			&\qquad \times I\{|\bS_{\bar{L}_2}^\top(\bbeta_{J_1\hat{\pi}^{1t}j_2^\prime}^*-\bbeta_{J_1\hat{\pi}^{1t}j_2}^*)|\le |\bS_{\bar{L}_2}^\top(\bbeta^*_{J_1\hat{\pi}^{1t}j_2^\prime}-\widehat\bbeta_{J_1\hat{\pi}^{1t}j_2^\prime})-\bS_{\bar{L}_2}^\top(\bbeta^*_{J_1\hat{\pi}^{1t}j_2}-\widehat\bbeta_{J_1\hat{\pi}^{1t}j_2})| \}
			\bigg]\\
			\le\ & E\bigg[|\bS_{\bar{L}_2}^\top(\bbeta^*_{J_1\hat{\pi}^{1t}j_2^\prime}-\widehat\bbeta_{J_1\hat{\pi}^{1t}j_2^\prime})-\bS_{\bar{L}_2}^\top(\bbeta^*_{J_1\hat{\pi}^{1t}j_2}-\widehat\bbeta_{J_1\hat{\pi}^{1t}j_2})|\\
&\qquad \times I\{|\bS_{\bar{L}_2}^\top(\bbeta_{J_1\hat{\pi}^{1t}j_2^\prime}^*-\bbeta_{J_1\hat{\pi}^{1t}j_2}^*)|\le |\bS_{\bar{L}_2}^\top(\bbeta^*_{J_1\hat{\pi}^{1t}j_2^\prime}-\widehat\bbeta_{J_1\hat{\pi}^{1t}j_2^\prime})-\bS_{\bar{L}_2}^\top(\bbeta^*_{J_1\hat{\pi}^{1t}j_2}-\widehat\bbeta_{J_1\hat{\pi}^{1t}j_2})| \}
\bigg]\\
			=&\ O\bigg[(\|\widehat\bbeta_{J_1\hat{\pi}^{1t}j_2^\prime}-\bbeta^*_{J_1\hat{\pi}^{1t}j_2^\prime}\|_2+\|\widehat\bbeta_{J_1\hat{\pi}^{1t}j_2}-\bbeta^*_{J_1\hat{\pi}^{1t}j_2}\|_2)\\
			&\qquad \times P\{|\bS_{\bar{L}_2}^\top(\bbeta_{J_1\hat{\pi}^{1t}j_2^\prime}^*-\bbeta_{J_1\hat{\pi}^{1t}j_2}^*)|\le |\bS_{\bar{L}_2}^\top(\bbeta^*_{J_1\hat{\pi}^{1t}j_2^\prime}-\widehat\bbeta_{J_1\hat{\pi}^{1t}j_2^\prime})-\bS_{\bar{L}_2}^\top(\bbeta^*_{J_1\hat{\pi}^{1t}j_2}-\widehat\bbeta_{J_1\hat{\pi}^{1t}j_2})| \}\bigg]\\
			=&\ O_p\{(1/\sqrt{n})^{r+1}\},
		\end{align*}
		where the fourth equality is by $\|\widehat{\bbeta}_{j_1a_1j_2}-\bbeta_{j_2a_1j_2}^*\|_2=O_p(n^{-1/2})$ for any $j_1,j_2$, and $a_1$, as well as Assumptions \ref{ass:uni_bound_}, and \ref{ass:margin2} with $r\ge 1$. Hence $T_5=O_p\{(1/\sqrt{n})^{r+1}\}$. Similarily, 
		\begin{align*}
			T_6
			&=E\bigg\{\sum_{j_2\in\mathcal{J}_2}\sum_{j_2^\prime\neq J_2}\{\Delta^{2c}(\bS_{\bar{L}_2},J_1,\hat{\pi}^{1t},j_2)-\bS_{\bar{L}_2}^\top\bbeta_{J_1\hat{\pi}^{1t}j_2}^*\\
			&\qquad\qquad\qquad\qquad\qquad-\Delta^{2c}(\bS_{\bar{L}_2},J_1,\hat{\pi}^{1t},j_2^\prime)+\bS_{\bar{L}_2}^\top\bbeta_{J_1\hat{\pi}^{1t}j_2^\prime}^*\}I(\mathcal{F}_{J_1j_2}\cap \mathcal{E}_{J_1j_2^\prime})\bigg\}\\
			&=O\bigg[\epsilon_n\cdot\sum_{j_2\in\mathcal{J}_2}\sum_{j_2^\prime\neq j_2} P\{|\bS_{\bar{L}_2}^\top(\bbeta_{J_1\hat{\pi}^{1t}j_2^\prime}^*-\bbeta_{J_1\hat{\pi}^{1t}j_2}^*)|\\
			&\qquad\qquad\qquad\qquad\qquad\le |\bS_{\bar{L}_2}^\top(\bbeta^*_{J_1\hat{\pi}^{1t}j_2^\prime}-\widehat\bbeta_{J_1\hat{\pi}^{1t}j_2^\prime})-\bS_{\bar{L}_2}^\top(\bbeta^*_{J_1\hat{\pi}^{1t}j_2}-\widehat\bbeta_{J_1\hat{\pi}^{1t}j_2})| \}\bigg]\\
			&=O\{\epsilon_n(1/\sqrt{n})^r\}.
		\end{align*}
		Consequently, $\hat{G}^{2c}=O_p\{(1/\sqrt{n})^{r+1}+\epsilon_n(1/\sqrt{n})^r+\epsilon_n^{r+1}\}$. Repeating the argument used to control the term $\hat{G}^{2t}$, we have
		\begin{align*}
			\hat{G}^{1t}&\le E\left[\Delta^{1t}(\bS_{\bar{J}_1},J_1)\{I(\Delta^{1t}(\bS_{\bar{J}_1},J_1)>0)-I(\bS_{\bar{J}_1}^\top\widehat{\bgamma}_{J_1}>0)\}\right]\\
			&=O_p\{(1/\sqrt{n})^{r+1}+\epsilon_n(1/\sqrt{n})^r+\epsilon_n^{r+1}\}.
		\end{align*}
		Moreover, we denote
		\begin{align*}
			\mathcal{G}_{j_1}&=\{\Delta^{1c}(\bS_{l_1},j_1)\ge \Delta^{1c}(\bS_{l_1},j_1^\prime),\Delta^{1c}(\bS_{l_1},j_1)> \Delta^{1c}(\bS_{l_1},j_1^{\prime\prime}),j_1^\prime<j_1,j_1^{\prime\prime}>j_1\},\\
			\mathcal{H}_{j_1}&=\{\bS_{l_1}^\top\widehat{\bdelta}_{j_1}\ge \bS_{l_1}^\top\widehat{\bdelta}_{j_1^\prime},\bS_{l_1}^\top\widehat{\bdelta}_{j_1}> \bS_{l_1}^\top\widehat{\bdelta}_{j_1^{\prime\prime}},j_1^\prime<j_1,j_1^{\prime\prime}>j_1\}.
		\end{align*}
		Repeating the argument used to control term $\hat{G}^{2c}$, we obtain
		\begin{align*}
			\hat{G}^{1c}&=E\bigg[\sum_{j_1\in\mathcal{J}_1}\Delta^{1c}(\bS_{l_1},j_1)\{I(\mathcal{G}_{j_1})-I(\mathcal{H}_{j_1})\}\bigg]=O_p\{(1/\sqrt{n})^{r+1}+\epsilon_n(1/\sqrt{n})^r+\epsilon_n^{r+1}\}.
		\end{align*}
		Hence,
		\begin{align*}
		{\rm Profit}(\check\pi^{1c},\check\pi^{1t},\check\pi^{2c},\check\pi^{2t})-{\rm Profit}(\hat{\pi}^{1c},\hat{\pi}^{1t},\hat{\pi}^{2c},\hat{\pi}^{2t})&=O_p\{(1/\sqrt{n})^{r+1}+\epsilon_n(1/\sqrt{n})^r+\epsilon_n^{r+1}\}\\
		&=O_p\{(1/\sqrt{n})^{r+1}+\epsilon_n^{r+1}\}.
		\end{align*}
	\end{proof}

\end{document}